\documentclass[sigconf]{acmart}
\setcopyright{none}
\renewcommand\footnotetextcopyrightpermission[1]{}

\usepackage{nicefrac}
\usepackage{siunitx}
\usepackage{array,framed}
\usepackage{booktabs}
\usepackage{
  color,
  float,
  epsfig,
  wrapfig,
  graphics,
  graphicx
}
\usepackage{textcomp}
\usepackage{setspace}
\usepackage{amsmath}
\usepackage{latexsym,fancyhdr}
\usepackage{url}

\usepackage{breakurl}
\usepackage{enumerate}
\usepackage{graphics}
\usepackage{xparse} 
\usepackage{xspace}
\usepackage{multirow}
\usepackage{csvsimple}
\usepackage{balance}
\usepackage{booktabs}
\usepackage{subfigure}

\usepackage[page,header]{appendix}
\usepackage{titletoc}
\makeatletter
  \def\ttl@Hy@steplink#1{%
    \Hy@MakeCurrentHrefAuto{#1*}%
    \edef\ttl@Hy@saveanchor{%
      \noexpand\Hy@raisedlink{%
        \noexpand\hyper@anchorstart{\@currentHref}%
        \noexpand\hyper@anchorend
        \def\noexpand\ttl@Hy@SavedCurrentHref{\@currentHref}%
        \noexpand\ttl@Hy@PatchSaveWrite
      }%
    }%
  }%
  \def\ttl@Hy@PatchSaveWrite{%
    \begingroup
      \toks@\expandafter{\ttl@savewrite}%
      \edef\x{\endgroup
        \def\noexpand\ttl@savewrite{%
          \let\noexpand\@currentHref
              \noexpand\ttl@Hy@SavedCurrentHref
          \the\toks@
        }%
      }%
    \x
  }%
  \def\ttl@Hy@refstepcounter#1{%
    \let\ttl@b\Hy@raisedlink
    \def\Hy@raisedlink##1{%
      \def\ttl@Hy@saveanchor{\Hy@raisedlink{##1}}%
    }%
    \refstepcounter{#1}%
    \let\Hy@raisedlink\ttl@b
  }%
\def\ttl@gobblecontents#1#2#3#4{\ignorespaces}%
\makeatother

\definecolor{pos}{rgb}{0,0.4,0.8}
\definecolor{neg}{rgb}{1.0,0.3,0.5}
\newcommand{\customspace}{.2cm}
\newcommand{\customsize}{\tiny}

\usepackage{algorithm,algpseudocode}
\algnewcommand{\Inputs}[1]{%
  \State \textbf{Inputs: }
  \Statex \hspace*{\algorithmicindent}\parbox[t]{\linewidth}{\raggedright #1}
}
\algnewcommand{\Output}[1]{%
  \State \textbf{Output: }
  \Statex \hspace*{\algorithmicindent}\parbox[t]{\linewidth}{\raggedright #1}
}
\algnewcommand{\Initialize}[1]{%
  \State \textbf{Initialize:}
  \Statex \hspace*{\algorithmicindent}\parbox[t]{\linewidth}{\raggedright #1}
}
\algnewcommand{\comm}[1]{ {\small\ttfamily\textcolor{blue}{// #1}} }

\usepackage{graphicx}
\graphicspath{{./images/}} 

\usepackage{hyperref}

\usepackage{pigemoji}

\usepackage{macros}

\begin{document}
\title{QuerySnout: Automating the Discovery of Attribute Inference Attacks against Query-Based Systems}

\author{Ana-Maria Cre\c{t}u}
\authornote{Equal contribution.}
\affiliation{
    \institution{Imperial College London}
    \city{London}
    \country{United Kingdom}
}
\email{a.cretu@imperial.ac.uk}

\author{Florimond Houssiau}
\authornotemark[1]
\affiliation{
    \institution{The Alan Turing Institute} 
    \city{London}
    \country{United Kingdom}
}
\email{fhoussiau@turing.ac.uk}

\author{Antoine Cully}
\affiliation{
    \institution{Imperial College London}
    \city{London}
    \country{United Kingdom}
}
\email{a.cully@imperial.ac.uk}

\author{Yves-Alexandre de Montjoye}
\affiliation{
    \institution{Imperial College London}
    London
    \country{United Kingdom}
}
\email{deMontjoye@imperial.ac.uk}

\date{}

\renewcommand{\shortauthors}{Ana-Maria Creţu, Florimond Houssiau, Antoine Cully, \& Yves-Alexandre de Montjoye} 

\begin{abstract}
Although query-based systems (QBS) have become one of the main solutions to share data anonymously, building QBSes that robustly protect the privacy of individuals contributing to the dataset is a hard problem. 
Theoretical solutions relying on differential privacy guarantees are difficult to implement correctly with reasonable accuracy, while ad-hoc solutions might contain unknown vulnerabilities.
Evaluating the privacy provided by QBSes must thus be done by evaluating the accuracy of a wide range of privacy attacks.
However, existing attacks against QBSes require time and expertise to develop, need to be manually tailored to the specific systems attacked, and are limited in scope.
In this paper, we develop QuerySnout, the first method to automatically discover vulnerabilities in query-based systems.
QuerySnout takes as input a target record and the QBS as a black box, analyzes its behavior on one or more datasets, and outputs a multiset of queries together with a rule to combine answers to them in order to reveal the sensitive attribute of the target record.
QuerySnout uses evolutionary search techniques based on a novel mutation operator to find a multiset of queries susceptible to lead to an attack, and a machine learning classifier to infer the sensitive attribute from answers to the queries selected.
We showcase the versatility of QuerySnout by applying it to two attack scenarios (assuming access to either the private dataset or to a different dataset from the same distribution), three real-world datasets, and a variety of protection mechanisms.
We show the attacks found by QuerySnout to consistently equate or outperform, sometimes by a large margin, the best attacks from the literature. 
We finally show how QuerySnout can be extended to QBSes that require a budget, and apply QuerySnout to a simple QBS based on the Laplace mechanism.
Taken together, our results show how powerful and accurate attacks against QBSes can already be found by an automated system, allowing for highly complex QBSes to be automatically tested ``at the pressing of a button''. We believe this line of research to be crucial to improve the robustness of systems providing privacy-preserving access to personal data in theory and in practice\footnote{Our code is available at \url{https://github.com/computationalprivacy/querysnout}.}\footnote{This is an extended version of our paper  published in the ACM CCS 2022 conference that includes the Appendix.}.
\end{abstract}

\maketitle
\pagestyle{plain}

\section{Introduction}
\label{sec:introduction}

Our ability to collect and store data has exploded in the last decade. Coupled with the development of AI and new computational tools, this data has the potential to drive scientific advancements in healthcare~\cite{topol2019high} and the social sciences~\cite{lazer2009social}, and promises to revolutionize the way businesses and governments function. 

However, most of this data is either personal or linked to individuals in one way or another. This raises serious privacy concerns, and as such this data falls under the scope of data protection laws such as the European Union's General Data Protection Regulation~\cite{eu-gdpr,greenleaf2021global}. Finding solutions to use data for good while preserving our fundamental right to privacy is a timely and crucial question.

Query-based systems (QBS), controlled interfaces through which analysts can query the data, have the potential to enable privacy-preserving anonymous data analysis at scale. As the curator keeps control over the data, they can audit queries sent by analysts and ensure that the answers returned do not reveal individual-level information. Typical queries include histograms, counts, correlations between attributes, and other aggregates over individual records. QBS interfaces can range from an online interface and API~\cite{okeefe2008table,curtis2018openprescribing}, to languages such as SQL~\cite{mcsherry2009privacy,francis2017diffix} and the submission of scripts~\cite{oehmichen2019opal,opensafely}.

It has however long been known that only releasing aggregate information is not sufficient to protect privacy.
As early as 1979, Denning et al.~\cite{denning1979tracker} theorized difference attacks (``trackers'') against databases accessible through user-specified counts.
Researchers have shown that releasing marginals or simple counts can reveal the presence of a target individual in the private dataset~\cite{homer2008resolving,dwork2015robust}.
Famously, Dinur and Nissim proved in 2003 that a database can be reconstructed with good accuracy from a large number of aggregate statistics~\cite{dinur2003revealing}.
This issue is particularly acute for QBSes, where attackers can design queries to infer information about specific people, including by exploiting vulnerabilities or implementation bugs of the system. 

In response to these risks, increasingly sophisticated defense mechanisms have been put in place. These include combining query set size restriction with noise addition mechanisms~\cite{francis2017diffix,okeefe2008table}, the use of unbounded static noise~\cite{francis2017diffix,dwork2006calibrating}, the introduction of limits on the number of queries (e.g., as a privacy budget in Differential Privacy~\cite{mcsherry2009privacy}), and even online evaluation~\cite{nabar2008survey} and rewriting of queries~\cite{uber_flex}. Computer security tools such as access control, logging queries, code verification, and AI-based anomaly detection mechanisms are typically then deployed on top of QBS-specific mechanisms.

While these measures have helped prevent and mitigate privacy risks~\cite{francis2017diffix}, the risk of unknown strong ``zero-day'' attacks has stalled the development and deployment of QBSes. Manually designing and implementing attacks against complex and expressive QBSes is a difficult and painstaking process. It typically requires careful analysis of the system by experts and can take months. Furthermore, existing attacks have only exploited a small subset of the syntax of modern expressive QBSes.
Limiting the risk of existence of strong unmitigated attacks is thus essential to unlock the potential of QBSes and make individual-level data available to researchers and companies while strongly preserving privacy in practice.

\textbf{Contributions.} We here propose QuerySnout, the first method to automatically discover privacy vulnerabilities in query-based systems. We frame the discovery of attacks against a QBS as a black-box optimization problem. Specifically, we formalize an attack as a multiset of queries jointly with a mathematical rule (e.g., a machine learning model) to combine answers to the queries in order to reveal a particular secret. 
Given a threat model and black-box access to a QBS, we optimize the attacks with respect to their performance at inferring the secret.

At a high level, QuerySnout analyzes the query answering behavior of the QBS for patterns that can be used to infer the secret. QuerySnout is fully automated, combining (1) evolutionary search techniques to discover the right set of queries to ask with (2) a machine learning model (``rule'') trained to infer the secret from query answers. To efficiently explore the search space of possible attacks, we design a novel mutation operator tailored to this problem.

We instantiate our approach on \textit{attribute inference attacks}, where the secret is the sensitive attribute of a target record. We study two attack scenarios. The first scenario (AUXILIARY) assumes that an attacker has access to an auxiliary dataset similar to the private dataset, e.g., drawn from the same distribution. This is a common assumption across the broader literature~\cite{pyrgelis2017knock,shokri2017membership}. The second scenario (EXACT-BUT-ONE) assumes that an attacker has access to the entire private dataset protected by the QBS except for the sensitive attribute. Attack models like this one, relying on a very strong attacker, are used in the literature to evaluate privacy protections in the ``worst-case'' scenario~\cite{jagielski2020auditing}. They can help uncover flaws in the system design or implementation and audit systems.

We use QuerySnout to attack two real-world query-based systems, Diffix~\cite{francis2017diffix} and TableBuilder~\cite{okeefe2008table}, and a generic query-based system. We show our attack to be highly successful against all three systems when protecting three real-world datasets, matching or outperforming previous expert-designed attacks~\cite{chipperfield2016,rinott2018, gadotti2019signal}.
A post-hoc analysis of the attacks found by QuerySnout in the AUXILIARY scenario suggests that they exploit the same vulnerability as the manual attacks, but more effectively.
We propose a heuristic to extend our method to budget-based QBSes (such as those guaranteeing differential privacy~\cite{dwork2006calibrating}), and use QuerySnout to attack a simple QBS based on the Laplace mechanism.
We show that our attack achieves near-optimal accuracy for $\varepsilon \in \{5,10\}$.

Finally, we discuss how our method can be extended to other attack models -- e.g., to perform membership inference attacks -- and larger query syntaxes.
Our results suggest that QuerySnout can be used to evaluate the privacy protection offered by QBSes and help improve their design by identifying vulnerabilities.

\section{Background}
\label{sec:background}

We consider a \textit{data curator} entity, such as a business or a government agency, holding data about a set of users denoted by $U \subset \mathcal{U}$, with $\mathcal{U}$ a population of users.
The data consists of values for an ordered set of attributes $\attrset = (a_1, \dots, a_{n})$, e.g., age and nationality.
Each attribute $a \in \attrset$ has a set $\valueset_a$ of acceptable values.
We denote by \textit{individual record} $r^u$ the data available about a user $u$, consisting of the corresponding values for attributes in $\attrset$: $r^u \in \mathcal{V}_{a_1} \times \dots \times \mathcal{V}_{a_{n}}$.
We denote by $r^u_{\attrset'}$ the restriction of a record to a subset of attributes $\attrset' \subset \attrset$ and by $r^u_a$ the restriction to one attribute $a \in \attrset$.
For instance, if the attributes collected are $\attrset = (age, nationality)$, a person's record could be (19, Bulgaria). Given a user set $U$, we call \textit{dataset} the multiset of individual records $D$ for users in $U$. This means that the same record can appear more than once if different users in $U$ have the same values for these attributes.

We assume that the data curator allows data analysts to access information about the dataset $D$ through a query-based system (QBS). The QBS allows the data analyst to retrieve answers to queries about the dataset without directly accessing the individual records~\cite{de2018privacy}. The interface can uses a query language such as SQL or a GUI to define the semantics of queries that can be asked. Formally, we denote by \textit{query space} $\mathcal{Q}$ the set of queries that can be asked to the system.
We denote by $T:\mathcal{D}\times\mathcal{Q}\rightarrow\mathbb{R}$ the function that gives the \textit{true answer} to a query over a dataset.

Formally, we consider a query-based system to be a randomized function $R: \mathcal{D} \times \mathcal{Q} \rightarrow \mathbb{R}$ assigning to a dataset and query pair $(D, q) \in \mathcal{D} \times \mathcal{Q}$ a real-valued random variable $R(D, q)$. The answers provided by the QBS are usually designed to be \textit{similar} to the true answer ($R(D,q) \approx T(D,q)$) but not equal, in order to protect user privacy. We use the more general notion of a random variable because to preserve the privacy of users, query-based systems commonly implement mechanisms for post-processing that involve randomization (e.g. noise addition). When the answer is fixed, $R(D, q)$ is a deterministic random variable.

Multiple building blocks are typically combined to preserve privacy. These include noise addition~\cite{denning1980secure} with a range of distributions (e.g. Laplace~\cite{dwork2006calibrating},  Gaussian~\cite{francis2017diffix}, uniform~\cite{fraser2005proposed}), suppression of answers below a threshold (called \textit{query set size restriction})~\cite{denning1979tracker}, and restrictions on the set of allowed queries $\mathcal{Q} \subset \mathcal{Q}'$ or the number of queries that can be meaningfully answered\footnote{Although this mechanism cannot be expressed formally as $R(D,Q)$, since it requires memory of the number of queries performed. We use this notation for simplicity.}~\cite{dwork2006calibrating}.

\section{Attack model}
\label{sec:attack_model}

We propose a general targeted attack against query-based systems, which we call \textit{automated query discovery attack}. The attack is composed of a search mechanism for a multiset of queries and a rule to combine them. While we focus on attribute inference attacks, our approach can be extended to other attacks, e.g., membership inference attacks.

\subsection{Attacker access to the query-based system}

We assume that the attacker has access to the query-based system protecting the dataset of interest (\textit{target QBS}), in the sense that they can send queries to it. We furthermore assume that the number of queries that can be performed on the target QBS is limited, e.g. to a few hundreds, as queries are typically logged and rate limited.

We also assume that the attacker has black-box access to the QBS software. The QBS software would typically be available freely or for a fee, potentially in compiled form. Alternatively, the attacker might replicate the protection deployed by the target QBS based on public information.
Formally, this means that the attacker can retrieve samples from $R(D', q)$ for any query $q \in \mathcal{Q}$ and dataset $D' \in \mathcal{D}$. Since the attacker chooses $D'$, they also know the true query answer $T(D', q)$ and can leverage this signal to devise attacks.
Note that QBSes are typically initiated with a \textit{seed} for pseudo-random noise generation: we assume that each QBS is initiated with a different seed, and that the attacker does not know the seed used in the target QBS.

\subsection{Attribute inference attack}
\label{subsec:automated-query-discovery-attack}

The attacker's goal is to infer the target user's value for one of the attributes (the \textit{sensitive} attribute) $s = r^u_{a^*},~a^* \in \attrset \setminus \attrset'$. We assume that the attacker knows part of the record of a \textit{target user} $u \in U$ consisting of the values for a subset of attributes $\attrset' \subset \attrset$. For simplicity, we also assume that the attacker knows (1) that the target is in the dataset ($u \in U$), and (2) that the target's record is unique in the dataset, given all known attributes ($\forall v \in U, v \neq u:~r^u_{\attrset'} \neq r^v_{\attrset'}$). 

Formally, the attacker's goal is to devise both a multiset of queries $\left(q_1, \ldots, q_k\right) \in \mathcal{Q}^k$ and a rule $G:\mathbb{R}^k\rightarrow\{0,1\}$ to combine the answers to the queries $R(D, q_1), \ldots, R(D, q_k)$ to retrieve the sensitive attribute:
$$\hat{s} = G(R(D, q_1), \ldots, R(D, q_k)).$$ The attack is considered successful if the predicted value $\hat{s}$ for the sensitive attribute matches the correct value, i.e., $\hat{s} = s$.

We focus on users that are unique in the dataset $D$ to simplify the evaluation of the privacy gain provided by the QBS alone.
An attacker would indeed be able to perfectly infer the sensitive attribute $a^*$ of unique users if the dataset were to be released as-is, allowing us to evaluate the privacy gain provided by the QBS. \footnote{Note that technically, this property is true for a larger class of users, who are so-called \textit{value-unique}~\cite{gadotti2019signal}: all records who share the same known attributes have the same sensitive attribute.
We focus on unique users for simplicity (who are all value-unique), although the method in this paper applies similarly to value-unique users.}

\subsection{Auxiliary knowledge on the target dataset}
\label{subsec:auxiliary-knowledge-dataset}

We consider two specific attack scenarios, with different assumptions on the attacker's knowledge about the private dataset:
\begin{itemize}
    \item \textbf{AUXILIARY}: the attacker has access to a dataset $D'$ similar to the private dataset $D$ (e.g., from the same distribution).
    \item \textbf{EXACT-BUT-ONE}: the attacker has perfect knowledge of the private dataset $D$, except for the target record's sensitive attribute $r^u_{a_n}$. 
\end{itemize}

The latter (EXACT-BUT-ONE) is the typical assumption made to evaluate  ``worst-case'' attacks~\cite{balle2022reconstructing,jagielski2020auditing}, evaluating how much information a very strong attacker would be able to infer. The former (AUXILIARY), on the other hand, is typically used to evaluate privacy risks in practice, i.e., as a more realistic setup~\cite{shokri2017membership,pyrgelis2017knock,stadler2020synthetic}. 

The attacker's knowledge allows them to define a training distribution $\traindist$ and a validation distribution $\valdist$ (that can be identical), from which to sample datasets. Datasets generated from the former are used to train the rule $G$, while datasets from the latter are used to estimate the fitness (test accuracy) of a multiset of queries.

In the AUXILIARY scenario, the attacker generates auxiliary datasets by sampling records uniformly at random without replacement from $D'$ and appending the target record $r^u_{\attrset'}$ with a random value for the sensitive attribute $r^u_{a_n}$, which defines $\attdist$. In practice, the attacker divides $D'$ in a \textit{training} and a \textit{validation} dataset, which define $\traindist$ and $\valdist$.

In the EXACT-BUT-ONE scenario, auxiliary datasets are obtained by choosing a random value for the target user's sensitive attribute, which defines $\attdist$. In this setup, the training and validation distributions are identical.

Note that in both scenarios, we randomize the target user's sensitive attribute. This breaks possible correlations with known attributes, which could be used to infer the value of the sensitive attribute even if the target user does not contribute their data. Randomization implies that the baseline success rate of an attack without access to data is 50\% when the sensitive attribute is binary.

\section{Attack methodology}
\label{sec:attack_methodology}
\subsection{Overview of QuerySnout}

In this section, we present QuerySnout, our method for automating the discovery of attribute inference attacks against query-based systems. The goal of QuerySnout is to find both a multiset of $m$ queries $q_1, \ldots, q_m \in \mathcal{Q}$ and a rule $G$ to combine the answers to these queries to retrieve $r^u_{a_n}$, the target record's value for the sensitive attribute $a_n$. We here assume, without loss of generality, that the last attribute $a^* = a_n$ is sensitive and that the remaining attributes are known auxiliary information about the target $\attrset'=(a_1, \ldots, a_{n-1})$. We frame the discovery of attacks against a QBS as an optimization problem over a search space of solutions consisting of all multisets of $m$ queries of a specific structure (Sec.~\ref{subsec:attack-search-space}).
QuerySnout optimizes solutions in this space with regards to the accuracy of the attribute inference attack obtained by applying an automatically trained rule $G$ on the answers to this multiset of queries (Sec.~\ref{subsec:automated-learning-rule}).
To optimize for attacks, QuerySnout uses an evolutionary algorithm based on a novel mutation algorithm (Sec.~\ref{subsec:search-space-exploration}). 

\subsection{Search space of solutions}\label{subsec:attack-search-space}

\begin{figure*}[htbp!]
\centering
\includegraphics[width=\linewidth]{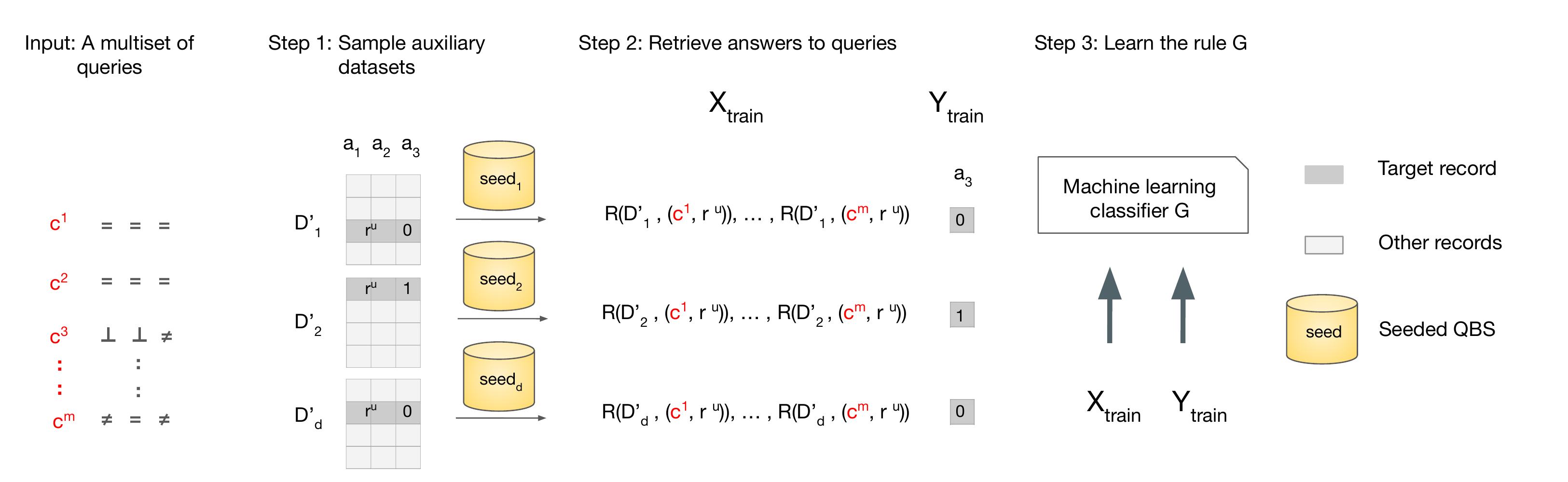}
\caption{\textbf{Automated learning of a rule $G$.} We illustrate our approach to automatically combine the answers to a multiset of queries, via a rule $G$ learned using a machine learning model. In our example, partial records $r^u$ consist of values for attributes $a_1$ and $a_2$, while $a_3$ is the sensitive attribute. The attacker samples $d$ auxiliary datasets (randomizing the value of $a_3$ for the target record), then retrieves $d\times m$ answers to the queries (in red) from (differently seeded) QBSes protecting the datasets. The answers together with the target's sensitive attributes from the corresponding datasets are used to train a classifier $G$.
}
\label{fig:learning-the-rule}
\end{figure*}

Given a restricted \textit{query search space} $\mathcal{Q}_s \subset \mathcal{Q}$, we define a \textit{solution} as an unordered list of $m$ queries $q_1 \ldots, q_m \in \mathcal{Q}_s$.
Formally, solutions are multisets of queries, meaning that the same queries can be repeated multiple times.
We denote by $\mathcal{S}^m$ the set of all \textit{solutions} that is explored by the evolutionary algorithm. We describe it formally as follows:
\begin{gather}
\begin{aligned}
\mathcal{S}^m := &\Big\{
    \{(q^1, m_1), \ldots, (q^k, m_k) \}:q^1,\dots, q^k \in \mathcal{Q}_s^n,\\ 
    & i \neq j \implies q_i \neq q_j, \forall i, m_i \in \mathbb{N}, ~\sum_{i=1}^k m_i = m \Big\}
\end{aligned}
\end{gather}
where $m_i$ denotes the multiplicity of the $i$-th unique query.

Note that we here assume that the order of the queries does not impact the way the QBS answers them. All the systems we consider satisfy this property. In the discussion, we explain how our method could be adapted to systems for which query order matters.

\textbf{Query search space.} For general query syntaxes (e.g., SQL), the query space $\mathcal{Q}$ is extremely large, even possibly unbounded.
In this work, we restrict -- as a starting point -- the search to a simple, yet still very large \textit{query space} $\mathcal{Q}_s \subset \mathcal{Q}$ consisting of counting queries that select records via a conjunction of up to $n$ simple conditions, with no more than one condition per attribute. 

Given a set of condition operators $\conditionset_s$ (such as ``equal to'' or ``different from''), we write $a_i \; c_i \; v_i$ to denote a condition on the $i$-th attribute of operator $c_i \in \conditionset_s$, with value $v_i \in \valueset_{a_i}$ belonging to set of acceptable values for the $i$-th attribute. To simplify the search and exploit the information available to the attacker about the target record, we restrict conditions relating to a known attribute to use the target record's value $v_i = r^u_{a_i}, i \in \attrset'$.
A condition on the unknown, sensitive attribute may use any of the acceptable values $v_n \in \valueset_{a_n}$. Formally, the query search space we consider, $Q_s$, consists of queries of the form:
\begin{gather}
\begin{aligned}
    \text{SELECT} & \; \text{COUNT(*)} \\
    \text{WHERE} & \; a_{1} \; c_1 \; r^u_{a_1} \; \text{AND} \; \ldots \;  \text{AND} \; a_{n-1} \; c_{n-1} \; r^u_{n-1} \\ \; & \text{AND} \; a_n \; c_n \; v_n 
\end{aligned}
\end{gather}

In this work, we set the set of condition operators to $\conditionset_s = \{ \neq, =, \perp \}$, meaning that conditions can be of operator $=$ (equal to), $\neq$ (different from) or $\perp$ indicating that there is no condition on the $i$-th attribute. In other words, a condition relating to the $i$-th attribute can be: (1) $a_i = v_i$, (2)  $a_i \neq v_i$ or (3) no condition, which we write as $a_i \perp v_{i}$ as a convention (the value is ignored); with $v_i \in \mathcal{V}_{a_i}$.

Finally, we assume -- for simplicity -- the sensitive attribute to be binary: $\valueset_{a_n} = \{ 0, 1\}$. Under this assumption, the conditions $a_n = 0$ and $a_n \neq 1$ are equivalent, and there is a one-to-one mapping between the query search space $\mathcal{Q}_s$ and $\conditionset_s^{~n}$, as one choice of condition operators $(c_1, \ldots, c_n) \in \conditionset_s^n$ corresponds to exactly one query with operator (2) and conditions $a_1 \; c_1 \; r^u_{a_1}, \,\, a_2 \; c_2 \; r^u_{a_2}, \,\, \ldots, \,\, a_{n-1} \; c_{n-1} \; r^u_{a_{n-1}}$ and $a_n \; c_n \; 0$, respectively.
The search algorithm can thus represent queries as strings of $k$ operators.

As an example, let attributes $a_1$, $a_2$ and $a_3$ denote a person's $age$, $nationality$ and $diagnosis$ for a disease. The query with condition operators $(c_1, c_2, c_3) = (=, \perp, \neq)$ applied to a target record of known age and nationality $(v_1, v_2, v_3)= (19, Bulgaria, 0)$ yields the query: ``How many people in the database have an age of 19 and a positive diagnosis?'' Note that there is no condition on the $nationality$ attribute and that the $\neq$ operator on the sensitive attribute is equivalent to selecting users with a positive diagnosis. 

We restrict the query space for two main reasons. 
First, a large number of attacks from previous work can be performed using only queries expressed this way, and we can thus compare our results with manual attacks.
For instance, averaging attacks~\cite{denning1980secure}, difference attacks~\cite{denning1979tracker} and the differential noise-exploitation attack against Diffix~\cite{gadotti2019signal} can all be written with queries in $\mathcal{Q}_s$.
Second, the problem is already computationally 
challenging, as the size of the attack search space is very large. Indeed, when the sensitive attribute is binary, the cardinality of the query space is $|\mathcal{Q}_s| = |\conditionset^n|=3^n$, which increases exponentially with the number of known attributes $n$. The size of the attack search space $\mathcal{S}^m$ is therefore equal to the number of multisubsets with $m$ elements\footnote{Note that the size of $\mathcal{S}^m$ is \textit{not} $\left|\mathcal{Q}_s\right|^m$, as the order does not matter.} from a set of size $3^n$~\cite{feller1968probability}:
$\binom{3^n+m-1}{m} \approx \frac{3^{n (m-1)}}{m!}$,
where the approximation holds if $m \ll 3^n$. For typical values such as those used in this paper, e.g., $m=100$ and $n=6$, an exact computation yields an extremely large attack search space size of $\approx 1.33 \times 10^{131}$. This emphasizes the importance of being able to search for solutions efficiently.

\subsection{Automated learning of a rule $G$}
\label{subsec:automated-learning-rule}

\begin{figure}[htbp!]
\centering
\includegraphics[width=\linewidth]{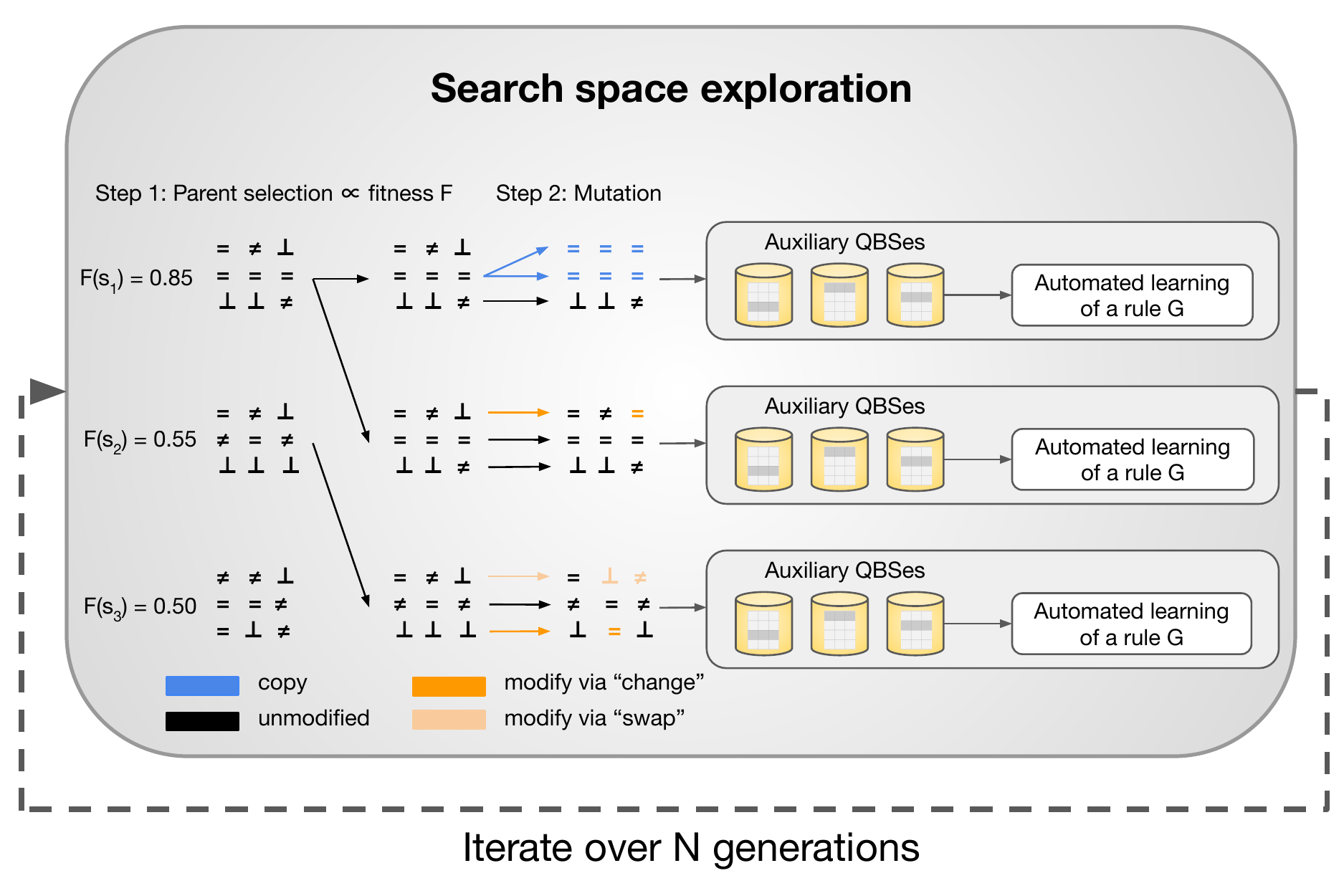}
\caption{\textbf{Illustration of how QuerySnout explores the search space using evolutionary algorithms.} In this example, the search space consists of multisets of $m=3$ queries relating to $n=3$ attributes. Solutions with higher fitness are more likely to be selected. We use different colors to illustrate the changes to a query and black for unmodified query operators.}
\label{fig:search-space-exploration}
\end{figure}

In order to evaluate whether a solution can be used to perform an attack, we propose a method to automatically learn a rule $G$ to combine the answers to the queries of a solution $\{q_1, \ldots, q_m\} \in \mathcal{S}^m$ to perform an attribute inference attack (Fig.~\ref{fig:learning-the-rule}).
Our method uses the training distribution $\traindist$ to sample $d$ auxiliary datasets $D'_1, \ldots, D'_d \sim \traindist$. Each dataset is protected by a QBS initialized with different seeds for their random number generator. We perform the queries on each QBS to obtain the answers $(R(D'_i, q_1), \ldots, \\ R(D'_i, q_m))$, $i=1, \ldots, d$. We denote by $X_{\text{train}}$ the dataset of $d$ samples of $m$ features each (a feature is a query answer) obtained in this way. We also denote by $Y_{\text{train}}$ the $d$ values for the target record's sensitive attribute in the corresponding auxiliary datasets. Finally, we train a binary classification model $G$ on $(X_{\text{train}}, y_{\text{train}})$ to infer the value of the sensitive attribute from the answer to queries. The model will be used as the rule $G$ to combine the answers to queries in the solution.

Our approach assumes that an optimal rule $G$ can be effectively approximated using machine learning models. For instance, the averaging attack described in Sec.~\ref{subsec:manual-attacks} computes a linear combination of the query answers, which can be represented by a logistic regression. Furthermore, multilayer perceptrons are known to be universal approximators~\cite{hornik1989multilayer}, which makes them even more apt to model an arbitrary rule. Our approach to learn the rule $G$ automatically given a set of queries is similar to the one proposed by Pyrgelis et al. to perform membership inference attacks on aggregate statistics~\cite{pyrgelis2017knock}, except we apply it to attribute inference attacks and use a broader and more flexible range of queries.

\subsection{Search space exploration}
\label{subsec:search-space-exploration}

We explore the search space for queries that can be combined to perform highly accurate attacks using evolutionary algorithms. Evolutionary algorithms are a family of optimization algorithms particularly well suited when the structure of the problem is discrete, unknown or too complex to be described mathematically, and the search space is very large~\cite{mitchell1998introduction}.
We refer the reader to the Appendix for a brief introduction to the topic.
Our approach (illustrated in Fig.~\ref{fig:search-space-exploration}) improves a \textit{population} of solutions over time by applying small random changes, called \textit{mutations}, to the solutions. We propose a mutation operator tailored to the task that allows our procedure to efficiently explore the search space. The solutions are optimized with regards to their \textit{fitness}: their ability to infer the target record's sensitive attribute.

Algorithm~\ref{alg:evolutionary-attack-search} details our procedure to explore the space of attacks using evolutionary algorithms. We start from a population of $P$ solutions $s_1, \ldots, s_P$. Each solution consists of $m$ queries $s_j = (c^{j1}, \dots, c^{jm})$, which are initialized uniformly at random: $c^{ji} \gets \mathcal{U}(\conditionset^n), i = 1,\dots,m$.
The algorithm runs for $N$ generations.
In each generation, we first evaluate the fitness of each of the $P$ solutions. The fitness function, which we describe in the next paragraph, estimates the accuracy of the attribute inference attack obtained with the queries in the solution.
Second, we sort the population decreasingly according to the fitness of each solution.
Third, we create a new population consisting of (1) a number $P_e$ of \textit{elites}~\cite{dejong1975analysis}, i.e., the $P_e$ solutions in the current population with the highest fitness, and (2) a number $P-P_e$ of \textit{offsprings}.
Each offspring is generated by applying a mutation to a \textit{parent} solution.

\begin{algorithm}[t]
\caption{\textsc{QuerySnoutEvolutionarySearch}}
\label{alg:evolutionary-attack-search}
    \begin{algorithmic}[1]
        \Inputs{
        $P$: Population size (number of solutions). \\
        $m$: Number of queries in a solution. \\
        $P_e$: Number of elites. \\
        $N$: Number of generations. \\
        $D'_{\text{aux}}$: Auxiliary datasets used to evaluate the fitness.\\
        $p_\text{mut}$: Mutation parameters.
        }
        
        \Output{
        $population$: A population of $P$ solutions in the search \\ space of solutions $\mathcal{S}^m$.
        }
    
        \Initialize{
        $population \gets [random\_solution(m)$ for $i = 1$ to $P]$
        }
        \For{$g = 1$ to $N$} \State{$fitnesses \gets [\textsc{EvaluateFitness}(solution, D'_{\text{aux}})$ for $solution$ in $population]$}
        \State{\comm{Sort the population by fitness.}}
        \State{$sort\_descending(population, fitnesses)$}
        \State{\comm{Pass the elites unchanged to the next generation.}}
        \State{$new\_population \gets population[: P_e]$ }
        \For{$i=1$ to $P - P_e$}
        \State{$parent \gets select\_parent(population, fitnesses)$}
        \State{$offspring \gets \textsc{ApplyMutation}(parent, p_{\text{mut}})$}
        \State{$new\_population.append(offspring)$}
        \EndFor
        \State{$population \gets new\_population$}
        \EndFor
    \end{algorithmic}
\end{algorithm}

The parents are selected from the current population by sampling with replacement using the \textit{biased roulette wheel}~\cite{sastry2005genetic}: the probability to sample a solution is equal to its fitness divided by the sum of fitnesses of solutions in the population. This ensures that solutions with higher fitness are more likely to generate offsprings.

\textbf{Fitness evaluation.}
Algorithm~\ref{alg:evaluate-fitness} in the Appendix describes our procedure to evaluate the fitness. The procedure uses training and validation auxiliary datasets $D'_{\text{aux}} = (D'_{\text{train}}, D'_{\text{val}})$. The datasets in $D'_{\text{train}}$ and $D'_{\text{val}}$ are (1) sampled upon initialization from $\traindist$ and $\valdist$, respectively, then (2) protected by individual QBS instances with different seeds initializing their random number generator.
To evaluate the fitness, we train a rule $G$ to combine the answer of queries to predict the sensitive attribute of the target record, using auxiliary datasets $D'_\text{train}$, as described in Sec.~\ref{subsec:automated-learning-rule}.
We measure the accuracy of this prediction on training and validation data, $a_\text{train}$ and $a_\text{val}$ respectively.
We use $\min(a_\text{train}, a_{\text{val}})$ as fitness in order to minimize the effect of randomness in the evaluation. Indeed, the accuracy estimates are noisy, and we found empirically that using $a_\text{val}$ as fitness leads to the algorithm selecting solutions that were \textit{lucky} when estimating the fitness, rather than really superior. The effect was however minor, occurring mostly when the algorithm had converged.

\textbf{Mutation operator.}
The mutation operator aims to explore the local space around known good solutions.
The mutation operator (we refer the reader to Algorithm \ref{alg:apply-mutation} in Appendix~\ref{appendix:apply-mutation} for the pseudocode), takes as input a \textit{parent} solution and makes small changes to its queries in order to generate an \textit{offspring} solution.
Each query in the parent solution is treated separately, and is either copied with probability $p_{\text{copy}}$, modified in place with probability $p_{\text{modify}}$, or left unmodified with probability $1-p_{\text{copy}}-p_{\text{modify}}$.
Note that copying a query means that we first add the query -- as is -- to the offspring and then perform the following on a copy of it. If the system is deterministic, we modify the copy; otherwise, with equal probability we either keep the unmodified copy or replace it with a modified version.
We add the resulting query to the offspring.
We distinguish between the two cases because asking a query repeatedly provides no additional information when the QBS is deterministic.
Finally, as the offspring may now contain more than $m$ queries (due to copying), we select a random subset of $m$ queries.
This mutation operator is inspired by the fact that many attacks on QBSes~\cite{denning1980secure,gadotti2019signal} use pairs of similar queries, or repeat the same queries multiple times.

\textbf{Modifying a query.}
We now describe our procedure to modify a query consisting of $n$ operators $c=(c_1, \ldots, c_n) \in \conditionset^n$ (and refer the reader to Algorithm~\ref{alg:modify-query} in Appendix~\ref{appendix:apply-mutation} for the pseudocode).
For each attribute $a_i$, we either (1) ``change''  the corresponding operator $c_i$ with probability $p_{\text{change}}$ by replacing it with a different value in the operator set which we sample uniformly at random, (2) swap  the operators between the $i$-th attribute and another attribute that has not been swapped yet with probability $p_{\text{swap}}$, or (3) leave the operator unchanged.
To ensure that no attribute pair is more likely to be swapped compared to the others, we randomly permute the order in which this procedure considers the entries of $c$.
We introduce the ``swap'' operation to exploit symmetry between attributes: if a given attack is successful, it is likely that swapping the conditions of two attributes will lead to a similarly successful attack. Note that although ``swap'' and ``change'' can produce the same outcome, they do so with different probabilities and are thus not equivalent.

\textbf{Choosing the mutation parameters.} The mutation parameters are set such that, on average, we copy (modify) a fraction $p_{\text{copy}} (p_{\text{modify}})$ of the $m$ queries in a solution, a common heuristic in evolutionary search. The same heuristic also applies to parameters $p_{\text{swap}}$ and $p_{\text{change}}$. In practice, the ranges of mutation parameters should be informed by the exploration-exploitation trade-off: making more changes to a solution allows to explore the space more, but making too many changes at once might hinder the improvement of solutions over time.

\section{Experimental setup}\label{subsec:experimental_setup}

In this section, we first present the datasets we use to evaluate our automated attacks. Second, we present how we instantiate the auxiliary knowledge in the AUXILIARY and EXACT-BUT-ONE scenarios. Third, we describe the query-based systems against which we evaluate our automated attacks. Finally, we describe the system-specific manual attacks we compare against.

\subsection{Datasets} 
We use three publicly available datasets: Adult~\cite{adults1996}, Census~\cite{censusincome2000} and Insurance~\cite{insurance2000}. Adult contains 48482 individual records of 14 socio-demographic attributes each. Census contains 299285 records of 41 socio-demographic attributes each. Demographic attributes are, for instance, ``age'', ``education'' and ``occupation''. We assign ``income'' as the binary sensitive attribute for both Adult and Census and randomize it (50:50) as described below. Insurance contains 9822 records of 86 attributes about customers of a car insurance company, aggregated at the zipcode level. We only use the 43 socio-demographic attributes and add a randomized (50:50) sensitive attribute.

\subsection{Instantiating the auxiliary knowledge} 
For each dataset, in each repetition of the experiments, we sample $n-1 = 5$ attributes uniformly at random without replacement and discard the others. These are the attributes $\attrset'$ whose values for a target record we assume to be known to the attacker. 
We then randomly partition the dataset between a training auxiliary dataset $D_\text{train}$, a validation auxiliary dataset $D_\text{val}$, and a testing dataset $D_\text{test}$ of equal sizes. We select 100 unique \textit{target records} from $D_\text{test}$ uniformly at random without replacement. For each target record, we instantiate the training and validation distributions $\traindist$ and $\valdist$. We also instantiate a test distribution $\testdist$ from which \textit{private} datasets $D$ containing the target record will be sampled in two steps.
\begin{itemize}
    \item In the AUXILIARY scenario, for $\text{split} \in \{ \text{train}, \text{val}, \text{test}\}$, we instantiate $\pi_{\text{split}}$ on the corresponding dataset $D_{\text{split}}$ as described in Sec. \ref{subsec:auxiliary-knowledge-dataset}. We also add a randomized sensitive attribute (50:50) to datasets sampled from $\pi_{\text{split}}$, and remove duplicates of the target record from all datasets generated to ensure that the it is unique.
    \item In the EXACT-BUT-ONE scenario, we sample one dataset $D$ from $D_{\text{test}}$ and add a randomized sensitive attribute (50:50). We then instantiate $\traindist$, $\valdist$ and $\testdist$ on $D$ as described in Sec. \ref{subsec:auxiliary-knowledge-dataset}.
\end{itemize}
To ensure that information from the test distribution does not unintendedly leak into the train and validation distributions, the seeds used to instantiate the QBSes on datasets sampled from $\traindist, \valdist$ and $\testdist$ are \textit{all distinct}. To evaluate the actual privacy leakage of the query-based system and in line with previous work~\cite{gadotti2019signal}, we choose to break the correlations between the sensitive attribute and the other attributes by randomizing the binary sensitive attribute with a 50:50 distribution. 

\subsection{Query-based systems}\label{subsec:query-based-systems}
We describe our implementation of three query-based systems. The first two are based on real-world privacy protection mechanisms used by Diffix~\cite{francis2017diffix} and TableBuilder~\cite{fraser2005proposed}, while the third one is a non-deterministic system combining two QBS building blocks: query set size restriction and (Gaussian) noise addition. 

These systems all rely on \textit{query set size restriction} (QSSR), where the QBS refuses to answer a query if it concerns less than $T$ users, for some threshold $T$. QSSR aims at preventing queries that reveal information about a small number of users, such as counting the number of record identical to the target user and either value of $s$.
We here assume that a QBS doing bucket suppression returns $0$ instead of an error message, as this makes the attack harder by giving less information to the attacker. For simplicity, we introduce the following notation for bucket suppression:
$\beta_{\tau}(D,Q) = I\{T(D,Q) > \tau\}$, where $I$ is the indicator function. 
We also introduce the notion of \textit{query set} $U(D,Q)$, the set of all users who satisfy the conditions of the query $Q$ in the dataset $D$.
Note that for counting queries, $T(D,Q) = |U(D,Q)|$.

\textbf{Rounding.}
For realism, since we focus on counting queries, we further assume that the answers obtained from all mechanisms are then rounded to the nearest integer, and that if the mechanism would return a negative answer it returns $0$ instead.
This again makes the attack harder by giving less information to the attacker.
When applying our attack on a mechanism, we thus use $R_\text{qbs}$ defined as:
\[
R_\text{qbs}(Q,D) = \left\lfloor \max(R_\text{mechanism}(Q,D), 0)\right\rceil~
\]

\textbf{Randomness. } Both Diffix and TableBuilder rely on \textit{seeded} noise, i.e., noise produced by a pseudo-random number generator (PRNG) which outputs the same noise given the same input seed. Hence, these two systems are deterministic: for the same input query, they always output the same result. On the contrary, SimpleQBS is non-deterministic, and sample fresh random variables for each query.

\textbf{Diffix. } Diffix is a commercial QBS developed by the startup Aircloak~\cite{francis2017diffix}. It uses bucket suppression with a noisy threshold and two layers of so-called \textit{static} and \textit{dynamic} noise addition. All noises are seeded with different elements of the query, in order to prevent specific attacks.
Specifically, we attack Diffix-Birch~\cite{francis2018extended}, which adds noise to counting queries as:
\[
R_\text{diffix}(D,Q) = \beta_{\min(2,\tau)}(D,Q) \left(T(D,Q) +\sum_{i=1}^{n_\text{cond}(Q)} N^S_i + N^D_i \right)
\]
where $n_\text{cond}(Q)$ is the number of non-empty conditions in $Q$ (i.e., $\sum_{i=1}^n I\{o_i \neq \perp\}$), $\tau\sim\mathcal{N}(4,0.5)$ is seeded with the query set, $N^S_i \sim\mathcal{N}(0,1)$ is seeded with the text of condition $i$ of $Q$, and $N^D_i \sim\mathcal{N}(0,1)$ is seeded with the text of condition $i$ of $Q$ and the query set $U(D,Q)$.

\textbf{TableBuilder. } TableBuilder is a QBS developed by the Australian Bureau of Statistics for census contingency tables~\cite{fraser2005proposed}. We here focus on the privacy mechanism used on individual cells (similarly to~\cite{asghar2020averaging}), and hence attack a simplified subset of the syntax of the real system. We define this mechanism as
\[
R_\text{tablebuilder}(D,Q) = \beta_4(D,Q) \cdot (T(D,Q)+U),
\]
where $U \sim \mathcal{U}\{-2,\dots,2\}$ is seeded with the query set $U(D,Q)$.\\

\textbf{SimpleQBS. }
We study a simple QBS that combines two common building blocks: Gaussian noise addition (with variance $\sigma$) and query set size restriction. 
Formally, the mechanism SimpleQBS$(\tau,\sigma^2)$ answers queries as:
\[
R_\text{simple}(D,Q) = \beta_\tau(D,Q) \cdot \left(T(D,Q) + N \right),
\]
with $N\sim\mathcal{N}(0,\sigma^2)$. A fresh noise sample is drawn every time a query is performed, even if the query is repeated. We instantiate this QBS with all pairs $(\tau,\sigma)$, $\tau=0,\dots,4$ and $\sigma = 0,\dots,4$.

Note that SimpleQBS$(\tau,\sigma>0)$ can be seen as an application of the Gaussian mechanism with sensitivity $\max(1,\tau)$~\cite{dwork2014algorithmic}, and thus provides $(\varepsilon,\delta)-$differential privacy for some values of $\varepsilon, \delta$.
However, for our choice of parameters, and since we do not restrict the number of queries, the corresponding values of $\varepsilon,\delta$ are very large and do not represent meaningful privacy guarantees.
In Sec.~\ref{sec:budget-based-mechanisms}, we propose an extension of our method to handle budget-based mechanisms.

\subsection{Manual attribute inference attacks}\label{subsec:manual-attacks}

We here describe manual attribute inference attacks from prior work against Diffix, TableBuilder, and SimpleQBS. We compare QuerySnout with these attacks both quantitatively (i.e., the accuracy of the inference) and qualitatively (i.e., what is the attack exploiting).
We first introduce the notion of \textit{difference attack}, which underlies most of the attacks we consider here.

\textbf{Difference attacks. } A common class of attacks against QBSes is \textit{difference attacks}, which consist of a pair of queries $(q_1, q_2)$ of the form\vspace{-.2cm}
\[\vspace{-.1cm}
\begin{array}{l}
q_1 = \text{COUNT WHERE } \bigwedge_{i \in A'} (a_i = r^u_{a_i}) \land a_n = s\\
q_2 = q_1 \land a_{i'} \neq r^u_{a_{i'}}
\end{array}
\]
where $A' \subset \{1,\dots,n-1\}$ is a subset of attributes, $i' \not\in A', i' \neq n$ is another attribute, and $s \in \{0,1\}$ is a possible value for the sensitive attribute.
The insight behind such attacks is that if the target user is uniquely identified by $(A', i')$, then the true counts of $q_1$ and $q_2$ differ by $1$ if and only if $r^u_{a_n} = s$, and $0$ otherwise.
The rule to combine the answers to $q_1$ and $q_2$, as well as how to select the parameters $(A', i', s)$ depend on the system.

\textbf{Diffix. } Gadotti et al.~\cite{gadotti2019signal} proposed the only known attribute inference attack against Diffix.
This attack uses difference queries as described above, and exploits the structure of Diffix's noise to combine results.
Specifically, if the target record is unique for $A' \cup \{i'\}$, then the distribution of $R(q_1,D) - R(q_2,D)$ is either $\mathcal{N}(0,2)$ if $r^u_{a_n}=1$ or $\mathcal{N}(1, 2 |A'|+2)$ otherwise. The attack uses a likelihood ratio test to distinguish between the cases.
In order to find values of $(A',i',s)$ such that the target record is unique and the queries are not suppressed, the attack performs a search over subsets of attributes using access to the target QBS (protecting the private dataset $D$) and heuristics to determine whether the assumptions are verified.
We implement their attack with the \textsf{ValueUnique}~\cite{gadotti2019signal} heuristic in the AUXILIARY setup, and the exact uniqueness oracle in the EXACT-BUT-ONE setup to reflect the additional knowledge of the attacker.
Note that this attack is \textit{interactive}, as it uses iterative access to the target QBS to choose which queries to perform.

\textbf{TableBuilder. }
Chipperfield et al.\cite{chipperfield2016} propose a simple difference attack against TableBuilder.
For each known attribute $j$ and value $s\in\{0,1\}$, they perform the difference queries $(q_1^{j,s},q_2^{j,s})$ for $A' = \{1,\dots,n-1\} \setminus \{j\}$ and $i' = j$.
They then compute the difference $r^{j,s} = R(q_1^{j,s}, D) - R(q_2^{j,s}, D)$, and observe that if $r^{j,s} \geq 5$, the true difference must be 1 (because the noise is bounded), and the attack predicts that the user's target value is $1-s$.
If $r^{j,s} < 5,\,\forall j,s$, the attacks computes the averages $\mu_s = \frac{1}{n-1}\sum_{j=1}^{n-1} r_{j,s}$ and predicts $0$ iff $\mu_0 > \mu_1$ (because the noise is centered).
Rinott et al.~\cite{rinott2018} propose another attack, based on the same queries but using a different combination rule.
This second attack exploits the fact the noise added to a query only depends on the userset, and the same noise is thus added to all the queries selecting the same userset. Then, if the answer to both queries in a pair are equal (i.e., $r^{j,s}=0$), the attack predicts $1-s$ (since the target user is not in the user set of either query, $r^u_{a_n} \neq s$).
Our work is, to the best of our knowledge, the first to empirically evaluate this attack, since the original paper only presented it as a theoretical vulnerability.
We implement both attacks with a small modification to only take into account non-suppressed queries. Additionally, we adapt the attacks to the EXACT-BUT-ONE scenario to reflect the additional knowledge of the attacker, by selecting attribute subsets for which the target record is unique (since they have perfect knowledge of $D$).

\textbf{SimpleQBS. } Non-deterministic noise addition and simple bucket suppression have long been known to be vulnerable to respectively averaging and difference attacks~\cite{denning1979tracker,denning1980secure}.

When $\tau=0$, SimpleQBS is vulnerable to a simple averaging attack, where the query $q_\text{direct} = \bigwedge_{i=1}^{n-1} (a_i = r^u_{a_i}) \land a_n = 0$ is repeated $m$ times to obtain the results $(r_1, \dots, r_m)$. Since the noise is centered, $\mathbb{E}[R_i] = 1$ (resp. 0) if and only if $r^u_{a_n} = 0$ (resp. 1). The attacker predicts $0$ iff $\frac{1}{m}\sum_{i=1}^m r_i < \frac{1}{2}$, and $1$ otherwise.

When $\tau>0$, the query $q_\text{direct}$ is suppressed by bucket suppression. We thus combine a difference and an averaging attack.
We use the attacker's auxiliary information to find values of $(A', i', s)$ such that the user is unique for $(A', i')$ and the queries bypass bucket suppression in two steps. First, we generate auxiliary datasets (from the auxiliary knowledge available in each scenario). Second, we select the pairs $(q_1, q_2)$ for which both assumptions are satisfied for the largest fraction of datasets.
When $\sigma > 0$, we repeat the best pair of queries $m/2$ times, averaging the results to obtain $\mu_{q_1}$ and $\mu_{q_2}$. We then output $s$ iff $\mu_{q_1}-\mu_{q_2} > \frac{1}{2}$, and $1-s$ otherwise.
When $\sigma = 0$, we select the $m/2$ best pairs and perform them in decreasing order until the result $(r_1, r_2)$ is such that $r_1 > 0$, $r_2 > 0$ and $r_1 - r_2 \in \{0,1\}$. We output $s$ iff $r_1 = r_2+1$, and $1-s$ otherwise.

\section{Empirical results}\label{sec:empirical-results}

The goal of our empirical evaluation is fourfold. First, we want to show that it is possible to automate the discovery of attribute inference attacks against a QBS by ``pressing a button''. Second, we want to understand quantitatively how well the attacks discovered by QuerySnout perform when compared to manual attacks. Third, we want to understand qualitatively what vulnerabilities are exploited by QuerySnout, for instance if they are similar to known attacks. Finally, we want to showcase the versatility of QuerySnout by deploying it on a variety of attack scenarios to derive new insights.

\subsection{Attack parameters}

For each target record, we run the evolutionary search and extract after $N$ generations the solution $s^*$ having the highest fitness. 
For computational efficiency and as the attacks against different records are independent, we parallelize the attack over the target records. 
We use 2000 training datasets (sampled from $\traindist$), 1000 validation datasets (sampled from $\valdist$), and 500 test datasets (sampled from $\testdist$). The datasets are of size 8000 for Adults and Census and 1000 for Insurance. 

\textbf{Evolutionary search parameters. } 
Each evolutionary search uses solutions of $m=100$ queries, a population size of $P=100$ and a maximum number of $N=200$ generations with a stopping criteria of having 10 generations of fitness superior to $99.99\%$. To evaluate the fitness of a solution, we use a Logistic Regression as the binary classifier $G$. We also experimented with a multilayer perceptron and found that it did not improve the performance, while significantly increasing the training time.
We set the mutation parameters for modifying operators in a query to $p_{\text{change}}=\frac{1}{n}=\frac{1}{6}$, so that on average we change one operator. Similarly, we use $p_{\text{swap}}=\frac{1}{6}$, $p_{\text{copy}}=0.025$ and $p_{\text{modify}}=0.025$.

\textbf{Attack success metric. } For each target record, we compute the accuracy of the attribute inference attack defined by the best solution $s^*$ and the trained classifier $G$ on $500$ datasets sampled from the corresponding test distribution $D \sim \testdist$. Our metric for the attack success is the average accuracy over the $100$ target users, which we report averaged over 5 repetitions. Note that we randomize the private attribute $a_n$, so that the random guess baseline has an accuracy of 50\%. 

\subsection{Real-world systems}

Our results show that the attacks discovered automatically by QuerySnout  match and often outperform by a large margin the manual attacks from previous work.

\textbf{AUXILIARY scenario. } Table~\ref{table:diffix-comparison-manual}a and Table~\ref{table:table-builder-comparison-manual}a show how, across all datasets and systems, QuerySnout matches or outperforms manual attacks. More specifically, it matches the accuracy of manual attack by Gadotti et al.~\cite{gadotti2019signal} against Diffix on the Adult and Census datasets and strongly outperforms it on the Insurance dataset. Similarly, QuerySnout strongly outperforms Chipperfield et al.~\cite{chipperfield2016} on all datasets and matches Rinott et al.~\cite{rinott2018} on Adult and Census, while strongly outperforming it on Insurance.
\textbf{Our results show how automated attacks, and QuerySnout in particular, can today not only replicate but also outperform existing manual attacks.}
We identify two reasons for the good performances of QuerySnout.

\begin{table}[htbp!]
\centering
\caption{Accuracy against Diffix in the (a) AUXILIARY and (b) EXACT-BUT-ONE scenarios. We report the mean and standard deviation over 5 repetitions.}
\begin{tabular}{|l|c|c|c|}
\hline
\textbf{(a) AUXILIARY} & Adult & Census & Insurance \\ \hline
QuerySnout (automated) & \textbf{77.8} (0.5)  & \textbf{78.3} (1.4) & \textbf{80.1} (0.6)  \\
Gadotti et al.~\cite{gadotti2019signal} (manual) & 76.3 (0.8) & 76.9 (1.4) & 73.0 (1.2) \\ \hline \hline
\textbf{(b) EXACT-BUT-ONE} & Adult & Census & Insurance \\ \hline
QuerySnout (automated) & \textbf{90.2} (0.6) & \textbf{88.3} (0.9) & \textbf{91.6} (1.2)  \\
Gadotti et al.~\cite{gadotti2019signal} (manual) & 77.1 (0.9) & 77.5 (2.0) & 74.4 (0.7) \\
\hline
\end{tabular}
\label{table:diffix-comparison-manual}
\end{table}

\begin{table}[htbp!]
\centering
\caption{Accuracy against TableBuilder in the (a) AUXILIARY and (b) EXACT-BUT-ONE scenarios. We report the mean and standard deviation over 5 repetitions.}
\begin{tabular}{|l|c|c|c|}
\hline
\textbf{(a) AUXILIARY} & Adult & Census & Insurance \\ \hline 
QuerySnout (automated) & \textbf{84.5} (0.6) & \textbf{85.5} (1.4) & \textbf{85.4} (0.6)  \\
Rinott et al.\cite{rinott2018} (manual) & 76.1 (7.5) & 78.1 (7.0) & 56.9 (4.6)  \\ 
Chip. et al.\cite{chipperfield2016} (manual) & 61.2 (3.5) & 62.4 (3.1) & 52.8 (1.8) \\ 
\hline \hline
\textbf{(b) EXACT-BUT-ONE} & Adult & Census & Insurance \\ \hline
QuerySnout (automated) & \textbf{98.1} (0.7) & \textbf{96.6} (0.9) & \textbf{98.8} (0.7)  \\
Rinott et al.\cite{rinott2018} (manual) & 83.1 (8.7) & 72.1 (13.4) & 76.5 (2.3)   \\ 
Chip. et al.\cite{chipperfield2016} (manual) & 72.3 (6.4) & 64.4 (7.2) & 67.2 (1.4) \\\hline 
\end{tabular}\label{table:table-builder-comparison-manual}
\end{table}

First, a manual inspection of queries found by QuerySnout on Diffix shows that difference queries account for $\geq97.0\%$ of the accuracy, while constituting less than half of the queries in solutions (see Appendix~\ref{appendix:manual-analysis} for a detailed analysis). Similarly, the difference queries found by QuerySnout on TableBuilder also account for $\geq97.5\%$ of the accuracy. QuerySnout is able to outperform the manual attacks because it is able to (1) find more attribute subsets for which the target record is likely unique, and (2) combine the results from all subsets taken together for the attack.
We report examples of the solutions found by QuerySnout for each QBS in Appendix~\ref{appendix:solutions-found}.

Second, the gap between QuerySnout and manual attacks is particularly strong on the Insurance dataset, for both QBSes. This is due to previous attacks performing worse on Insurance than on Adult and Census and to QuerySnout performing better. We hypothesize that this difference is due to the smaller size of the Insurance dataset ($|D|=1000$ for Insurance and $|D|=8000$ for Adult and Census), making queries more likely to be suppressed by the query set size restriction (QSSR). While our automated attacks also exploit difference attacks, we believe its flexibility allows it to find more specific queries that bypass QSSR. In Appendix~\ref{appendix:scaling-dataset_size}, we show that QuerySnout performs slightly better in general on smaller datasets.

\textbf{EXACT-BUT-ONE scenario.} Table~\ref{table:diffix-comparison-manual}a and Table ~\ref{table:table-builder-comparison-manual}a show how the accuracy of the best attack found by QuerySnout increases by at least $10\%$ when moving from the AUXILIARY to the EXACT-BUT-ONE scenarios across datasets and QBSes. More specifically, the accuracy of QuerySnout increases by $12.4\%$ (resp. $10.1\%$ and $11.5\%$) for Adults (resp. Census and Insurance) compared to AUXILIARY against Diffix and by $13.7\%$ (resp. $11.2\%$ and $13.4\%$) against TableBuilder. It also strongly outperforms the baselines across all datasets and QBSes, with the exception of Rinott et al. on Census where it matches them.
\textbf{These results show that much stronger attacks than previously believed can be devised against both Diffix and TableBuilder by a strong (``worst-case'') attacker. They also show how QuerySnout is able to automatically discover new and better attacks across scenarios and datasets.}

For instance, in the case of Diffix, an attacker could perform the query $a_1 = r^u_{a_1} \land a_n = 0$. Denote by $C$ the true count of this query when $r^u_{a_n} \neq 0$.
The answer returned by Diffix is distributed as $\mathcal{N}(C + I\{r^u_{a_n} = 0\}, 4)$, assuming the query is thus not bucket suppressed (which will occur only if $r^u_{a_1}$ is very rare).
Repeating this query for different attributes $a_2, \dots, a_{n-1}$ thus gives $n$ samples to distinguish between (in effect) $\mathcal{N}(1,4)$ and $\mathcal{N}(0,4)$.
For $n=5$, a likelihood ratio test distinguishing between these distributions can already achieve an accuracy of $\approx 73\%$.


\subsection{Non-deterministic systems: SimpleQBS}
Most real-world systems are deterministic, either through seeded noise or caching, ensuring that the same answer is returned every time the same query is sent. This is desirable from both a privacy perspective (preventing averaging attacks) and a utility perspective (always returning the same answer).
In some cases though, systems might be non-deterministic. This can be by design, e.g., lack of awareness of averaging attacks; because of implementation or operational issues, e.g., non-functional cache or answer retrieving mechanism; or in cases where the same dataset is made available from different instances of the QBS seeded differently, e.g., distributed computing. 

We here show how QuerySnout is able to find efficient attacks against non-deterministic systems. While we here show and discuss results against the Adult dataset, our conclusions hold for Census and Insurance (cf. Fig.~\ref{figure:qbs-simple_census-insurance} in Appendix~\ref{appendix:simpleqbs}). We also compare our results to the baseline inspired by Denning et al.~\cite{denning1979tracker} and using a simple heuristic to choose difference queries from the auxiliary information available to the attacker (see Sec.~\ref{subsec:manual-attacks} for details).

\begin{figure}[htbp!]
\centering
\includegraphics[width=\linewidth]{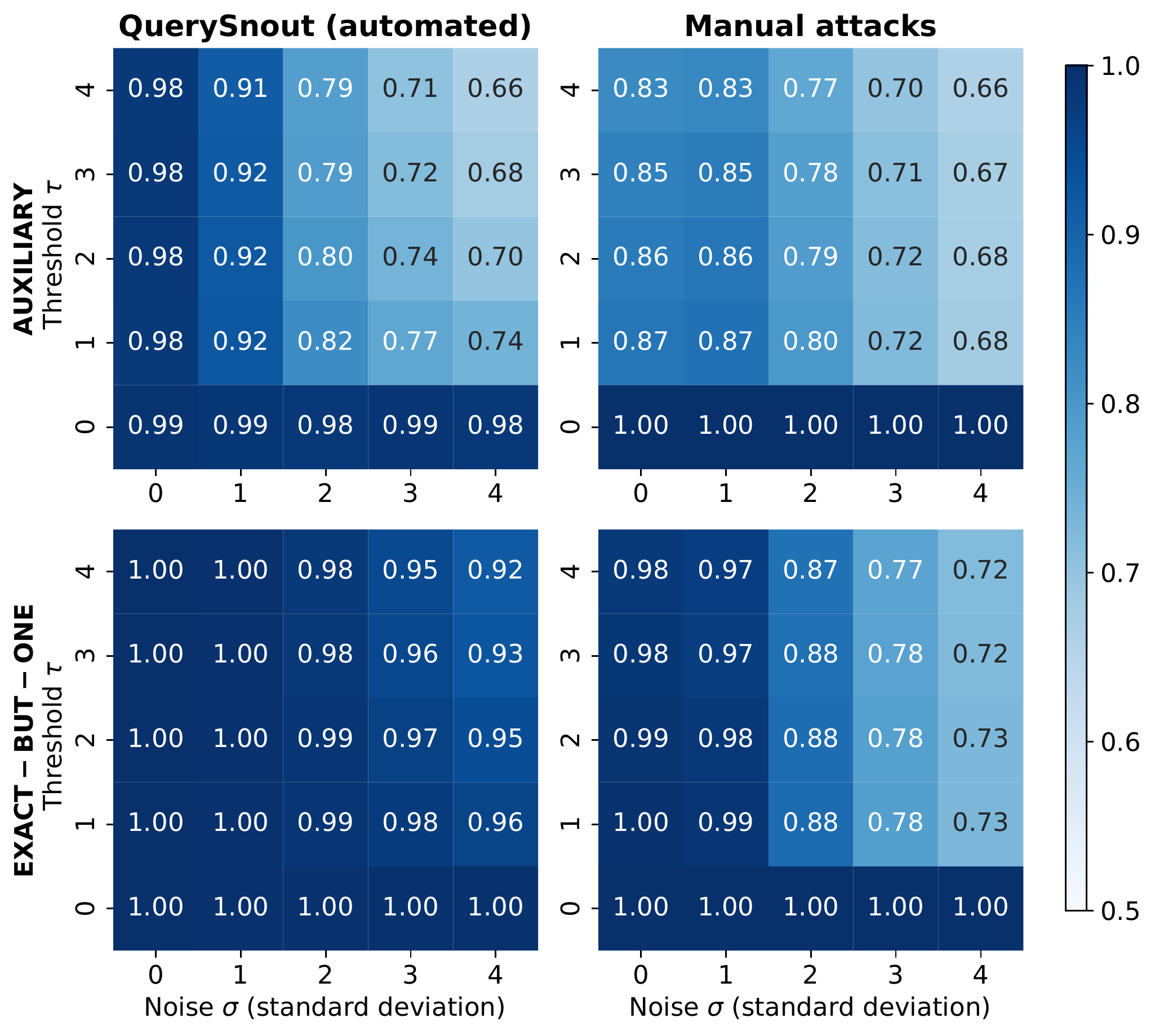}
\caption{\textbf{Comparison between automated and manual attacks against SimpleQBS($\tau, \sigma$) on the Adult dataset.} 
We report the accuracy of attribute inference attacks discovered by QuerySnout (left) and of manual attacks tailored to each system (right) in the AUXILIARY (top) and EXACT-BUT-ONE (bottom) scenarios. The accuracy is averaged over 5 repetitions. 
}
\label{figure:qbs-simple-adults}
\end{figure}

Fig.~\ref{figure:qbs-simple-adults} shows how QuerySnout manages to find effective attacks against SimpleQBS across a range of $\sigma$ (noise) and $\tau$ (threshold) outperforming, often strongly, manual attacks. More specifically, QuerySnout outperforms --on average-- manual attacks in the AUXILIARY scenario by $3.8\%$, particularly when the noise added is small ($\sigma$ from 0 to 3). The gap in accuracy is even larger ($8.4\%$ on average) in the EXACT-BUT-ONE scenario, finding attacks with $>90\%$ accuracy across all values of $(\tau,\sigma)$ considered. Noticeably, QuerySnout finds much better attacks, by 16.7\% on average, than the manual ones for high values of $\sigma$ (from 2 to 4) and $\tau \geq 1$.

\subsection{Comparison with random search}
To evaluate the impact of our search procedure, we compare our results with the accuracy obtained by a random search. The random search uses the same fitness evaluation but samples a new population of random solutions at every generation rather than mutating and keeping previous solutions. The best solution found over all generations is used. We provide more details on this baseline in Appendix~\ref{appendix:random-search}. We here compare our approach to random search for Diffix, Table Builder, and two selected version of SimpleQBS: ($\tau=4,\sigma=3$) and ($\tau=3, \sigma=4$).
We find that our approach strongly improves upon a random search procedure using the same parameters, across all datasets and QBSes (see Table~\ref{table:random-search} in the Appendix). For instance, in the AUXILIARY scenario, our approach outperforms the random search by $6.3\%$ for Diffix, $5.3\%$ for TableBuilder, $6.7\%$ for SimpleQBS($\tau=4,\sigma=3$), and $5.4\%$ for SimpleQBS($\tau=3, \sigma=4$), each time averaged over the datasets. 

\subsection{Impact of the number of queries}
\label{subsec:scaling-number-queries}

Limiting the number of queries is, along with authentication, a popular ``non-QBS'' defense used in practice. This can be done directly, often per user and per period of time, or through budgeting~\cite{mcsherry2009privacy}. We here evaluate the impact of limiting or increasing the number of queries made by our attack to the target QBS. For simplicity, we report results on a deterministic QBS (Diffix) and a non-deterministic QBS (SimpleQBS($\tau=4,\sigma=3$)) on the Adult dataset in the AUXILIARY scenario. The $\tau$ and $\sigma$ of SimpleQBS were chosen to be comparable to those of Diffix, same threshold mean and similar average noise per query (for queries with $5$ conditions).

\begin{figure}[htbp!]
\centering
\includegraphics[width=\linewidth]{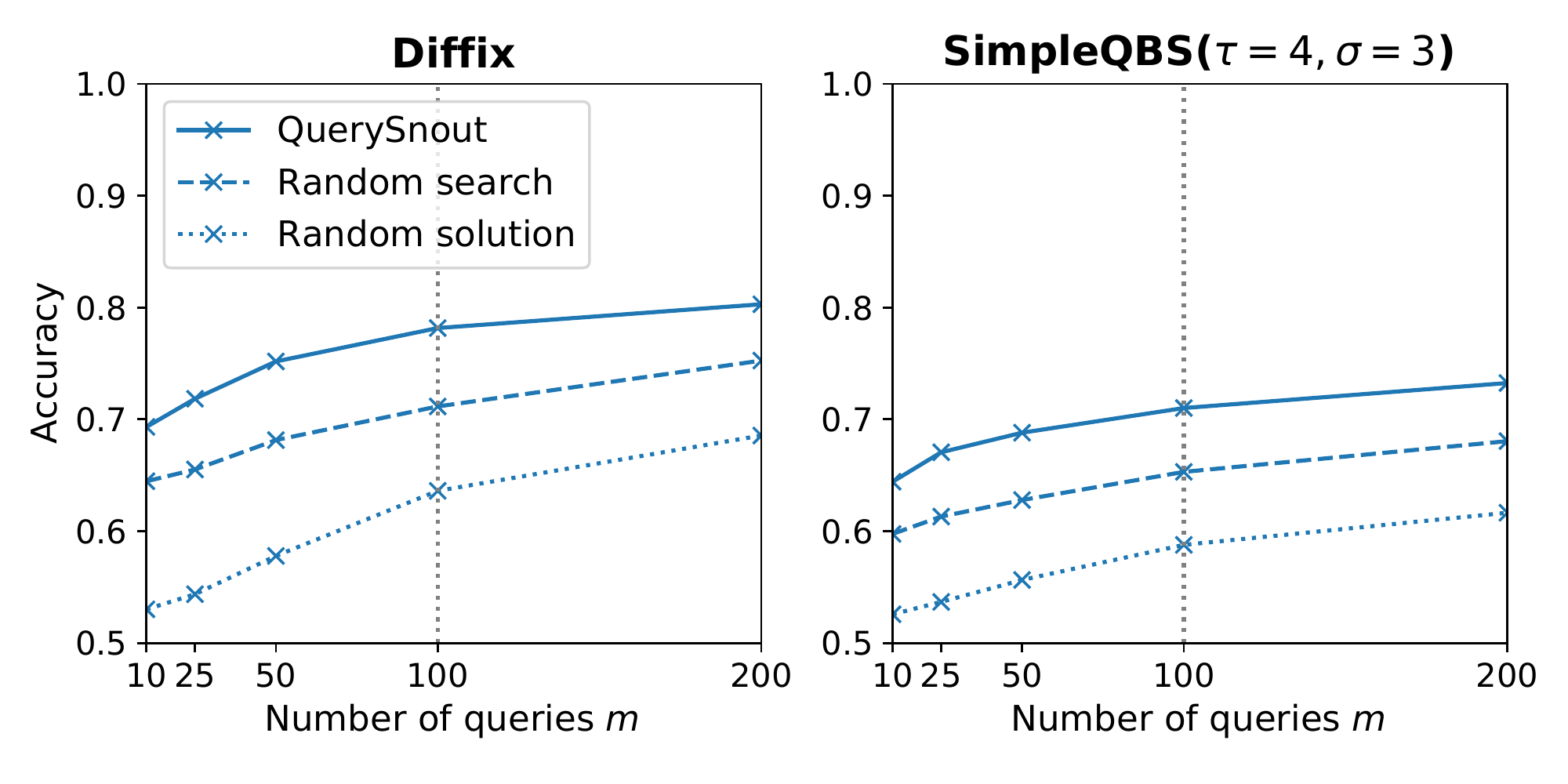}
\caption{\textbf{Impact of the number of queries on the accuracy of QuerySnout.}
We report the accuracy of QuerySnout (which uses an evolutionary search), of a random search, and of a solution chosen uniformly at random as we vary the number of queries $m$. The results are computed on the Adult dataset in the AUXILIARY scenario.
}
\label{figure:scaling-num-queries}
\end{figure}

Fig.~\ref{figure:scaling-num-queries} shows that, as expected, the performance of QuerySnout increases with the number of queries $m$, albeit slowly. Indeed, even with as little as 10 queries, QuerySnout still reaches an accuracy of $69.3\%$ against Diffix and $64.4\%$ against SimpleQBS(4, 3). Increasing the number of queries would further increase the accuracy of the attack, e.g., against systems that do no implement some kind of query limiting mechanisms.

We also evaluate the impact of $m$ on both random search and a na\"ive baseline consisting of one solution sampled uniformly at random from the search space of solutions $\mathcal{S}^m$. Our results show that QuerySnout's mutation operators enable it to consistently strongly outperform the random search, which itself --as expected-- outperforms a random solution. It is however interesting to note that, when the number of queries is large, even a single random solution achieves better than random accuracy.

\subsection{Impact of search parameters}\label{subsec:impact-search-parameters}

We study the impact of different components on the performance of QuerySnout. Fig.~\ref{figure:scaling-search-parameters} in the Appendix shows that increasing the size of the population $P$, the number of generations $N$, or the number of datasets (auxiliary QBSes) results in better performance, with decreasing returns. However, the marginal increase in performance comes at a computational cost, as doubling any of the parameters will double the running time. In some cases, the memory footprint also increases significantly, e.g., increasing linearly with the number of auxiliary datasets.

For the mutation operators, Table~\ref{table:ablation} in the Appendix shows ``copy'' to be the most impactful one. For non-deterministic systems, ``copy'' duplicates a query and -- with equal probability -- either keeps the duplicate unchanged or modifies it with $p_{\text{change}}$ and $p_{\text{swap}}$. For deterministic systems, ``copy'' duplicates a query and always modifies it, since repeating the same query would yield the same answer. In both cases, removing this component negatively impacts the evolutionary search by a significant margin. Even though queries are still being modified (according to $p_{\text{modify}}$), it becomes harder for QuerySnout to find difference queries or to repeat good queries many times. As for the mutations of query operators, the search is robust to removing either ``swap'' or ``change'', as there is significant overlap between the two mutations. We further hypothesize that ``swap'' is likely to be more useful when increasing the number of condition operators $|\conditionset_s|$ (e.g., allowing for $<, \leq, \geq, >$ operators).

\section{Budget-based systems}
\label{sec:budget-based-mechanisms}

Some QBSes, mostly based on $\varepsilon-$differential privacy~\cite{dwork2006calibrating}, require the user to divide a given privacy budget $\varepsilon$ between the queries they ask. We here present a simple heuristic to apply QuerySnout to QBSes that require a budget, which we call Budget-Based Systems (BBS).

\subsection{Extending QuerySnout to BBSes}

Formally, we model budget-based systems similarly to other QBSes, with the addition of a parameter called the \textit{partial budget} of the query, specifying the fraction of the total budget ($\varepsilon$) to use on this query:
$R:\mathcal{D}\times\mathcal{Q}\times(0,1]$.
A new query is only answered if the sum of partial budgets used and of the new query's partial budget is less than or equal to $1$. 

We here propose a simple heuristic to transform the multiset of queries manipulated by QuerySnout to a list of pairs of query and partial budget.
Given a solution of $m$ queries $s = (c^1, \ldots, c^m)$ we group identical queries to obtain $k$ unique queries with multiplicities $(c^{i_1}, m_1), \ldots, (c^{i_k}, m_k)$ such that $\sum_{i=1}^k m_i=m$. We allocate to each unique query a partial budget proportional to its multiplicity. This way, we perform $k$ queries to the QBS, where the partial budget allocated to the $i$-th query is equal to $\frac{m_i}{m}$.
This heuristic assumes \textit{monotonicity of accuracy}: performing any query $q$ with full budget yields more accurate results than repeating a query $k$ times with a partial budget of $1/k$ and averaging out the results.
We discuss this in Appendix~\ref{appendix:monotonicity}.

Note that while running the evolutionary search, we are performing the queries on QBSes instantiated on the auxiliary datasets (resetting their budget at each iteration), and only the final $k \leq m$ queries will be performed on the target QBS. Hence the budget of the target QBS is only used once.

\begin{table*}[!htbp]
\centering
\caption{Accuracy against DPLaplace for different values of $\varepsilon$ in the AUXILIARY and EXACT-BUT-ONE scenarios. We report the mean and standard deviation over 5 repetitions.}
\begin{tabular}{|l|ccc|ccc|ccc|}
\cline{2-10}
\multicolumn{1}{c|}{}
 & \multicolumn{3}{c|}{\textbf{DPLaplace($\varepsilon=1$)}} & \multicolumn{3}{c|}{\textbf{DPLaplace($\varepsilon=5$)}}  & \multicolumn{3}{c|}{\textbf{DPLaplace($\varepsilon=10$)}}  \\
 \hline
\textbf{AUXILIARY} & Adult & Census & Insurance & Adult & Census & Insurance & Adult & Census & Insurance \\
\hline
QuerySnout (automated) & 62.9 (0.8) & 64.2 (0.8) & 63.6 (0.8) & 94.2 (1.0) & 94.6 (1.0) & 94.9 (0.7) & 98.6 (0.4) & 99.4 (0.2) & 99.0 (0.3) \\
Uniqueness attack (manual) & 69.2 (0.8) & 70.2 (0.9) & 70.4 (1.2) & 95.6 (0.8) & 95.5 (0.7) & 95.9 (0.6) & 99.6 (0.2) & 99.6 (0.2) & 99.7 (0.1) \\
\hline
\textbf{EXACT-BUT-ONE} & Adult & Census & Insurance & Adult & Census & Insurance & Adult & Census & Insurance \\
\hline
QuerySnout (automated) & 64.2 (0.7) & 65.0 (0.6) & 65.6 (0.8) & 95.4 (0.8) & 95.3 (0.7) & 95.7 (0.7) & 99.5 (0.3) & 99.6 (0.1) & 99.7 (0.1) \\
Uniqueness attack (manual) & 68.6 (1.1) & 68.9 (2.6) & 69.5 (2.6) & 95.7 (1.1) & 95.6 (1.1) & 96.0 (0.6) & 99.5 (0.3) & 99.5 (0.4) & 99.8 (0.2) \\
\hline
\end{tabular}\label{table:dp-laplace-comparison-manual}
\end{table*}

\subsection{DPLaplace}
\label{subsec:dplaplace}

We consider a simple budget-based mechanism using addition of Laplace noise.
The mechanism assumes that a total budget $\varepsilon$ is allocated by the data curator, and that each query is answered by adding independent Laplace noise scaled with $(p\varepsilon)^{-1}$, where $p$ is the fractional budget allocated for this query.
This ensures that each answer satisfies $(p\varepsilon)-$Differential Privacy (DP)~\cite{dwork2006calibrating} and, by composition~\cite{mcsherry2009privacy}, performing $k$ queries with fractional budgets $(p_1, \dots, p_k)$ satisfies $(\sum_{i=1}^k p_i \varepsilon)-$DP.
Formally, having answered $k-1$ queries with total fractional budget $p_{1:k-1}$, the QBS answers the $k^\text{th}$ query with fractional budget $p_k$ iff $p_{1:k-1} + p_k \leq 1$, and answers as:
$$R_{\text{DPLaplace}(\varepsilon)}(D,Q,p) = T(D,Q) + L~~\text{ with } L \sim Lap\left((p\,\varepsilon)^{-1}\right)$$
Similarly to the other systems, we round answers to the nearest integer and threshold the answers at $0$, which doesn't affect the privacy guarantees.
We prove in Appendix~\ref{appendix:monotonicity} that this QBS satisfies the monotonicity of accuracy assumption.
Although the DPLaplace mechanism is particularly simple, it serves as the core component of many more complex systems implementing Differential Privacy, such as PINQ~\cite{mcsherry2009privacy}, Chorus~\cite{near2018differential} and PriPearl~\cite{kenthapadi2018pripearl}.
We instantiate this QBS for $\varepsilon\in\{1,5,10\}$. As we explain below, finding solutions for our search procedure is challenging against DPLaplace.

\textbf{Known optimal attack. }
A known \emph{optimal} attack exists against DPLaplace: a uniqueness attack that performs the query $q = \text{COUNT}\\ \text{WHERE} \bigwedge_{i=0}^{n-1} (a_i = r^u_{a_i}) \land a_n = 0$ with partial budget $1$ and returns $s=0$ iff the count is larger than 0.5. We prove that this attack achieves maximal accuracy in Appendix~\ref{appendix:dplaplace}. Given a number of queries $m$, this attack can be mapped to the solution in $\mathcal{S}^m$ consisting of the query $q$ repeated $m$ times. While the attack is simple for a knowledgeable attacker, it present challenges for our search procedure. Showing that QuerySnout can be extended to budget-based mechanisms is therefore important.

Finding good solutions with our current black-box search is challenging here, particularly for small values of $\varepsilon$. Indeed, a query that is not repeated (equivalently, that has a small partial budget) will be answered with a noise of large standard deviation $\approx 14$ for $\varepsilon=1$ and $m=10$.
This has two consequences: first, estimating the fitness of solutions is difficult (due to the increased randomness), which can lead to overfitting, both of the rule $G$ and the fitness; second, it complicates the discovery of good queries, since the signal of optimal queries is significantly smaller than the noise added. For instance, even the optimal query $\bigwedge_{i=1}^{n-1} (a_i = r^u_{a_i}) \land s=0$ has little signal if not repeated, since performing it once (with budget $\frac{1}{10}$) leads to an attack with accuracy of at most $\approx52.4\%$. 

\textbf{Attack parameters. }
Unlike the previous mechanisms, increasing the number of queries in a solution $m$ amounts to adding more noise to the answers.
The standard deviation of the noise is indeed proportional to $m$, due to the budget being split among the queries (and is equal to $\sqrt{2}m/\varepsilon$ for queries with multiplicity 1).
We thus use smaller solutions of $m=10$ queries.
All other parameters are identical as those in Sec.~\ref{sec:empirical-results}, except we update the mutation operators accordingly to $p_{\text{copy}}=p_{\text{modify}}=\frac{1}{m} = 0.1$.  

\subsection{Empirical results}\label{subsec:dp-laplace-results}

Table~\ref{table:dp-laplace-comparison-manual} shows that the attacks discovered by QuerySnout match the performance of the optimal attack for large budget values $\varepsilon \in \{5, 10 \}$ but fall slightly short in the $\varepsilon=1$ case. Indeed, in the latter case QuerySnout performs on average $6.4\%$ worse than the optimal attack for AUXILIARY, and $4.1\%$ for EXACT-BUT-ONE.

These results show that QuerySnout can already find the optimal attacks against budget-based mechanisms for medium to small amounts of noise, and very good ones when a large amount of noise is added. We believe the search can be further improved in future work. 


In Appendix~\ref{appendix:random-search}, we report the difference in accuracy between QuerySnout and the random search for DPLaplace. Table~\ref{table:dp-laplace-random-search} shows QuerySnout to vastly outperform the random search for all datasets, scenarios, and values of $\varepsilon$. Specifically, in the AUXILIARY scenario, the gap is of $11.5\%$ on average across the datasets for $\varepsilon=1$, with the random search (at $52.1\%$) barely improving on the random guess baseline. Similar or higher gaps (up to $29.3\%$) are obtained for other datasets, scenarios, and values of $\varepsilon$. 
This result both confirms that finding good solutions for DPLaplace is challenging, and shows that QuerySnout is able to efficiently explore the search space even in complex cases. 

Overall, the gap between QuerySnout and the random search is much larger for DPLaplace compared to the other QBSes considered in this paper. This is due to (1) random solutions yielding low accuracy against DPLaplace (as explained above) and (2) the noise added by the other QBSes to answers on one query being overall smaller ($\leq \sqrt{12}$ for Diffix and $\sigma \leq 5$ for SimpleQBS, compared to $10\sqrt{2}$ for DPLaplace$(\varepsilon=1)$) and independent of the number of queries. 

\section{Discussion}
\label{sec:discussion}

QuerySnout \textit{automatically} finds attacks that equate and often outperform previous manual attacks across a range of systems and datasets. This demonstrates the potential of automating attacks against QBSes and opens the door for new, much more complex attacks to be discovered. Our work is the first to enable systems to be automatically tested before deployment to identify and patch vulnerabilities. We now discuss in more detail the extensive potential avenues for future work.

\textbf{Extensions. }
QuerySnout currently operates in the same constrained space explored by manual attacks. More complex and accurate attacks exploiting e.g. different SQL operators ($\leq, \geq$) or logical expressions (such as OR) are, in our opinion, likely to exist. Finding such attacks however presents some challenges as the size of the search space increases fast with new operators or conditions per attributes, and with it the computational cost of the search. Specialised mutation operators can be used to overcome this challenge, for instance, using empirical discrete gradient approximation~\cite{dorogush2018catboost}. Furthermore, generative encodings~\cite{stanley2009hypercube} can be used to improve the expressiveness of the genotype and enhance the capability of QuerySnout to find regularities in query sets.

QuerySnout currently relies on a set of assumptions, primarily that the target record is unique in the private dataset and that the attacker has access to auxiliary knowledge on the target dataset. We believe both can be relaxed. In particular, different auxiliary information can likely be used, for instance a few statistics on the data from which synthetic data is generated~\cite{shokri2017membership}, or syntactically similar datasets from another distribution. These would not require changing the method as QuerySnout only requires two samplers $\traindist$ and $\valdist$.

The attacks that we consider here are \textit{non-interactive}: the attacker chooses which queries to perform \textit{before} observing answers from the target QBS. It is likely that attacks could be made more accurate or efficient if future automated attacks were to use the answer to previous queries to inform the choice of the next queries~\cite{gadotti2019signal}.

Finally, QuerySnout could be extended to use additional QBS outputs such as the time it took for the answer to a query to be computed (such as exploited by Boenisch et al.~\cite{boenisch2021side}), information about the query provided by the QBS, or warnings and error messages produced by the QBS (e.g., allowing our system to treat bucket suppression as different from a $0$). This would enable the system to discover so-called \textit{side-channel attacks} in addition to the privacy attacks that this paper focuses on.

\textbf{Membership inference attacks. }
QuerySnout currently focuses on attribute inference attacks.
Another class of attacks, membership inference attacks (MIA), has been considered in the privacy literature. MIAs are at the core of formal privacy guarantees such as differential privacy~\cite{dwork2006calibrating} and can lead to privacy concerns when being part of a dataset is sensitive. 
More importantly here, membership inference can also be a first step towards attribute inference attacks by allowing the attacker to verify the assumption that the target record is in the dataset.
QuerySnout can be extended to membership inference attacks with minimal changes.
Indeed, the main change to our method is in the sampling of auxiliary datasets used to train the attack: after sampling a dataset (as a subset of $D'$ in the AUXILIARY scenario, or as $D$ in the EXACT-BUT-ONE), with equal probability either a random record in the dataset is replaced with the target record, or the auxiliary dataset left unchanged.
The rule $G$ can then be trained to predict membership (whether the target record is in the data) rather than the value of the sensitive attribute.
The rest of the method can be applied unchanged, although it is likely that designing specific mutations for membership inference might result in better accuracy.

\textbf{Automating the analysis of solutions. }
In this paper, we manually inspected the queries found by QuerySnout to understand what vulnerabilities of the QBS they exploited.
This analysis could likely be automated in the future, e.g. using techniques from interpretability in machine learning~\cite{ribeiro2016} to understand the relative importance of queries in the attack. 
This would greatly help better understand what the attack exploits and how the attack could be mitigated and patched to make the system stronger.

\textbf{Automating the defense. }
QuerySnout discovers highly accurate attacks against real-world systems in the AUXILIARY scenario, and new, more powerful, attacks in the EXACT-BUT-ONE scenario. 
We believe that defending against these attacks will likely require both improving these systems but also ensuring that appropriate risk mitigation strategies are in place, including authentication, query limits, and logging. 
We believe however that the existence of an automated attack system such as QuerySnout might help develop stronger mitigation strategies such as anomaly detection and attack detection mechanisms.

\textbf{Order of queries. }
We here assume that the order of queries does not matter to the QBS (Sec ~\ref{subsec:attack-search-space}), a property satisified by all the query-based systems we consider, as well as, in our opinion, the majority of real-world systems. In some cases however, this might not be the case. For instance, one could imagine a system auditing queries in order to avoid answering the second half of difference query pairs~\cite{nabar2008survey}. 
Such a QBS could be represented by a function $R_k:\mathcal{D}\times\mathcal{Q}\times\mathcal{H} \rightarrow \mathbb{R}$, where $\mathcal{H}$ is the set of \textit{histories}, i.e., lists of query-answer pairs. Here, solutions would be represented as \textit{ordered} lists of queries rather than multisets, which increases the size of the solution search space by a factor of approximately $m!/3^n$.
Our method can be adapted by adding a mutation operator which changes the query order.

\section{Related work}
\label{sec:related_work}

Attacks have long been used to evaluate and improve query-based systems.
Difference attacks against counting queries were studied by Denning et al.~\cite{denning1979tracker}, as well as solutions based on randomisation~\cite{denning1980secure}, leading to the development of averaging attacks.
The seminal linear reconstruction attack by Dinur and Nissim~\cite{dinur2003revealing} allows to retrieve the records of a dataset using (many) counting queries, and gives a lower bound on the noise added by query-based systems to prevent it.
Further work by Dwork and Yekhanin~\cite{dwork2008new} showed that this attack could be performed efficiently with a smaller number of queries, and that the impossibility result extends to curators that aim to answer a small subset of queries accurately (while maintaining average noise level satisfying the lower bound).
Diffix~\cite{francis2017diffix}, developed by Aircloak, was improved through two bounty challenges~\cite{bauer2017bounty,bauer2020bounty}, where researchers were invited to develop attacks against it. 
Researchers showed that Diffix was vulnerable to reconstruction attacks using direct~\cite{cohenLinearProgramReconstruction2018} and quasi-identifiers~\cite{cohen2020reconstruction}, as well as membership inference attacks~\cite{pyrgelis2018blogpost}.
Gadotti et al. proposed an attack specifically exploiting the noise structure of Diffix~\cite{gadotti2019signal}.

Query auditing~\cite{nabar2008survey} is a line of research which aims to detect whether queries can be answered without disclosing private information.
Auditing can be performed in the offline setup, where $m$ queries are audited together before being answered, and in the online setup, where each query is audited individually before being answered without knowledge of subsequent queries.
While the latter is known to be challenging, even the former can be difficult: for instance, query auditing is co-NPHard for counting queries~\cite{kleinberg2003auditing}. 
Research on query auditing has been restricted to systems that answer queries exactly ($R(d,q) = T(d,q)$).
QuerySnout can be seen as an automated auditor applicable to more general QBSes.

Differential privacy~\cite{dwork2006calibrating} is a formal guarantee of privacy, which was proposed as a solution to the reconstruction attack of Dinur and Nissim.
Researchers have proposed query-based systems guaranteeing differential privacy, such as PINQ~\cite{mcsherry2009privacy}, PriPearl~\cite{kenthapadi2018pripearl}, and Flex~\cite{johnson2018towards}.
For the counting queries we consider, all three systems behave similarly to the DPLaplace QBS we attack in this paper.
More complex mechanisms guaranteeing differential privacy have also been developed, such as query release through adaptive projection~\cite{dwork2014algorithmic,aydore2021differentially} and the matrix mechanism~\cite{li2010optimizing,li2015matrix}.

Our method uses a technique similar to shadow modelling to infer the rule $G$ that combines answers to queries.
Shadow modelling has been used extensively for property inference~\cite{ateniese2013} and membership inference~\cite{shokri2017membership} attacks against machine learning algorithms.
Pyrgelis et al. have used it to attack aggregated location data~\cite{pyrgelis2017knock}, and Stadler et al. have used it to develop membership inference attacks against synthetic data~\cite{stadler2020synthetic}. Machine learning-based attacks have also been used by Bichsel et al. to detect violations of DP~\cite{bichsel2021dp}.

Evolutionary algorithms have been applied to adversarial attacks against machine learning models~\cite{Nguyen_2015_CVPR,su2019one,dai2018adversarial,alzantot2019genattack}. They are especially suitable under black-box attack scenarios where the attacker does not have access to the gradients or only hard (discrete) labels are available from the model. To the best of our knowledge, our work is the first to use evolutionary algorithms for attacks against QBSes.

\section{Conclusion}
\label{sec:conclusion}

In this paper, we propose QuerySnout, the first approach to automatically discover attribute inference attacks against query-based systems.
QuerySnout discovers a multiset of queries and a rule to combine them. We learn the rule by training a machine learning classifier on answers from auxiliary QBSes protecting datasets sampled from the auxiliary knowledge available to the attacker. We use an evolutionary algorithm to find an optimal multiset of queries by iteratively improving a population of solutions using a mutation operator specifically tailored for this task.
We show QuerySnout to find attacks against two deterministic real-world mechanisms, Diffix and TableBuilder, and a non-deterministic system (SimpleQBS). Across systems and datasets, the attacks found equate and often outperform previous known manual attacks. 
Finally, we show how QuerySnout can be extended to QBSes that require a budget and show it to approximately match the optimal accuracy against systems implementing the Laplace mechanism for large values of $\epsilon$.
We also discuss extensions of our attack to other attack models (e.g., membership inference) and query syntax.

Taken together, our results show how QuerySnout can be used to automatically find powerful attacks against QBSes and evaluate the privacy protection they offer. We believe that automated attack discovery procedures such as QuerySnout will help detect issues before systems are deployed, helping to patch the systems and mitigate risks, and ultimately help design query-based systems offering a high level of protection in practice.

\textbf{Acknowledgements.} We thank Bozhidar Stevanoski for his feedback, the anonymous reviewers for their comments which have helped improve the paper, and Tianhao Wang for shepherding our paper. 
A.-M. C. was partially funded by UKRI Research England via the ``Policy Support Fund 2021/22 Evidence-based policy making'' call. 
We acknowledge computational resources and support provided by the Imperial College Research Computing Service\footnote{\url{http://doi.org/10.14469/hpc/2232}}.

\bibliographystyle{ACM-Reference-Format}
\bibliography{references}

\appendix
\section{Appendix}

\startcontents[sections]
\printcontents[sections]{l}{1}{\setcounter{tocdepth}{2}}

\subsection{Background on evolutionary algorithms} 
The goal of an evolutionary algorithm is to produce one or more high-quality \textit{solutions} $s$ to an optimization problem with search space $\mathcal{S}$, as quantified by a \textit{fitness} function $\mathcal{F}:\mathcal{S} \rightarrow \mathbb{R}_+$. The fitness function assigns a score $\mathcal{F}(s)$ to each solution $s$, measuring its quality for the specific use case. An evolutionary algorithm consists of an iterative procedure operating on a \textit{population} of $P$  solutions $s_1, \ldots, s_P \in \mathcal{S}$. Its goal is to improve the quality of solutions over time using the feedback from the fitness scores of solutions in the population. This improvement is typically achieved by (1) maintaining a population of high-quality solutions selected among those discovered so far, (2) evolving the population over many iterations (\textit{generations}) to improve their quality. In each generation, new solutions (\textit{offsprings}) are generated by applying small random changes via a \textit{mutation operator} $\mathcal{M}: \mathcal{S} \rightarrow \mathcal{S}$ to existing solutions (\textit{parents}), or by combining them via a \textit{crossover operator} $\mathcal{C}:\mathcal{S}\times\mathcal{S}\rightarrow\mathcal{S}$.
While the ideas behind the algorithm are general, its building blocks: solution, fitness function and mutation or crossover operators need to be defined for the specific use case. 

\subsection{Fitness evaluation}
Algorithm~\ref{alg:evaluate-fitness} details our procedure to evaluate the fitness of a solution. We refer the reader to Sec.~\ref{subsec:search-space-exploration} of the main paper for a high-level description.

\begin{algorithm}[t]
\caption{\textsc{EvaluateFitness}}
\label{alg:evaluate-fitness}
    \begin{algorithmic}[1]
        \Inputs{
        $solution$: A solution in the search space $\mathcal{S}^m$.
        $D'_{\text{aux}}=(D'_{\text{train}}, D'_{\text{val}})$: Two lists of auxiliary datasets \\ used to train and evaluate the fitness function.\\
        }
        
        \Output{
        $fitness$: Accuracy of an ML model trained to predict the sensitive attribute of a record.
        }
        
        \Initialize{\comm{Answers to queries used as input features} \\
        \comm{for training and evaluation.}\\
        $X_{\text{train}}, X_{\text{val}} \gets empty\_list, empty\_list$\\
        \comm{Corresponding sensitive attributes.}\\
         $y_{\text{train}}, y_{\text{val}} \gets empty\_list, empty\_list$\\
        }
        
        \For{$\text{split} \in \{\text{train}, \text{val}\}$}
        \For{$k=1$ to $length(D'_{\text{split}})$}
        \State{$answers \gets empty\_list$}
        \For{$c$ in $solution$}
        \State{$q = \text{COUNT WHERE} \bigwedge_{i=1}^{n-1} (a_i\,c_i\,r^u_{a_i}) \land (a_n\,c_n\,0)$}
        \State{$answers.append(R(D'_{\text{split}}[k],q))$}
        \EndFor
        \State{$X_{\text{split}}.append(answers)$}
        \State{$y_{\text{split}}.append(r^u_{a_n})$}  \comm{Extracted from $D_{\text{split}}[k]$.}
        \EndFor
        \EndFor
    \State{$G = train\_model(X_{\text{train}}, y_{\text{train}})$ \comm{Train the rule $G$.}}
    \State{\comm{Evaluate the model on each split.}}
    \State{$\hat{y}_\text{train}, \hat{y}_\text{val}  \gets G.predict(X_{\text{train}}), G.predict(X_{\text{val}})$}
    \State{$fitness\_train, fitness\_val \gets acc(\hat{y}_{\text{train}},  y_{\text{train}}), acc(\hat{y}_{\text{val}}, y_{\text{val}})$}
    \State{$fitness\_val \gets  accuracy(\hat{y}_{\text{val}}, y_{\text{val}})$}
    \State{$fitness \gets min(fitness\_train, fitness\_val)$}
    \end{algorithmic}
\end{algorithm}

\subsection{Mutation algorithms}
\label{appendix:apply-mutation}
Algorithm~\ref{alg:apply-mutation} presents the pseudocode for our mutation algorithm and Algorithm~\ref{alg:modify-query} presents the pseudocode for the algorithm to modify a query. We refer the reader to the main paper for a high-level description.

\begin{algorithm}[htbp!]
\caption{\textsc{ApplyMutation}: Mutation algorithm.}
\label{alg:apply-mutation}
    \begin{algorithmic}[1]
        \Inputs{
        $parent$: A solution in the attack search space $\mathcal{S}^m$. \\
        $p_{\text{copy}}$: Probability to copy a query. \\
        $p_{\text{modify}}$: Probability to modify a query. \\
        $p_{\text{change}}$: Probability to change an operator in a query. \\
        $p_{\text{swap}}$: Probability to swap two operators in a query.
        }
        
        \Output{
        $offspring$: New solution in the attack search space $\mathcal{S}^m$.}
        \Initialize{$offspring \gets empty\_list$
        }
        \For{$i = 1$ to $m$}
        \State{$c \gets solution[i]$}
        \State{$u \gets \mathcal{U}(0, 1)$}
        \State{\comm{Case 1: Copy the query.}}
        \If
        {$0 \leq u < p_\text{copy}$}
        \State{$offspring.append(c)$}
        \State{\comm{Case 1a: Modify the copy.}}
        \If{$qbs\_deterministic$}
        \State{$c' \gets \textsc{ModifyQuery}(c, p_\text{change}, p_\text{swap})$}
        \State{$offspring.append(c')$}
        \Else
        \State{\comm{Case 1b: Modify the copy or add the query unchanged with equal chance.}}
        \State{$u' \gets \mathcal{U}(0, 1)$}
        \If{$u' > 0.5$}
        \State{$c' \gets \textsc{ModifyQueryRule}(c, p_\text{change}, p_\text{swap})$}
        \State{$offspring.append(c')$}
        \Else
        \State{$offspring.append(c)$}
        \EndIf
        \EndIf
        \EndIf
        \State{\comm{Case 2: Modify the query.}}
        \If{$p_{\text{copy}} \leq u < p_{\text{copy}} + p_{\text{modify}}$}
        \State{$c' \gets \textsc{ModifyQuery}(c, p_\text{change}, p_\text{swap})$}
        \State{$offspring.append(c')$}
        \EndIf
        \State{\comm{Case 3: Add the query unchanged.}}
        \If{$p_{\text{copy}} + p_{\text{modify}} \leq u$}
        \State{$offspring.append(c)$}
        \EndIf
        \EndFor
    \State{\comm{Keep a random subset of $m$ queries.}}
    \State{$offspring \gets random\_permutation(offspring)[:m]$}
    \end{algorithmic}
\end{algorithm}

\newpage
\begin{algorithm}[htbp!]
\caption{\textsc{ModifyQuery}: Algorithm to modify a query.}
\label{alg:modify-query}
    \begin{algorithmic}[1]
        \Inputs{
        $c = (c_1, \ldots, c_n)$: Query to be modified. \\
        $p_{\text{change}}$: Probability to switch an operator. \\
        $p_{\text{swap}}$: Probability to swap two query operators.
        }
        
        \Output{
        $new\_c$: A modified version of the query $c$.
        }
        \Initialize{
        $new\_c \gets (c_1, \ldots, c_n)$
        }
        \State{\comm{Randomize the attribute order.}}
        \State{$\sigma \gets random\_permutation(n)$}
        \State{$swapped \gets empty\_set$}
        
        \For{$i=1$ to $n$}
        \State{\comm{Skip already swapped attributes.}}
        \If{$swapped.contains(i)$}
        \State{$continue$}
        \EndIf
         \State{$u \gets \mathcal{U}(0, 1)$}
        \State{\comm{Case 1: Change the operator.}}
        \If
        {$0 \leq u < p_\text{change}$}
        \State{$new\_c[\sigma(i)] \gets \mathcal{U}(\conditionset_s \setminus \{ c_i\})$}
        \EndIf
        \State{\comm{Case 2: Swap two operators.}}
        \If
        {$p_\text{change} \leq u < p_\text{change}+p_\text{swap}$}
        \State{$new\_c[\sigma(i)] \gets c_{\sigma(i+1)}$}
        \State{$new\_c[\sigma(i+1)] \gets c_{\sigma(i)}$}
        \State{$swapped.add(i+1)$}
        \EndIf
        \State{\comm{Case 3: Leave the query unchanged.}}
        \EndFor
    \end{algorithmic}
\end{algorithm}

\subsection{Real-world mechanisms: Analysis of solutions}
\label{appendix:manual-analysis}

We analyze the solutions found by QuerySnout against Diffix and TableBuilder in the AUXILIARY scenario. First, we automatically identify and extract the difference queries from each solution (corresponding to a target record) and repetition. We then fit a Logistic Regression on answers to these queries using the same seed and hyperparameters as in the default experiment. Fig.~\ref{figure:diffix_analysis_number_difference_queries} shows that out of $m=100$ queries, QuerySnout uses an average of 29.0 difference queries on Adult, 26.9 on Census and 29.0 on Insurance. These queries account for $98.0\%$ of the accuracy on Adult and for $97.0\%$ of the accuracy on Census and Insurance, as shown on Fig.~\ref{figure:diffix_analysis_difference_queries}. Note that using a random subset of each solution of the same size only preserves, on average, $80.0\%$ of the performance on Adult, $79.8\%$ on Census, and $80.5\%$ on Insurance. We obtain similar results on TableBuilder, shown on Fig.~\ref{figure:table_builder_analysis_number_difference_queries}-~\ref{figure:table_builder_analysis_difference_queries}. Specifically, QuerySnout finds fewer difference queries (an average of 23.9 on Adult, 20.7 on Census, and 24.2 on Insurance), but those found account for most of the performance: $99.8\%$ on Adult, $97.5\%$ on Census, and $98.2\%$ on Insurance. The randomly selected queries only account for $73.1\%$ of the performance on Adult, $71.7\%$ on Census, and $74.8\%$ on Insurance.

\begin{figure*}[htbp!]
\centering
\includegraphics[width=0.7\linewidth]{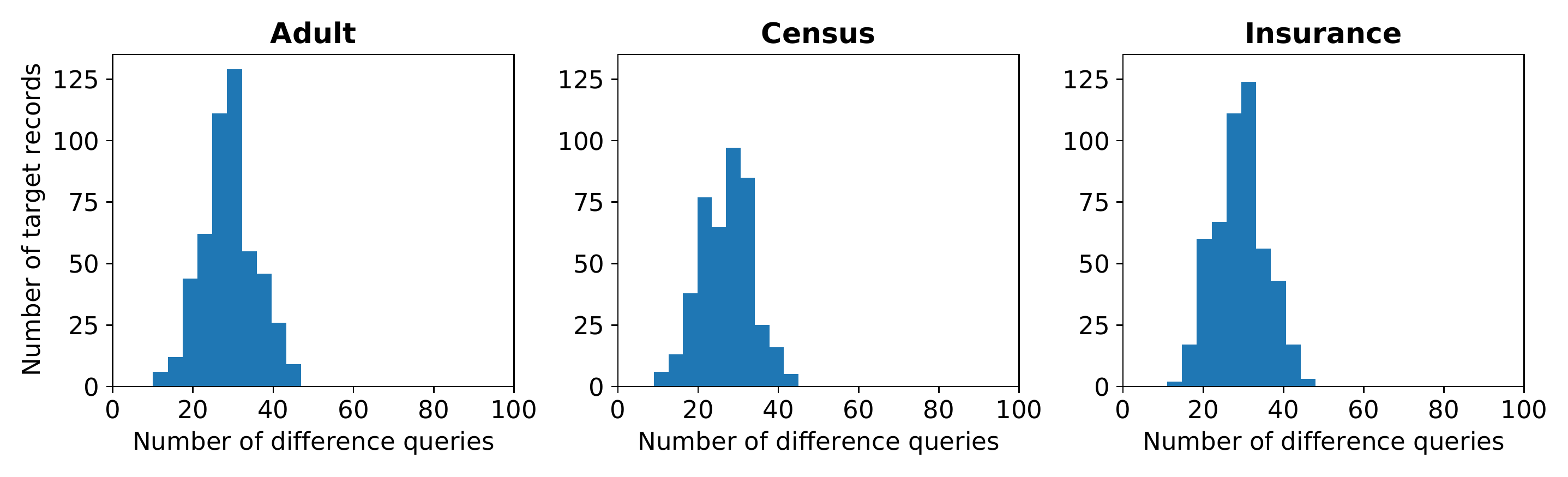}
\caption{\textbf{Diffix: Histogram of the number of difference queries in each solution found by QuerySnout in the AUXILIARY scenario.} 
For each dataset, the histogram aggregates the results from the 5 repetitions.
}
\label{figure:diffix_analysis_number_difference_queries}
\end{figure*}

\begin{figure*}[htbp!]
\centering
\includegraphics[width=0.78\linewidth]{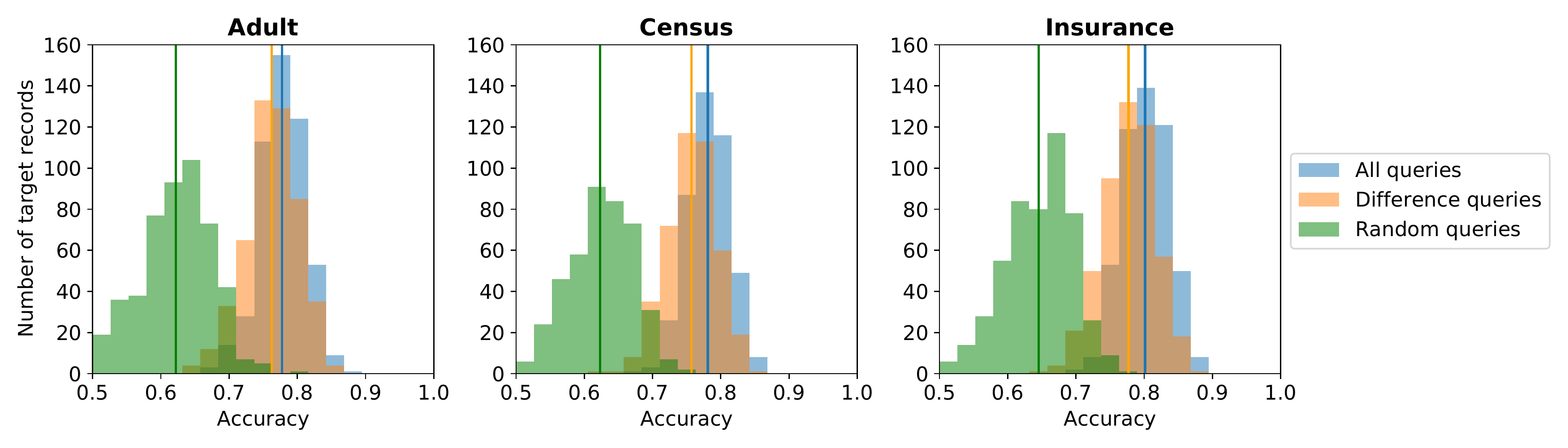}
\caption{\textbf{Diffix: Comparison between QuerySnout's performance in the AUXILIARY scenario using all the queries in the solution (blue), using only the difference queries (orange), or the same number of random queries in the solution (green).} 
For each dataset, the histogram aggregates the results from the 5 repetitions. The vertical bars show the mean accuracy.
}
\label{figure:diffix_analysis_difference_queries}
\end{figure*}

\begin{figure*}[htbp!]
\centering
\includegraphics[width=0.7\linewidth]{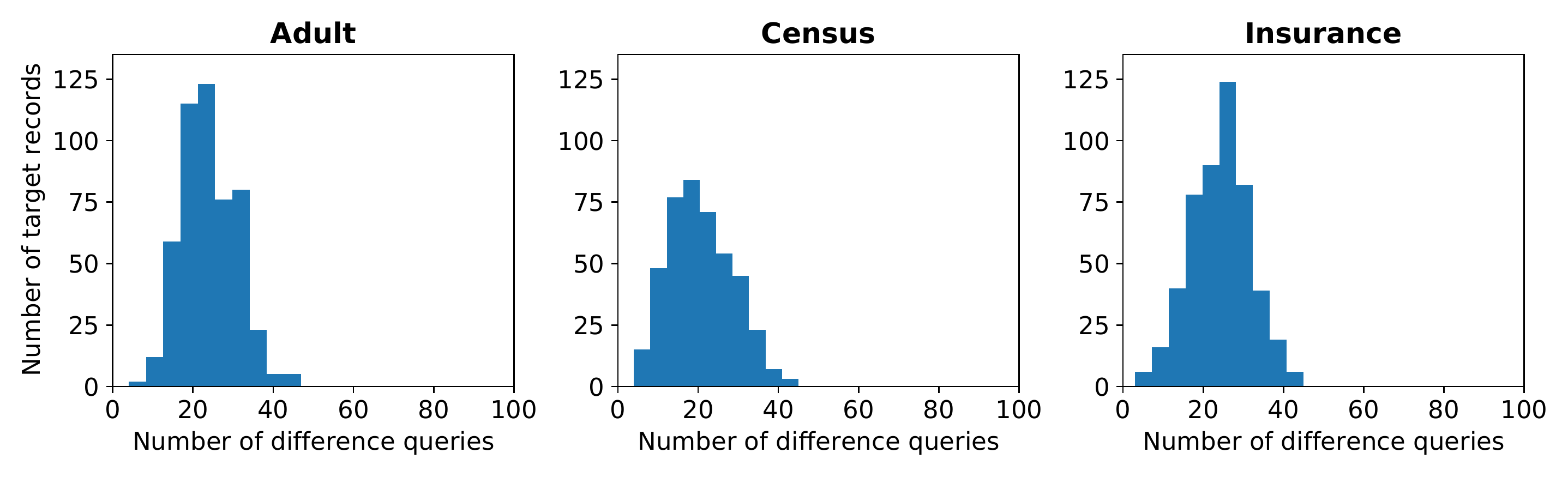}
\caption{\textbf{TableBuilder: Histogram of the number of difference queries in each solution found by QuerySnout.} 
For each dataset, the histogram aggregates the results from the 5 repetitions.
}
\label{figure:table_builder_analysis_number_difference_queries}
\end{figure*}

\begin{figure*}[htbp!]
\centering
\includegraphics[width=0.78\linewidth]{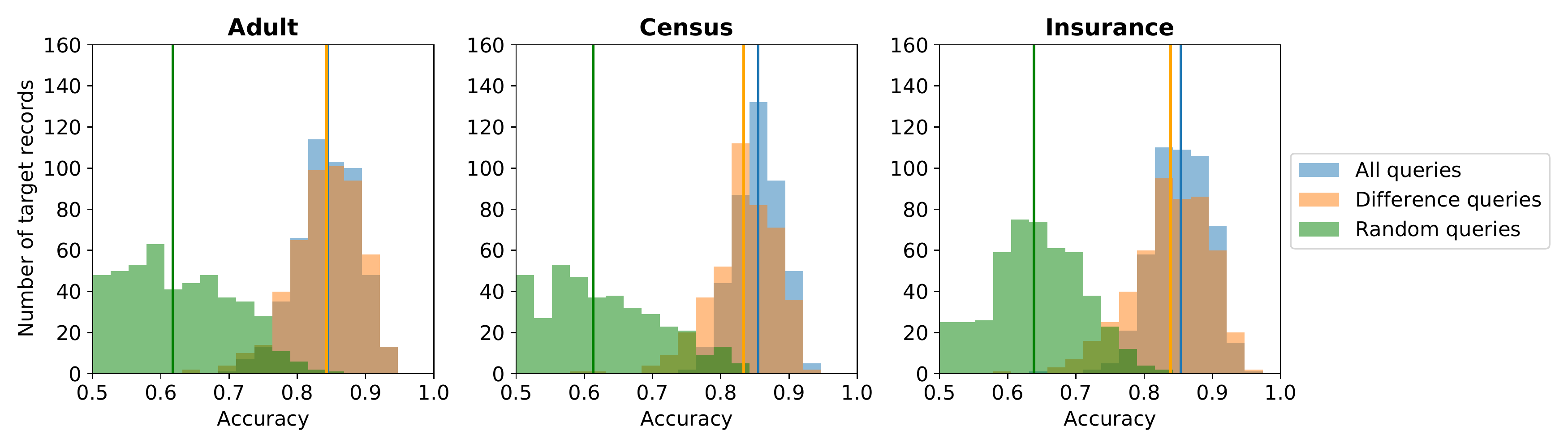}
\caption{\textbf{TableBuilder: Comparison between QuerySnout's performance in the AUXILIARY scenario using all the queries in the solution (blue), using only the difference queries (orange), or the same number of random queries in the solution (green).} 
For each dataset, the histogram aggregates the results from the 5 repetitions. The vertical bars show the mean accuracy.
}
\label{figure:table_builder_analysis_difference_queries}
\end{figure*}

\newpage

\subsection{SimpleQBS: Comparison with manual attacks on other datasets}
\label{appendix:simpleqbs}

Fig.~\ref{figure:qbs-simple_census-insurance} shows the results for the comparison with manual attacks in the AUXILIARY and EXACT-BUT-ONE scenarios for SimpleQBS($\tau, \sigma$) on the Insurance and Census datasets, for $\tau \in \{ 0, 1, 2, 3, 4\}$ and $\sigma \in \{ 0, 1, 2, 3, 4\}$. 

\begin{figure}[htbp!]
\centering
\subfigure[Census]{
\includegraphics[width=\linewidth]{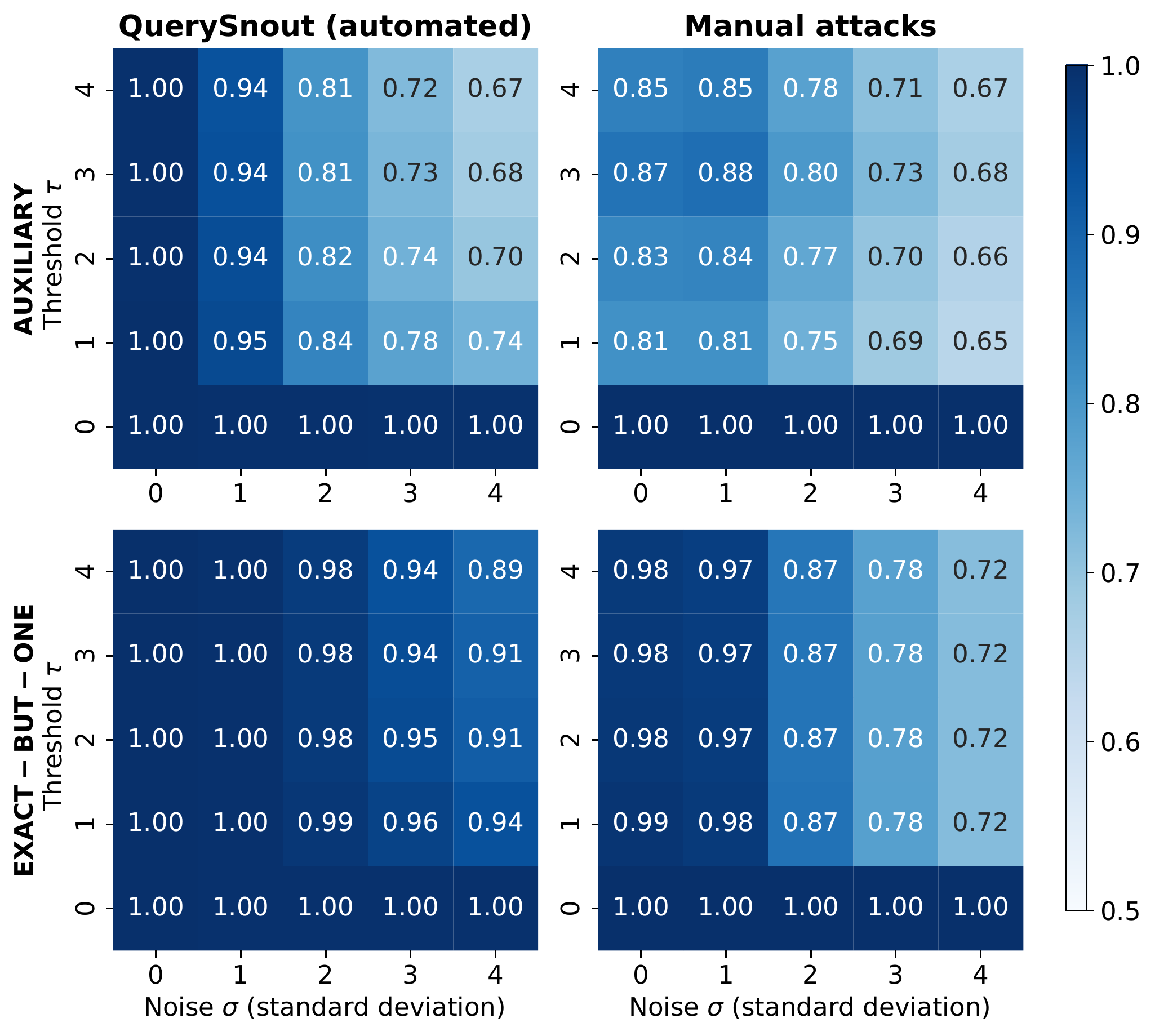}
}
\subfigure[Insurance]{\includegraphics[width=\linewidth]{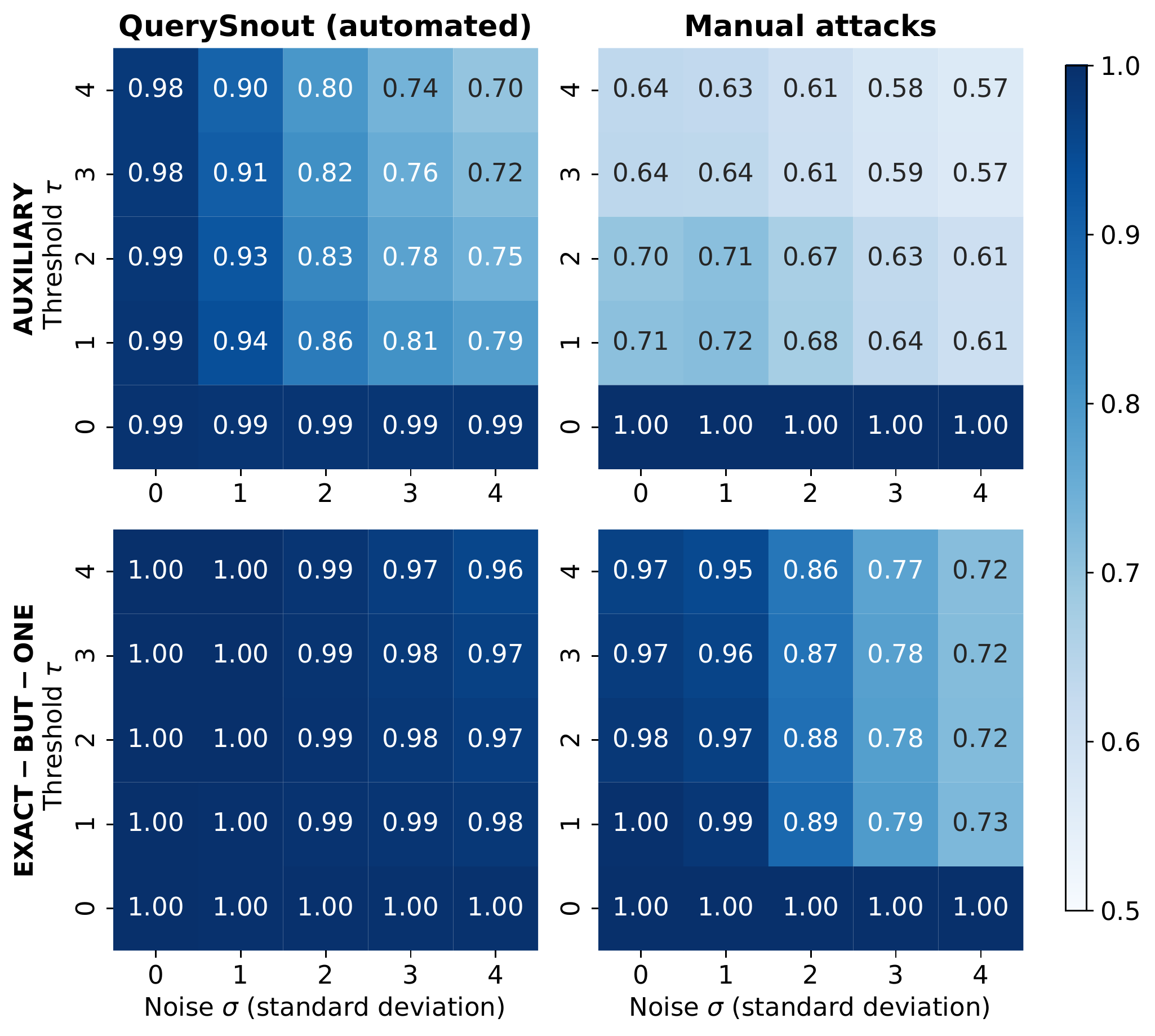}
}
\caption{\textbf{Comparison between automated and manual attacks against SimpleQBS($\tau, \sigma$) on the Census  and Insurance datasets.} 
We report the accuracy of attribute inference attacks discovered by QuerySnout (left) and of manual attacks tailored to each system (right) in the AUXILIARY (top) and EXACT-BUT-ONE (bottom) scenarios. The accuracy is averaged over 5 repetitions. 
}
\label{figure:qbs-simple_census-insurance}
\end{figure}

\begin{table*}[htbp!]
\centering
\caption{Comparison between QuerySnout and the random search approach on Diffix, TableBuilder,  SimpleQBS($\tau=4,\sigma=3$), and SimpleQBS($\tau=3,\sigma=4$). We report the difference $\Delta$ (in percentage points) between the accuracy of QuerySnout and the accuracy of the random search. We report the mean and standard deviation over 5 repetitions, for both the AUXILIARY (abbr. \textbf{AUX}) and EXACT-BUT-ONE (abbr. \textbf{EB1}) scenarios.}
\begin{tabular}{|l|ccc|ccc|ccc|ccc|}
\cline{2-13}
\multicolumn{1}{c|}{}
 & \multicolumn{3}{c|}{\textbf{Diffix}} & \multicolumn{3}{c|}{\textbf{TableBuilder}}  & \multicolumn{3}{c|}{\textbf{SimpleQBS($\tau=4,\sigma=3$)}} & \multicolumn{3}{c|}{\textbf{SimpleQBS($\tau=3,\sigma=4$)}} \\
 \hline
\textbf{AUX} & Adult & Census & Insurance & Adult & Census & Insurance & Adult & Census & Insurance & Adult & Census & Insurance \\
\hline
$\Delta$ & 6.8 (0.4) & 5.6 (0.8) & 6.5 (0.7) &  6.0 (0.3) & 4.1 (1.1) & 5.9 (1.5) & 6.4 (0.4) & 7.0 (0.4) & 6.7 (0.3) & 5.1 (0.4) & 4.9 (0.9) & 6.2 (0.4) \\
\hline
\textbf{EB1} & Adult & Census & Insurance & Adult & Census & Insurance & Adult & Census & Insurance & Adult & Census & Insurance \\
\hline
$\Delta$ & 4.8 (0.3) & 5.0 (0.3) & 4.7 (0.3) & 2.9 (1.0) & 3.3 (0.8) & 1.3 (0.5) & 7.5 (2.5) & 10.3 (2.2) & 4.1 (1.7) & 6.7 (1.7) & 8.4 (2.0) & 3.4 (1.8)\\
\hline
\end{tabular}\label{table:random-search}
\end{table*}

\begin{table*}[htbp!]
\centering
\caption{Comparison between QuerySnout and the random search approach on DPLaplace for different values of $\varepsilon$. We compute the absolute difference $\Delta$ (in percentage points) between the accuracy of QuerySnout and the accuracy of the random search. QuerySnout vastly outperforms the random search. We report the mean and standard deviation over 5 repetitions.}
\begin{tabular}{|l|ccc|ccc|ccc|}
\cline{2-10}
\multicolumn{1}{c|}{}
 & \multicolumn{3}{c|}{\textbf{DPLaplace($\varepsilon=1$)}} & \multicolumn{3}{c|}{\textbf{DPLaplace($\varepsilon=5$)}}  & \multicolumn{3}{c|}{\textbf{DPLaplace($\varepsilon=10$)}}  \\
 \hline
\textbf{AUXILIARY} & Adult & Census & Insurance & Adult & Census & Insurance & Adult & Census & Insurance \\
\hline
Absolute difference $\Delta$ & 10.8 (0.6) & 11.9 (0.7) & 11.4 (0.9) & 26.0 (0.7) & 25.8 (0.8) & 26.5 (0.4) & 17.7 (0.5) & 17.6 (0.8) & 17.9 (0.6)  \\
\hline
\textbf{EXACT-BUT-ONE} & Adult & Census & Insurance & Adult & Census & Insurance & Adult & Census & Insurance \\
\hline
Absolute difference $\Delta$ & 11.5 (0.8) & 12.2 (0.6) & 12.9 (1.0) & 29.1 (1.0) & 28.6 (1.1) & 29.3 (1.2) & 16.1 (0.5) & 15.6 (0.3) & 15.8 (0.5) \\
\hline
\end{tabular}\label{table:dp-laplace-random-search}
\end{table*}


\begin{algorithm}[htbp!]
\caption{\textsc{RandomSearch}}
\label{alg:random-search}
    \begin{algorithmic}[1]
        \Inputs{
        $P$: Population size (number of solutions). \\
        $m$: Number of queries in a solution. \\
        $N$: Number of generations. \\
        $D_{\text{aux}}$: Auxiliary datasets used to evaluate the fitness.\\
        }
        \Output{
        $best\in\mathcal{S}^m$: the best solution found in the search.
        }
        \Initialize{
        $best \gets None$\\
        $best\_fitness \gets - \infty$\\
        }
        \comm{Iterate over $N$ generations.}
        \For{$g = 1$ to $N$}
            \State{\comm{Generate a random population.}}\
            \For{$i = 1$ to $P$}
                \State{$solution \gets random\_solution(m)$}
                \State{$fit \gets \textsc{EvaluateFitness}(solution, D_{\text{aux}},$ $U_{\text{aux}})$}
                \If{$fit > best\_fitness$}
                    \State{$best\_fitness \gets fit$}
                    \State{$best \gets solution$}
                \EndIf
            \EndFor
        \EndFor
    \end{algorithmic}
\end{algorithm}

\subsection{Random search}
\label{appendix:random-search}

The random search procedure is a simple evolutionary algorithm where the population is entirely renewed at each generation, and the best solution found in all generations is returned at the end.
This is a baseline for evolutionary algorithms that allows to evaluate the gain in performance from the mutation operations.
We present the pseudocode for the random search procedure in Algorithm~\ref{alg:random-search} and compare its performance with  QuerySnout in Tables~\ref{table:random-search} and ~\ref{table:dp-laplace-random-search}. We refer the reader to the main paper (end of Sec.~\ref{subsec:dp-laplace-results}) for an analysis of the results.

\subsection{Impact of the dataset size}\label{appendix:scaling-dataset_size}

\begin{figure}[htbp!]
\centering
\includegraphics[width=\linewidth]{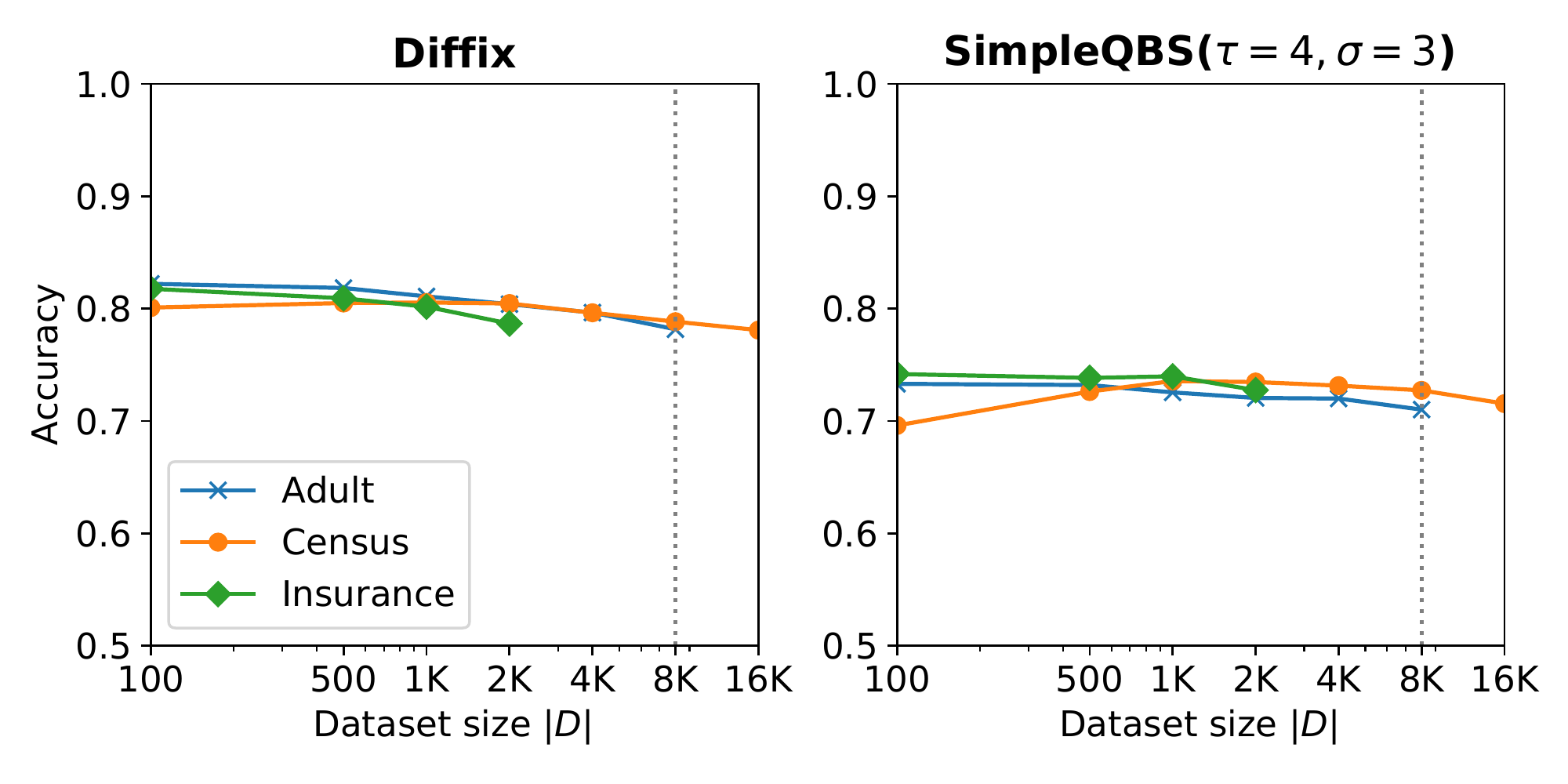}
\caption{\textbf{Impact of the dataset size on the performance of QuerySnout.} 
We show the accuracy of an attribute inference attack discovered by QuerySnout in the AUXILIARY scenario, for varying dataset sizes. The dotted vertical line represents the value $8000$ which we use for experiments on Adult and Census.
}
\label{figure:scaling-dataset-size}
\end{figure}

Fig.~\ref{figure:scaling-dataset-size} shows that the accuracy of attacks discovered by QuerySnout slowly increases as the dataset $|D|$ decreases (in the AUXILIARY scenario). This can partly explain why our attack performs better on Insurance ($|D| = 1000$) than on Adult ($|D| = 8000$). We believe that this might be due to partitions of the original dataset being of limited size (e.g., slightly larger than 16K for Adult), so that large datasets sampled from, e.g., $\traindist$ tend to overlap more and hence generalize less well. In our experiments, we found the validation accuracy to remain the same or slightly decrease with $|D|$, while the gap between the validation and test accuracy slighly increases. Note that these results were conducted using only one seed, for computational reasons. We leave the thorough investigation of reasons why this occurs for future work.

\subsection{Impact of search parameters}
\label{appendix:impact-other-parameters}

Fig.~\ref{figure:scaling-search-parameters} shows the impact of different search parameters on the accuracy of QuerySnout, as we vary them one at a time. We also study the impact of one-by-one removal of the different mutation components on different systems and report results in Table~\ref{table:ablation}. We refer the reader to Sec.~\ref{subsec:impact-search-parameters} for an analysis of the results.

\begin{figure*}[htbp!]
\centering

\subfigure{%
\includegraphics[width=0.3\textwidth]{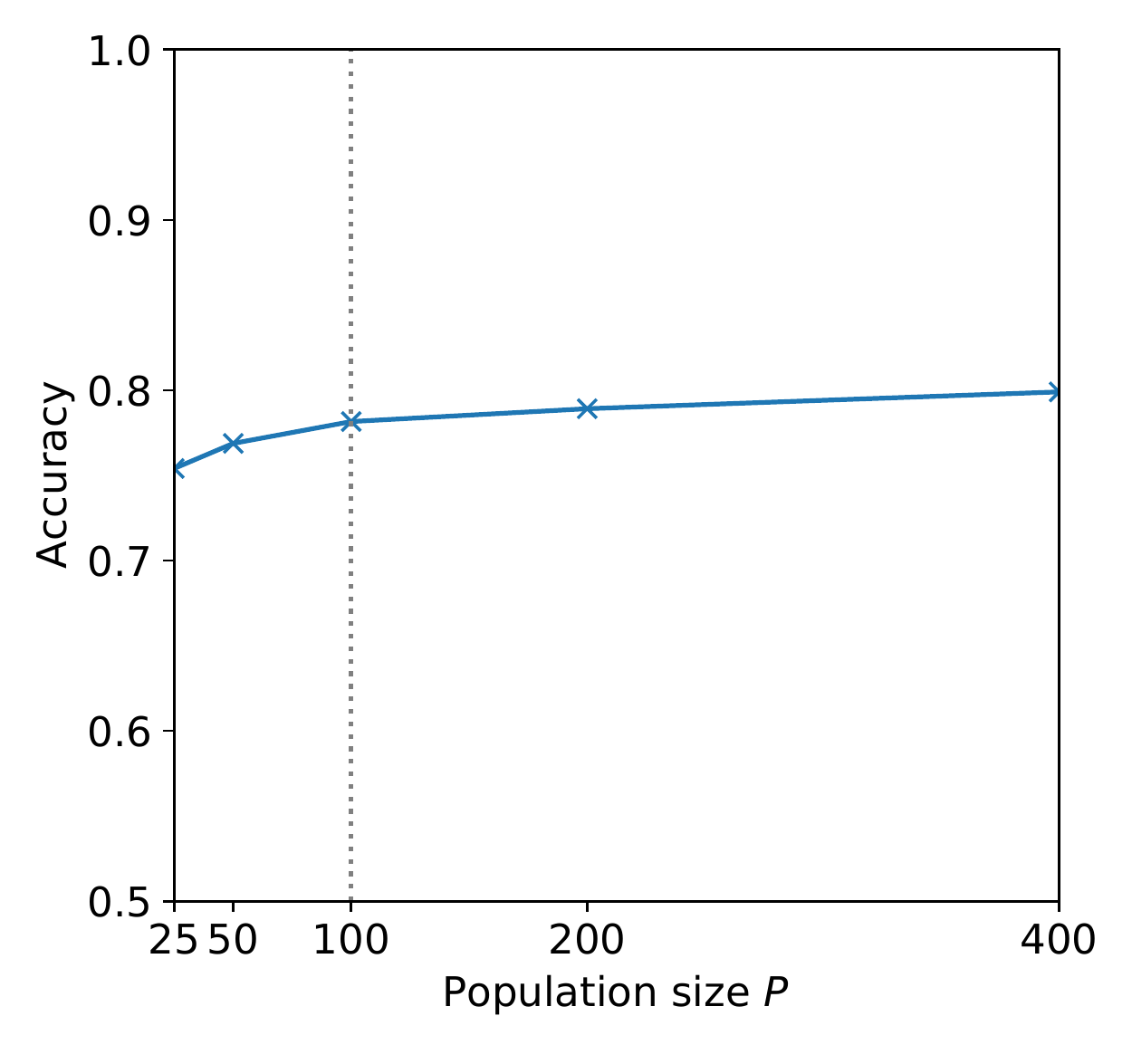}%
}\hfil
\subfigure{%
\includegraphics[width=0.3\textwidth]{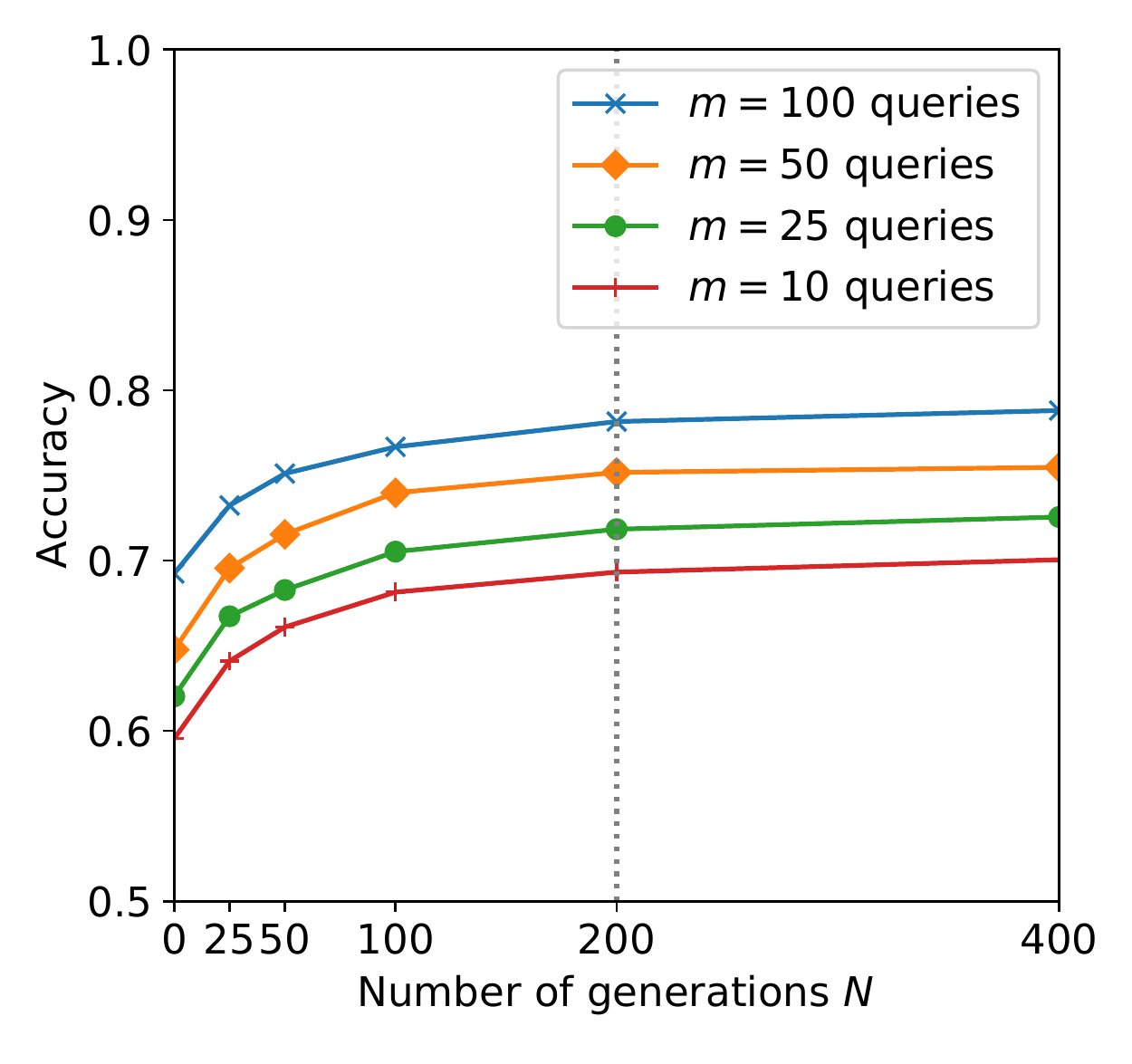}%
}\hfil
\subfigure{%
\includegraphics[width=0.3\textwidth]{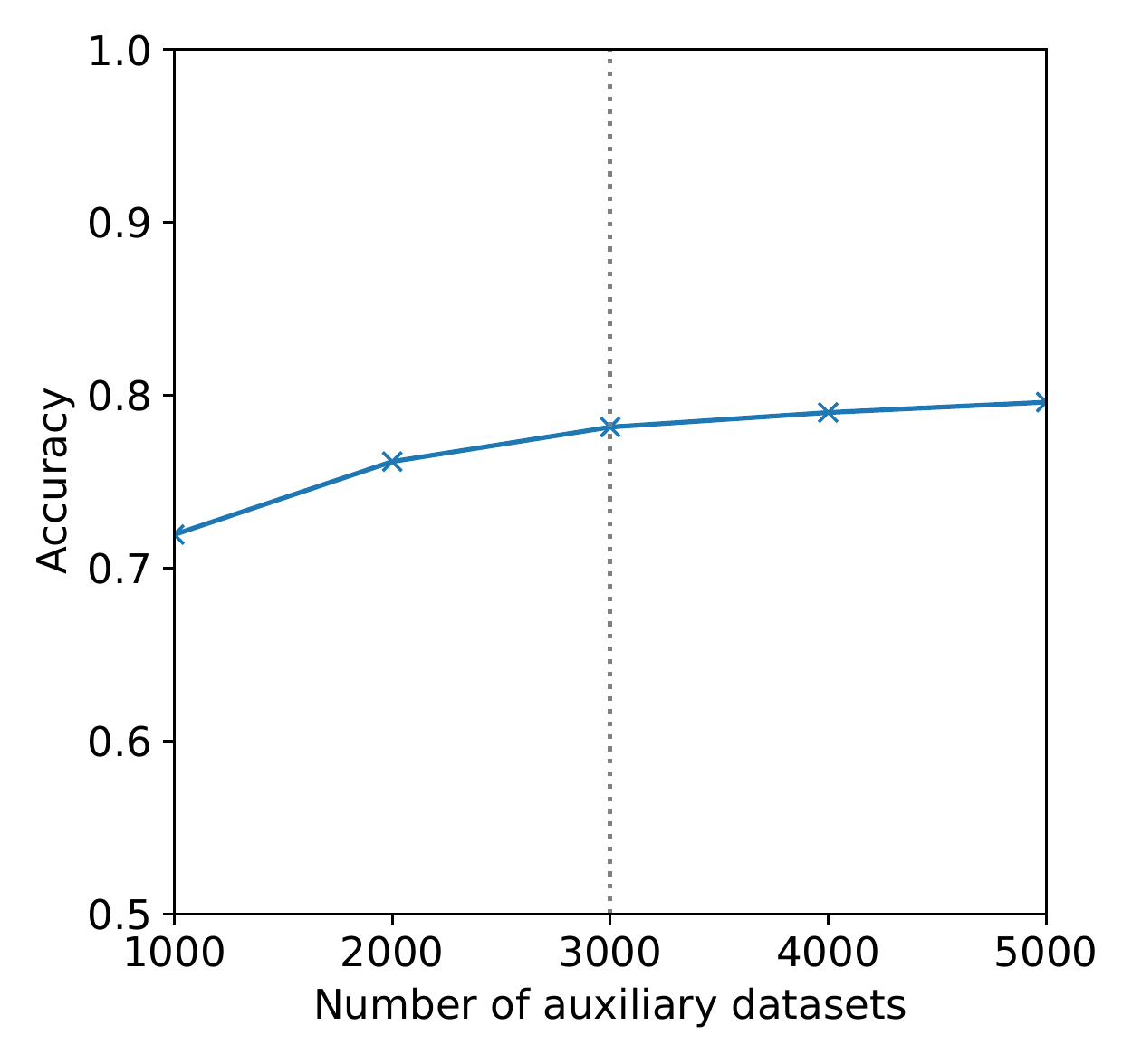}%
}

\caption{\textbf{Impact of various search parameters on the accuracy of QuerySnout.}
The results are computed on the Adult dataset in the AUXILIARY scenario. The vertical gray line indicates the default value used in the main paper experiments.
\vspace{0.5cm}
}

\label{figure:scaling-search-parameters}
\end{figure*}

\begin{table*}[htbp!]
\centering
\caption{\textbf{Ablation experiment in the AUXILIARY scenario.} We set the various probabilities controlling the evolutionary search to zero, one-by-one, keeping all the others equal to the default value. We find that the search is robust and that ``copy'' has the largest impact. We report the mean and standard deviation over 5 repetitions.}
\begin{tabular}{|l|cc|ccc|}
\cline{2-6}
\multicolumn{1}{c|}{}
 & \multicolumn{2}{c|}{\textbf{Deterministic systems}} & \multicolumn{3}{c|}{\textbf{Non-deterministic systems}}  \\
 \cline{2-6}
\multicolumn{1}{c|}{}
 & \textbf{Diffix}  & \textbf{TableBuilder} & \textbf{SimpleQBS($\tau=4,\sigma=3$)} & \textbf{DPLaplace($\varepsilon=1$)} & \textbf{DPLaplace($\varepsilon=10$)} \\
 \hline
Default values & 77.8 (0.5) & 84.5 (0.6) & 70.7 (1.1) & 62.9 (0.8) & 98.6 (0.4) \\
\hline
\hline
$p_{\text{copy}}=0$ & 76.1 (0.5) & 82.9 (0.3) & 67.2 (1.3) & 52.1 (0.3) & 93.0 (0.3) \\
\hline
\hline
$p_{\text{change}}=0$ & 76.1 (0.7) & 83.5 (0.5) & 70.5 (1.2) & 65.5 (1.0) & 98.5 (0.4) \\
\hline
$p_{\text{swap}}=0$ & 77.7 (0.6) & 84.5 (0.6) & 71.0 (1.1) & 63.4 (0.8) & 98.5 (0.4) \\
\hline
\end{tabular}\label{table:ablation}
\end{table*}

\subsection{Monotonicity of accuracy}
\label{appendix:monotonicity}

In section~\ref{sec:budget-based-mechanisms}, we propose a heuristic to apply QuerySnout to budget-based mechanisms, where a query repeated $k$ times is performed once with partial budget $\frac{k}{m}$.
This heuristic assumes assume \textit{monotonicity of accuracy} : repeating a query $k$ times with budget $\frac{1}{k}$ returns a less precise answer than answering the query one time with the full budget. Formally, $R$ satisfies monotonicity of accuracy if and only if for all $q\in \mathcal{Q}$, $d \in \mathcal{D}$, and $k \in \mathbb{N}, k>1$ (with $error:r\mapsto (r - T(d,q))^2$):
\[
\mathbb{E}\left[error\left(R(d,q,1)\right)\right]
\leq
\mathbb{E}\left[error\left(\sum_{i=1}^k \frac{R(d,q,1/k)}{k}
\right)\right]
\]
Note that this is a desirable property of QBS based on budget: if a QBS is non-monotonous, analysts are incentivized to repeat the same query as many times with as small a budget as possible, rather than perform queries directly, which is impractical.

We then show that the DPLaplace mechanism satisfies the monotonicity of accuracy property.

\begin{theorem}
The $R_\text{DPLaplace}$ be the DPLaplace mechanism as defined in section~\ref{subsec:dplaplace} satisfies monotonicity of accuracy for all budgets $\varepsilon>0$, i.e., for all $q\in \mathbb{Q}$, $d \in \mathcal{D}$, and $k \in \mathbb{N}, k>2$:
\[
\mathbb{E}\left[e\left(R_{\text{DPLaplace}(\varepsilon)}(d,q,1)\right)\right]
\leq
\mathbb{E}\left[e\left(\sum_{i=1}^k \frac{R_{\text{DPLaplace}(\varepsilon)}(d,q,1/k)}{k}
\right)\right]
\]
where $e:r\mapsto (r - T(d,q))^2$ is the L$_2$ error.
\end{theorem}
\begin{proof}
Define $L = R_{\text{DPLaplace}(\varepsilon)}(d,q,1) - T(d,q)$. By definition, observe that $L \sim Lap(\frac{1}{\varepsilon})$. By definition, we have that:
\[
\mathbb{E}\left[e\left(R_{\text{DPLaplace}(\varepsilon)}(d,q,1)\right)\right] = \mathbb{E}\left[L^2\right] = \frac{2}{\varepsilon^2}
\]
Similarly, define $L_i = R_{\text{DPLaplace}(\varepsilon)}(d,q,1/k) - T(d,q)$ for $i=1,\dots,k$, and observe that $L_i \sim Lap(\frac{k}{\varepsilon})$. We have that:
\[
\mathbb{E}\left[e\left(\sum_{i=1}^k \frac{R_{\text{DPLaplace}(\varepsilon)}(d,q,1/k)}{k}
\right)\right] = \mathbb{E}\left[\left(\frac{\sum_{i=1}^k L_i}{k}\right)^2\right]
\]
By linearity of the expectation:
\[
\mathbb{E}\left[\left(\frac{\sum_{i=1}^k L_i}{k}\right)^2\right] = \frac{1}{k^2} \cdot \sum_{i=1}^k \mathbb{E}[L_i^2] = \frac{1}{k}\frac{2 k^2}{\varepsilon^2} = \frac{2 k}{\varepsilon^2} > \frac{2}{\varepsilon^2},
\]
which concludes the proof.
\end{proof}

\subsection{Optimality of the uniqueness attack against DPLaplace}
\label{appendix:dplaplace}
Since DPLaplace does not use query-set size restriction, we propose a simple \textit{uniqueness} attack that directly targets the user's record:
\[
q = \text{COUNT}\left(\land_{a \in \attrset'}(a = r^u_a) \land (a_n = 1)\right).
\]
This query is such that it reveals the target user's secret attribute, with some added noise:
\[
R := R(d,q^*,p) = I\{r^{u}_n = 1\} + L,~~L \sim Lap(\frac{1}{p\varepsilon}).
\]
Because the QBS has monotonicity of accuracy (see Appendix~\ref{appendix:monotonicity}), the best choice of $p$ possible is $p=1$, i.e. using all the budget for this query.
Hence, $R\sim Lap(0, \varepsilon^{-1})$ if the user's target record is $0$, and $R \sim Lap(1, \varepsilon^{-1})$ otherwise.We can thus use a likelihood ratio test to determine the target record's value. In the continuous case, this gives the rule:
\[
\widehat{r^{u}_n = 1} \Longleftrightarrow \frac{pdf_{Lap(1,\varepsilon^{-1})}(R)}{pdf_{Lap(0,\varepsilon^{-1})}(R)} \geq 1 \Longleftrightarrow R \geq \frac{1}{2},
\]
For the rounded and thresholded at zero QBS that we report results on, the same rule applies. Indeed, let $P_{\mathcal{L}}(n)$ the probability of sampling $n \in \mathbb{N}$ from rounding and thresholding of a sample of $L \sim \mathcal{L}$. By developing the pdf of the Laplace distribution and integrating over bins, we get:
\begin{itemize}

    \item $P_{Lap(0,\varepsilon^{-1})}(0) = \frac{1}{2} + \int_0^{\frac{1}{2}} \frac{\varepsilon}{2} e^{ - x \cdot \varepsilon} dx > \frac{1}{2}$,

    \item $P_{Lap(1,\varepsilon^{-1})}(0) = \int_{-\infty}^{\frac{1}{2}} \frac{\varepsilon}{2} e^{ (x-1) \cdot \varepsilon} dx < \frac{1}{2}$,

    \item $\forall r \in \mathbb{N}_0:~P_{Lap(0,\varepsilon^{-1})}(r) = \int_{r - \frac{1}{2}}^{r + \frac{1}{2}} \frac{\varepsilon}{2} e^{ -x \cdot \varepsilon} dx$,

    \item $
    \begin{array}{ll}
        P_{Lap(1,\varepsilon^{-1})}(1) & = \int_{\frac{1}{2}}^{1} \frac{\varepsilon}{2} e^{ (x-1) \cdot \varepsilon} dx + \int_{1}^{\frac{3}{2}} \frac{\varepsilon}{2} e^{ (1-x) \cdot \varepsilon} dx \\
         & = P_{Lap(0,\varepsilon^{-1})}(1) + f(\varepsilon)
    \end{array}
 $, where $f(\varepsilon) = 1 -\frac{3}{2}e^{\frac{\varepsilon}{2}} + \frac{1}{2}e^{\frac{3\varepsilon}{2}} > 0,~~\forall \varepsilon>0$ (since $f(0) = 0$ and $f'(\varepsilon) > 0$ for $\varepsilon>0$),

    \item $\forall r \in \mathbb{N}\setminus\{0,1\}:~P_{Lap(1,\varepsilon^{-1})}(r) = \int_{r - \frac{1}{2}}^{r + \frac{1}{2}} \frac{\varepsilon}{2} e^{ (1 - x) \cdot \varepsilon} dx = e^\varepsilon P_{Lap(0,\varepsilon^-1)}(r) > P_{Lap(0,\varepsilon^-1)}(r)$.

\end{itemize}
We thus observe that for all $n \in \mathbb{N}$, $P_{Lap(0,\varepsilon^{-1})}(n) \geq P_{Lap(1,\varepsilon^{-1})}(n)$ iff $n = 0$.
These tests have maximum statistical power, thanks to the Neyman-Pearson lemma~\cite{neyman1933ix}.

\clearpage
\subsection{Example of solutions}
\label{appendix:solutions-found}
We report the best solution discovered by our attack against various QBSes in run \#1 on the Adult dataset in the AUXILIARY scenario. In each figure, we select a random target record and display each unique query in the corresponding solution, line by line, along with the number of times this query is repeated in the solution (in the rightmost column).
We write each condition $a_i \; c_i \; v_i$ (with $v_i = r^u_{a_i}$) in a query as \textcolor{pos}{\textsf{\textbf{ai = vi}}} for $c_i=$``='', \textcolor{neg}{\textsf{\textbf{ai != vi}}} for $c_i=$``$\neq$'', and no condition for $c_i = \perp$.
We use $s$(ensitive) to denote the sensitive attribute and write conditions on that attribute as \textcolor{pos}{\textsf{\textbf{s = 1}}} ($a_n \neq 0$) and \textcolor{neg}{\textsf{\textbf{s = 0}}} ($a_n = 0$). 
For readability, we write conditions for attribute $a_i$ in column $i$, and leave the column blank if no condition is put on the attribute.
We refer to a conjunction between conditions using the character $\land$. For readability purposes, we leave an empty space whenever there is no condition on a given attribute.

\textbf{Diffix and TableBuilder. }
We report the attacks found for Diffix and TableBuilder.
As expected, both solutions contain difference queries (queries \#1-\#36 for Diffix, and \#1-\#12 for TableBuilder).
This is in line with prior work, which showed that both Diffix~\cite{gadotti2019signal} and TableBuilder~\cite{chipperfield2016,rinott2018} can be attacked with queries similar to difference attacks.
We also note that the optimal solutions contain few repeated queries, which is expected, since both systems are deterministic.

\newpage
\begin{framed}
\noindent\textbf{Example of solution found against Diffix in the AUXILIARY scenario.}
\[
{\customsize
\begin{array}{llllllll}
 \multicolumn{3}{l}{\textbf{Difference queries}} \\
 (1) &  & \land~ & \land~ & \land~\textcolor{pos}{a_{4} = v_{4}} & \land~ & \land~\textcolor{neg}{s=0} & \\
 (2) &  & \land~ & \land~\textcolor{neg}{a_{3} \neq v_{3}} & \land~\textcolor{pos}{a_{4} = v_{4}} & \land~ & \land~\textcolor{neg}{s=0} & \\
 (3) &  & \land~ & \land~\textcolor{pos}{a_{3} = v_{3}} & \land~ & \land~ & \land~\textcolor{neg}{s=0} & \\
 (4) &  & \land~\textcolor{neg}{a_{2} \neq v_{2}} & \land~\textcolor{pos}{a_{3} = v_{3}} & \land~ & \land~ & \land~\textcolor{neg}{s=0} & \\
 (5) & \textcolor{pos}{a_{1} = v_{1}} & \land~ & \land~ & \land~ & \land~ & \land~\textcolor{neg}{s=0} & \\
 (6) & \textcolor{pos}{a_{1} = v_{1}} & \land~ & \land~\textcolor{neg}{a_{3} \neq v_{3}} & \land~ & \land~ & \land~\textcolor{neg}{s=0} & \\
 (7) & \textcolor{pos}{a_{1} = v_{1}} & \land~ & \land~ & \land~\textcolor{neg}{a_{4} \neq v_{4}} & \land~ & \land~\textcolor{neg}{s=0} & \\
 (8) &  & \land~ & \land~\textcolor{pos}{a_{3} = v_{3}} & \land~ & \land~\textcolor{pos}{a_{5} = v_{5}} & \land~\textcolor{pos}{s=1} & \\
 (9) & \textcolor{neg}{a_{1} \neq v_{1}} & \land~ & \land~\textcolor{pos}{a_{3} = v_{3}} & \land~ & \land~\textcolor{pos}{a_{5} = v_{5}} & \land~\textcolor{pos}{s=1} & \\
 (10) &  & \land~ & \land~\textcolor{pos}{a_{3} = v_{3}} & \land~\textcolor{pos}{a_{4} = v_{4}} & \land~ & \land~\textcolor{neg}{s=0} & \\
 (11) & \textcolor{neg}{a_{1} \neq v_{1}} & \land~ & \land~\textcolor{pos}{a_{3} = v_{3}} & \land~\textcolor{pos}{a_{4} = v_{4}} & \land~ & \land~\textcolor{neg}{s=0} & \\
 (12) &  & \land~\textcolor{pos}{a_{2} = v_{2}} & \land~ & \land~ & \land~\textcolor{pos}{a_{5} = v_{5}} & \land~\textcolor{neg}{s=0} & \\
 (13) &  & \land~\textcolor{pos}{a_{2} = v_{2}} & \land~ & \land~\textcolor{neg}{a_{4} \neq v_{4}} & \land~\textcolor{pos}{a_{5} = v_{5}} & \land~\textcolor{neg}{s=0} & \\
 (14) & \textcolor{pos}{a_{1} = v_{1}} & \land~ & \land~ & \land~ & \land~\textcolor{pos}{a_{5} = v_{5}} & \land~\textcolor{neg}{s=0} & \\
 (15) & \textcolor{pos}{a_{1} = v_{1}} & \land~ & \land~\textcolor{neg}{a_{3} \neq v_{3}} & \land~ & \land~\textcolor{pos}{a_{5} = v_{5}} & \land~\textcolor{neg}{s=0} & \\
 (16) & \textcolor{pos}{a_{1} = v_{1}} & \land~ & \land~ & \land~\textcolor{pos}{a_{4} = v_{4}} & \land~ & \land~\textcolor{neg}{s=0} & \\
 (17) & \textcolor{pos}{a_{1} = v_{1}} & \land~ & \land~\textcolor{neg}{a_{3} \neq v_{3}} & \land~\textcolor{pos}{a_{4} = v_{4}} & \land~ & \land~\textcolor{neg}{s=0} & \mathbf{(\times 2)}\\
 (18) & \textcolor{pos}{a_{1} = v_{1}} & \land~ & \land~ & \land~\textcolor{pos}{a_{4} = v_{4}} & \land~ & \land~\textcolor{pos}{s=1} & \\
 (19) & \textcolor{pos}{a_{1} = v_{1}} & \land~ & \land~\textcolor{neg}{a_{3} \neq v_{3}} & \land~\textcolor{pos}{a_{4} = v_{4}} & \land~ & \land~\textcolor{pos}{s=1} & \\
 (20) & \textcolor{pos}{a_{1} = v_{1}} & \land~\textcolor{pos}{a_{2} = v_{2}} & \land~ & \land~ & \land~ & \land~\textcolor{neg}{s=0} & \\
 (21) & \textcolor{pos}{a_{1} = v_{1}} & \land~\textcolor{pos}{a_{2} = v_{2}} & \land~\textcolor{neg}{a_{3} \neq v_{3}} & \land~ & \land~ & \land~\textcolor{neg}{s=0} & \\
 (22) &  & \land~ & \land~\textcolor{pos}{a_{3} = v_{3}} & \land~\textcolor{pos}{a_{4} = v_{4}} & \land~\textcolor{pos}{a_{5} = v_{5}} & \land~\textcolor{pos}{s=1} & \\
 (23) & \textcolor{neg}{a_{1} \neq v_{1}} & \land~ & \land~\textcolor{pos}{a_{3} = v_{3}} & \land~\textcolor{pos}{a_{4} = v_{4}} & \land~\textcolor{pos}{a_{5} = v_{5}} & \land~\textcolor{pos}{s=1} & \\
 (24) &  & \land~\textcolor{pos}{a_{2} = v_{2}} & \land~ & \land~\textcolor{pos}{a_{4} = v_{4}} & \land~\textcolor{pos}{a_{5} = v_{5}} & \land~\textcolor{pos}{s=1} & \\
 (25) & \textcolor{neg}{a_{1} \neq v_{1}} & \land~\textcolor{pos}{a_{2} = v_{2}} & \land~ & \land~\textcolor{pos}{a_{4} = v_{4}} & \land~\textcolor{pos}{a_{5} = v_{5}} & \land~\textcolor{pos}{s=1} & \\
 (26) &  & \land~\textcolor{pos}{a_{2} = v_{2}} & \land~\textcolor{pos}{a_{3} = v_{3}} & \land~\textcolor{pos}{a_{4} = v_{4}} & \land~ & \land~\textcolor{pos}{s=1} & \\
 (27) & \textcolor{neg}{a_{1} \neq v_{1}} & \land~\textcolor{pos}{a_{2} = v_{2}} & \land~\textcolor{pos}{a_{3} = v_{3}} & \land~\textcolor{pos}{a_{4} = v_{4}} & \land~ & \land~\textcolor{pos}{s=1} & \\
 (28) & \textcolor{pos}{a_{1} = v_{1}} & \land~ & \land~ & \land~\textcolor{pos}{a_{4} = v_{4}} & \land~\textcolor{pos}{a_{5} = v_{5}} & \land~\textcolor{neg}{s=0} & \\
 (29) & \textcolor{pos}{a_{1} = v_{1}} & \land~ & \land~\textcolor{neg}{a_{3} \neq v_{3}} & \land~\textcolor{pos}{a_{4} = v_{4}} & \land~\textcolor{pos}{a_{5} = v_{5}} & \land~\textcolor{neg}{s=0} & \\
 (30) & \textcolor{pos}{a_{1} = v_{1}} & \land~ & \land~ & \land~\textcolor{pos}{a_{4} = v_{4}} & \land~\textcolor{pos}{a_{5} = v_{5}} & \land~\textcolor{pos}{s=1} & \\
 (31) & \textcolor{pos}{a_{1} = v_{1}} & \land~\textcolor{neg}{a_{2} \neq v_{2}} & \land~ & \land~\textcolor{pos}{a_{4} = v_{4}} & \land~\textcolor{pos}{a_{5} = v_{5}} & \land~\textcolor{pos}{s=1} & \\
 (32) & \textcolor{pos}{a_{1} = v_{1}} & \land~ & \land~\textcolor{neg}{a_{3} \neq v_{3}} & \land~\textcolor{pos}{a_{4} = v_{4}} & \land~\textcolor{pos}{a_{5} = v_{5}} & \land~\textcolor{pos}{s=1} & \\
 (33) & \textcolor{pos}{a_{1} = v_{1}} & \land~\textcolor{pos}{a_{2} = v_{2}} & \land~ & \land~ & \land~\textcolor{pos}{a_{5} = v_{5}} & \land~\textcolor{pos}{s=1} & \\
 (34) & \textcolor{pos}{a_{1} = v_{1}} & \land~\textcolor{pos}{a_{2} = v_{2}} & \land~\textcolor{neg}{a_{3} \neq v_{3}} & \land~ & \land~\textcolor{pos}{a_{5} = v_{5}} & \land~\textcolor{pos}{s=1} & \\
 (35) & \textcolor{pos}{a_{1} = v_{1}} & \land~\textcolor{pos}{a_{2} = v_{2}} & \land~\textcolor{pos}{a_{3} = v_{3}} & \land~ & \land~ & \land~\textcolor{pos}{s=1} & \\
 (36) & \textcolor{pos}{a_{1} = v_{1}} & \land~\textcolor{pos}{a_{2} = v_{2}} & \land~\textcolor{pos}{a_{3} = v_{3}} & \land~ & \land~\textcolor{neg}{a_{5} \neq v_{5}} & \land~\textcolor{pos}{s=1} & \vspace{\customspace}\\
 \multicolumn{3}{l}{\textbf{Other queries}} \\
 (37) &  & \land~ & \land~\textcolor{neg}{a_{3} \neq v_{3}} & \land~ & \land~ & \\
 (38) &  & \land~ & \land~ & \land~ & \land~\textcolor{neg}{a_{5} \neq v_{5}} & \\
 (39) & \textcolor{pos}{a_{1} = v_{1}} & \land~ & \land~ & \land~ & \land~ & \\
 (40) &  & \land~ & \land~\textcolor{neg}{a_{3} \neq v_{3}} & \land~\textcolor{pos}{a_{4} = v_{4}} & \land~ & \\
 (41) & \textcolor{pos}{a_{1} = v_{1}} & \land~ & \land~ & \land~\textcolor{pos}{a_{4} = v_{4}} & \land~ & \\
 (42) & \textcolor{neg}{a_{1} \neq v_{1}} & \land~ & \land~\textcolor{neg}{a_{3} \neq v_{3}} & \land~\textcolor{neg}{a_{4} \neq v_{4}} & \land~ &  & \mathbf{(\times 2)}\\
 (43) & \textcolor{neg}{a_{1} \neq v_{1}} & \land~ & \land~ & \land~\textcolor{pos}{a_{4} = v_{4}} & \land~\textcolor{neg}{a_{5} \neq v_{5}} & \\
 (44) & \textcolor{neg}{a_{1} \neq v_{1}} & \land~ & \land~ & \land~\textcolor{pos}{a_{4} = v_{4}} & \land~ & \land~\textcolor{pos}{s=1} & \\
 (45) & \textcolor{neg}{a_{1} \neq v_{1}} & \land~ & \land~\textcolor{pos}{a_{3} = v_{3}} & \land~ & \land~ & \land~\textcolor{pos}{s=1} & \\
 (46) & \textcolor{neg}{a_{1} \neq v_{1}} & \land~\textcolor{pos}{a_{2} = v_{2}} & \land~ & \land~\textcolor{neg}{a_{4} \neq v_{4}} & \land~ & \\
 (47) & \textcolor{neg}{a_{1} \neq v_{1}} & \land~\textcolor{pos}{a_{2} = v_{2}} & \land~ & \land~ & \land~ & \land~\textcolor{pos}{s=1} & \\
 (48) &  & \land~\textcolor{neg}{a_{2} \neq v_{2}} & \land~ & \land~\textcolor{pos}{a_{4} = v_{4}} & \land~\textcolor{neg}{a_{5} \neq v_{5}} & & \mathbf{(\times 2)}\\
 (49) &  & \land~\textcolor{neg}{a_{2} \neq v_{2}} & \land~ & \land~\textcolor{pos}{a_{4} = v_{4}} & \land~ & \land~\textcolor{pos}{s=1} & \\
 (50) &  & \land~ & \land~\textcolor{neg}{a_{3} \neq v_{3}} & \land~ & \land~\textcolor{neg}{a_{5} \neq v_{5}} & \land~\textcolor{neg}{s=0} & \\
 (51) &  & \land~ & \land~\textcolor{neg}{a_{3} \neq v_{3}} & \land~\textcolor{pos}{a_{4} = v_{4}} & \land~ & \land~\textcolor{pos}{s=1} & \\
 (52) &  & \land~ & \land~\textcolor{pos}{a_{3} = v_{3}} & \land~\textcolor{pos}{a_{4} = v_{4}} & \land~ & \land~\textcolor{pos}{s=1} & \\
 (53) &  & \land~\textcolor{pos}{a_{2} = v_{2}} & \land~\textcolor{pos}{a_{3} = v_{3}} & \land~\textcolor{pos}{a_{4} = v_{4}} & \land~ & \\
 (54) & \textcolor{pos}{a_{1} = v_{1}} & \land~ & \land~ & \land~\textcolor{neg}{a_{4} \neq v_{4}} & \land~ & \land~\textcolor{pos}{s=1} & \\
 (55) & \textcolor{pos}{a_{1} = v_{1}} & \land~ & \land~ & \land~\textcolor{neg}{a_{4} \neq v_{4}} & \land~\textcolor{pos}{a_{5} = v_{5}} & \\
 (56) & \textcolor{pos}{a_{1} = v_{1}} & \land~ & \land~\textcolor{pos}{a_{3} = v_{3}} & \land~\textcolor{pos}{a_{4} = v_{4}} & \land~ & \\
 (57) & \textcolor{neg}{a_{1} \neq v_{1}} & \land~\textcolor{neg}{a_{2} \neq v_{2}} & \land~\textcolor{pos}{a_{3} = v_{3}} & \land~ & \land~ & \land~\textcolor{neg}{s=0} & \\
 (58) & \textcolor{neg}{a_{1} \neq v_{1}} & \land~\textcolor{neg}{a_{2} \neq v_{2}} & \land~\textcolor{pos}{a_{3} = v_{3}} & \land~ & \land~\textcolor{pos}{a_{5} = v_{5}} & \\
 (59) & \textcolor{neg}{a_{1} \neq v_{1}} & \land~ & \land~\textcolor{neg}{a_{3} \neq v_{3}} & \land~\textcolor{pos}{a_{4} = v_{4}} & \land~ & \land~\textcolor{neg}{s=0} & \\
 (60) & \textcolor{neg}{a_{1} \neq v_{1}} & \land~ & \land~ & \land~\textcolor{neg}{a_{4} \neq v_{4}} & \land~\textcolor{neg}{a_{5} \neq v_{5}} & \land~\textcolor{neg}{s=0} & \\
 (61) & \textcolor{neg}{a_{1} \neq v_{1}} & \land~ & \land~ & \land~\textcolor{pos}{a_{4} = v_{4}} & \land~\textcolor{pos}{a_{5} = v_{5}} & \land~\textcolor{neg}{s=0} & \\
 (62) & \textcolor{neg}{a_{1} \neq v_{1}} & \land~\textcolor{pos}{a_{2} = v_{2}} & \land~\textcolor{neg}{a_{3} \neq v_{3}} & \land~ & \land~\textcolor{neg}{a_{5} \neq v_{5}} & \\
 (63) & \textcolor{neg}{a_{1} \neq v_{1}} & \land~\textcolor{pos}{a_{2} = v_{2}} & \land~ & \land~ & \land~\textcolor{neg}{a_{5} \neq v_{5}} & \land~\textcolor{neg}{s=0} & \\
 (64) & \textcolor{neg}{a_{1} \neq v_{1}} & \land~\textcolor{pos}{a_{2} = v_{2}} & \land~ & \land~\textcolor{pos}{a_{4} = v_{4}} & \land~ & \land~\textcolor{neg}{s=0} & \\
 (65) &  & \land~ & \land~\textcolor{neg}{a_{3} \neq v_{3}} & \land~\textcolor{pos}{a_{4} = v_{4}} & \land~\textcolor{neg}{a_{5} \neq v_{5}} & \land~\textcolor{pos}{s=1} & \mathbf{(\times 2)}\\
 (66) & \textcolor{pos}{a_{1} = v_{1}} & \land~ & \land~\textcolor{neg}{a_{3} \neq v_{3}} & \land~ & \land~\textcolor{neg}{a_{5} \neq v_{5}} & \land~\textcolor{neg}{s=0} & \\
 (67) & \textcolor{pos}{a_{1} = v_{1}} & \land~ & \land~\textcolor{neg}{a_{3} \neq v_{3}} & \land~\textcolor{pos}{a_{4} = v_{4}} & \land~\textcolor{pos}{a_{5} = v_{5}} & \\
 (68) & \textcolor{pos}{a_{1} = v_{1}} & \land~ & \land~ & \land~\textcolor{neg}{a_{4} \neq v_{4}} & \land~\textcolor{pos}{a_{5} = v_{5}} & \land~\textcolor{pos}{s=1} & \\
 (69) & \textcolor{neg}{a_{1} \neq v_{1}} & \land~\textcolor{neg}{a_{2} \neq v_{2}} & \land~\textcolor{neg}{a_{3} \neq v_{3}} & \land~\textcolor{neg}{a_{4} \neq v_{4}} & \land~\textcolor{pos}{a_{5} = v_{5}} & \\
 (70) & \textcolor{neg}{a_{1} \neq v_{1}} & \land~ & \land~\textcolor{neg}{a_{3} \neq v_{3}} & \land~\textcolor{neg}{a_{4} \neq v_{4}} & \land~\textcolor{neg}{a_{5} \neq v_{5}} & \land~\textcolor{neg}{s=0} & \\
 (71) & \textcolor{neg}{a_{1} \neq v_{1}} & \land~ & \land~\textcolor{neg}{a_{3} \neq v_{3}} & \land~\textcolor{neg}{a_{4} \neq v_{4}} & \land~\textcolor{neg}{a_{5} \neq v_{5}} & \land~\textcolor{pos}{s=1} & \\
 (72) & \textcolor{neg}{a_{1} \neq v_{1}} & \land~ & \land~\textcolor{pos}{a_{3} = v_{3}} & \land~\textcolor{pos}{a_{4} = v_{4}} & \land~\textcolor{pos}{a_{5} = v_{5}} & \land~\textcolor{neg}{s=0} & \mathbf{(\times 2)}\\
 (73) & \textcolor{neg}{a_{1} \neq v_{1}} & \land~\textcolor{pos}{a_{2} = v_{2}} & \land~\textcolor{neg}{a_{3} \neq v_{3}} & \land~ & \land~\textcolor{pos}{a_{5} = v_{5}} & \land~\textcolor{pos}{s=1} & \\
 (74) & \textcolor{neg}{a_{1} \neq v_{1}} & \land~\textcolor{pos}{a_{2} = v_{2}} & \land~ & \land~\textcolor{neg}{a_{4} \neq v_{4}} & \land~\textcolor{neg}{a_{5} \neq v_{5}} & \land~\textcolor{pos}{s=1} & \\
 (75) & \textcolor{neg}{a_{1} \neq v_{1}} & \land~\textcolor{pos}{a_{2} = v_{2}} & \land~ & \land~\textcolor{pos}{a_{4} = v_{4}} & \land~\textcolor{pos}{a_{5} = v_{5}} & \land~\textcolor{neg}{s=0} & \\
 (76) & \textcolor{neg}{a_{1} \neq v_{1}} & \land~\textcolor{pos}{a_{2} = v_{2}} & \land~\textcolor{pos}{a_{3} = v_{3}} & \land~ & \land~\textcolor{neg}{a_{5} \neq v_{5}} & \land~\textcolor{neg}{s=0} & \\
 (77) & \textcolor{neg}{a_{1} \neq v_{1}} & \land~\textcolor{pos}{a_{2} = v_{2}} & \land~\textcolor{pos}{a_{3} = v_{3}} & \land~ & \land~\textcolor{pos}{a_{5} = v_{5}} & \land~\textcolor{neg}{s=0} & \\
 (78) &  & \land~\textcolor{neg}{a_{2} \neq v_{2}} & \land~\textcolor{neg}{a_{3} \neq v_{3}} & \land~\textcolor{neg}{a_{4} \neq v_{4}} & \land~\textcolor{neg}{a_{5} \neq v_{5}} & \land~\textcolor{pos}{s=1} & \\
 (79) &  & \land~\textcolor{pos}{a_{2} = v_{2}} & \land~\textcolor{neg}{a_{3} \neq v_{3}} & \land~\textcolor{neg}{a_{4} \neq v_{4}} & \land~\textcolor{neg}{a_{5} \neq v_{5}} & \land~\textcolor{pos}{s=1} & \mathbf{(\times 2)}\\
 (80) &  & \land~\textcolor{pos}{a_{2} = v_{2}} & \land~\textcolor{pos}{a_{3} = v_{3}} & \land~\textcolor{pos}{a_{4} = v_{4}} & \land~\textcolor{pos}{a_{5} = v_{5}} & \land~\textcolor{neg}{s=0} & \\
 (81) & \textcolor{pos}{a_{1} = v_{1}} & \land~\textcolor{neg}{a_{2} \neq v_{2}} & \land~\textcolor{neg}{a_{3} \neq v_{3}} & \land~\textcolor{pos}{a_{4} = v_{4}} & \land~ & \land~\textcolor{neg}{s=0} & \\
 (82) & \textcolor{pos}{a_{1} = v_{1}} & \land~ & \land~\textcolor{neg}{a_{3} \neq v_{3}} & \land~\textcolor{neg}{a_{4} \neq v_{4}} & \land~\textcolor{pos}{a_{5} = v_{5}} & \land~\textcolor{neg}{s=0} & \\
 (83) & \textcolor{pos}{a_{1} = v_{1}} & \land~ & \land~\textcolor{pos}{a_{3} = v_{3}} & \land~\textcolor{pos}{a_{4} = v_{4}} & \land~\textcolor{neg}{a_{5} \neq v_{5}} & \land~\textcolor{neg}{s=0} & \\
 (84) & \textcolor{pos}{a_{1} = v_{1}} & \land~\textcolor{pos}{a_{2} = v_{2}} & \land~ & \land~\textcolor{pos}{a_{4} = v_{4}} & \land~\textcolor{neg}{a_{5} \neq v_{5}} & \land~\textcolor{neg}{s=0} & \\
 (85) & \textcolor{pos}{a_{1} = v_{1}} & \land~\textcolor{pos}{a_{2} = v_{2}} & \land~ & \land~\textcolor{pos}{a_{4} = v_{4}} & \land~\textcolor{pos}{a_{5} = v_{5}} & \land~\textcolor{pos}{s=1} & \\
 (86) & \textcolor{neg}{a_{1} \neq v_{1}} & \land~\textcolor{neg}{a_{2} \neq v_{2}} & \land~\textcolor{neg}{a_{3} \neq v_{3}} & \land~\textcolor{neg}{a_{4} \neq v_{4}} & \land~\textcolor{neg}{a_{5} \neq v_{5}} & \land~\textcolor{neg}{s=0} & \\
 (87) & \textcolor{neg}{a_{1} \neq v_{1}} & \land~\textcolor{pos}{a_{2} = v_{2}} & \land~\textcolor{neg}{a_{3} \neq v_{3}} & \land~\textcolor{neg}{a_{4} \neq v_{4}} & \land~\textcolor{neg}{a_{5} \neq v_{5}} & \land~\textcolor{neg}{s=0} & \mathbf{(\times 2)}\\
 (88) & \textcolor{neg}{a_{1} \neq v_{1}} & \land~\textcolor{pos}{a_{2} = v_{2}} & \land~\textcolor{neg}{a_{3} \neq v_{3}} & \land~\textcolor{pos}{a_{4} = v_{4}} & \land~\textcolor{neg}{a_{5} \neq v_{5}} & \land~\textcolor{pos}{s=1} & \\
 (89) & \textcolor{pos}{a_{1} = v_{1}} & \land~\textcolor{neg}{a_{2} \neq v_{2}} & \land~\textcolor{neg}{a_{3} \neq v_{3}} & \land~\textcolor{pos}{a_{4} = v_{4}} & \land~\textcolor{neg}{a_{5} \neq v_{5}} & \land~\textcolor{neg}{s=0} & \\
 (90) & \textcolor{pos}{a_{1} = v_{1}} & \land~\textcolor{neg}{a_{2} \neq v_{2}} & \land~\textcolor{pos}{a_{3} = v_{3}} & \land~\textcolor{neg}{a_{4} \neq v_{4}} & \land~\textcolor{neg}{a_{5} \neq v_{5}} & \land~\textcolor{neg}{s=0} & \mathbf{(\times 3)}\\
 (91) & \textcolor{pos}{a_{1} = v_{1}} & \land~\textcolor{pos}{a_{2} = v_{2}} & \land~\textcolor{pos}{a_{3} = v_{3}} & \land~\textcolor{pos}{a_{4} = v_{4}} & \land~\textcolor{pos}{a_{5} = v_{5}} & \land~\textcolor{neg}{s=0} & \\
\end{array}}
\]
\end{framed}

\newpage
\begin{framed}
\noindent\textbf{Example of solution found against TableBuilder in the AUXILIARY scenario.}
\[
{\customsize
\begin{array}{llllllll}
 \multicolumn{3}{l}{\textbf{Difference queries}} \\
 (1) &  & \land~\textcolor{pos}{a_{2} = v_{2}} & \land~ & \land~ & \land~\textcolor{pos}{a_{5} = v_{5}} & \land~\textcolor{pos}{s=1} & \\
 (2) &  & \land~\textcolor{pos}{a_{2} = v_{2}} & \land~ & \land~\textcolor{neg}{a_{4} \neq v_{4}} & \land~\textcolor{pos}{a_{5} = v_{5}} & \land~\textcolor{pos}{s=1} & \\
 (3) & \textcolor{pos}{a_{1} = v_{1}} & \land~ & \land~ & \land~ & \land~\textcolor{pos}{a_{5} = v_{5}} & \land~\textcolor{pos}{s=1} & \mathbf{(\times 2)}\\
 (4) & \textcolor{pos}{a_{1} = v_{1}} & \land~ & \land~\textcolor{neg}{a_{3} \neq v_{3}} & \land~ & \land~\textcolor{pos}{a_{5} = v_{5}} & \land~\textcolor{pos}{s=1} & \\
 (5) &  & \land~\textcolor{pos}{a_{2} = v_{2}} & \land~\textcolor{pos}{a_{3} = v_{3}} & \land~ & \land~\textcolor{pos}{a_{5} = v_{5}} & \land~\textcolor{pos}{s=1} & \\
 (6) & \textcolor{neg}{a_{1} \neq v_{1}} & \land~\textcolor{pos}{a_{2} = v_{2}} & \land~\textcolor{pos}{a_{3} = v_{3}} & \land~ & \land~\textcolor{pos}{a_{5} = v_{5}} & \land~\textcolor{pos}{s=1} & \\
 (7) & \textcolor{pos}{a_{1} = v_{1}} & \land~ & \land~ & \land~\textcolor{pos}{a_{4} = v_{4}} & \land~\textcolor{pos}{a_{5} = v_{5}} & \land~\textcolor{pos}{s=1} & \\
 (8) & \textcolor{pos}{a_{1} = v_{1}} & \land~ & \land~\textcolor{neg}{a_{3} \neq v_{3}} & \land~\textcolor{pos}{a_{4} = v_{4}} & \land~\textcolor{pos}{a_{5} = v_{5}} & \land~\textcolor{pos}{s=1} & \mathbf{(\times 2)}\\
 (9) & \textcolor{pos}{a_{1} = v_{1}} & \land~ & \land~\textcolor{pos}{a_{3} = v_{3}} & \land~ & \land~\textcolor{pos}{a_{5} = v_{5}} & \land~\textcolor{neg}{s=0} & \\
 (10) & \textcolor{pos}{a_{1} = v_{1}} & \land~ & \land~\textcolor{pos}{a_{3} = v_{3}} & \land~\textcolor{neg}{a_{4} \neq v_{4}} & \land~\textcolor{pos}{a_{5} = v_{5}} & \land~\textcolor{neg}{s=0} & \mathbf{(\times 2)}\\
 (11) & \textcolor{pos}{a_{1} = v_{1}} & \land~\textcolor{pos}{a_{2} = v_{2}} & \land~ & \land~ & \land~\textcolor{pos}{a_{5} = v_{5}} & \land~\textcolor{pos}{s=1} & \\
 (12) & \textcolor{pos}{a_{1} = v_{1}} & \land~\textcolor{pos}{a_{2} = v_{2}} & \land~\textcolor{neg}{a_{3} \neq v_{3}} & \land~ & \land~\textcolor{pos}{a_{5} = v_{5}} & \land~\textcolor{pos}{s=1} & \vspace{\customspace}\\
 \multicolumn{3}{l}{\textbf{Other queries}} \\
 (13) & \textcolor{pos}{a_{1} = v_{1}} & \land~ & \land~ & \land~ & \land~ & \\
 (14) & \textcolor{neg}{a_{1} \neq v_{1}} & \land~ & \land~\textcolor{neg}{a_{3} \neq v_{3}} & \land~ & \land~ & \\
 (15) &  & \land~ & \land~\textcolor{neg}{a_{3} \neq v_{3}} & \land~ & \land~\textcolor{pos}{a_{5} = v_{5}} & \\
 (16) &  & \land~ & \land~ & \land~\textcolor{neg}{a_{4} \neq v_{4}} & \land~ & \land~\textcolor{pos}{s=1} & \\
 (17) &  & \land~ & \land~ & \land~\textcolor{pos}{a_{4} = v_{4}} & \land~\textcolor{neg}{a_{5} \neq v_{5}} & \\
 (18) &  & \land~ & \land~\textcolor{pos}{a_{3} = v_{3}} & \land~ & \land~ & \land~\textcolor{pos}{s=1} & \\
 (19) &  & \land~\textcolor{pos}{a_{2} = v_{2}} & \land~ & \land~\textcolor{neg}{a_{4} \neq v_{4}} & \land~ & \\
 (20) & \textcolor{pos}{a_{1} = v_{1}} & \land~\textcolor{neg}{a_{2} \neq v_{2}} & \land~ & \land~ & \land~ & \\
 (21) & \textcolor{neg}{a_{1} \neq v_{1}} & \land~\textcolor{neg}{a_{2} \neq v_{2}} & \land~ & \land~ & \land~ & \land~\textcolor{neg}{s=0} & \mathbf{(\times 2)}\\
 (22) & \textcolor{neg}{a_{1} \neq v_{1}} & \land~\textcolor{neg}{a_{2} \neq v_{2}} & \land~ & \land~\textcolor{pos}{a_{4} = v_{4}} & \land~ & \\
 (23) & \textcolor{neg}{a_{1} \neq v_{1}} & \land~ & \land~ & \land~\textcolor{neg}{a_{4} \neq v_{4}} & \land~\textcolor{neg}{a_{5} \neq v_{5}} & \\
 (24) & \textcolor{neg}{a_{1} \neq v_{1}} & \land~ & \land~ & \land~\textcolor{neg}{a_{4} \neq v_{4}} & \land~\textcolor{pos}{a_{5} = v_{5}} & \\
 (25) &  & \land~\textcolor{neg}{a_{2} \neq v_{2}} & \land~\textcolor{neg}{a_{3} \neq v_{3}} & \land~ & \land~\textcolor{pos}{a_{5} = v_{5}} & \\
 (26) &  & \land~\textcolor{neg}{a_{2} \neq v_{2}} & \land~ & \land~ & \land~\textcolor{pos}{a_{5} = v_{5}} & \land~\textcolor{neg}{s=0} & \mathbf{(\times 2)}\\
 (27) &  & \land~ & \land~\textcolor{neg}{a_{3} \neq v_{3}} & \land~\textcolor{neg}{a_{4} \neq v_{4}} & \land~ & \land~\textcolor{pos}{s=1} & \\
 (28) &  & \land~ & \land~\textcolor{neg}{a_{3} \neq v_{3}} & \land~\textcolor{pos}{a_{4} = v_{4}} & \land~\textcolor{pos}{a_{5} = v_{5}} & \\
 (29) &  & \land~ & \land~ & \land~\textcolor{neg}{a_{4} \neq v_{4}} & \land~\textcolor{neg}{a_{5} \neq v_{5}} & \land~\textcolor{pos}{s=1} & \\
 (30) &  & \land~ & \land~ & \land~\textcolor{pos}{a_{4} = v_{4}} & \land~\textcolor{neg}{a_{5} \neq v_{5}} & \land~\textcolor{pos}{s=1} & \\
 (31) &  & \land~\textcolor{pos}{a_{2} = v_{2}} & \land~\textcolor{neg}{a_{3} \neq v_{3}} & \land~\textcolor{pos}{a_{4} = v_{4}} & \land~ & \\
 (32) &  & \land~\textcolor{pos}{a_{2} = v_{2}} & \land~ & \land~\textcolor{neg}{a_{4} \neq v_{4}} & \land~\textcolor{neg}{a_{5} \neq v_{5}} & \\
 (33) &  & \land~\textcolor{pos}{a_{2} = v_{2}} & \land~ & \land~ & \land~\textcolor{neg}{a_{5} \neq v_{5}} & \land~\textcolor{neg}{s=0} & \\
 (34) &  & \land~\textcolor{pos}{a_{2} = v_{2}} & \land~ & \land~ & \land~\textcolor{neg}{a_{5} \neq v_{5}} & \land~\textcolor{pos}{s=1} & \\
 (35) &  & \land~\textcolor{pos}{a_{2} = v_{2}} & \land~\textcolor{pos}{a_{3} = v_{3}} & \land~ & \land~ & \land~\textcolor{pos}{s=1} & \mathbf{(\times 2)}\\
 (36) & \textcolor{pos}{a_{1} = v_{1}} & \land~ & \land~ & \land~\textcolor{pos}{a_{4} = v_{4}} & \land~\textcolor{neg}{a_{5} \neq v_{5}} & \\
 (37) & \textcolor{pos}{a_{1} = v_{1}} & \land~ & \land~\textcolor{pos}{a_{3} = v_{3}} & \land~ & \land~\textcolor{neg}{a_{5} \neq v_{5}} & \\
 (38) & \textcolor{pos}{a_{1} = v_{1}} & \land~\textcolor{pos}{a_{2} = v_{2}} & \land~\textcolor{neg}{a_{3} \neq v_{3}} & \land~ & \land~ & \\
 (39) & \textcolor{pos}{a_{1} = v_{1}} & \land~\textcolor{pos}{a_{2} = v_{2}} & \land~ & \land~ & \land~\textcolor{neg}{a_{5} \neq v_{5}} & \\
 (40) & \textcolor{neg}{a_{1} \neq v_{1}} & \land~\textcolor{neg}{a_{2} \neq v_{2}} & \land~ & \land~\textcolor{neg}{a_{4} \neq v_{4}} & \land~\textcolor{pos}{a_{5} = v_{5}} & \\
 (41) & \textcolor{neg}{a_{1} \neq v_{1}} & \land~\textcolor{neg}{a_{2} \neq v_{2}} & \land~ & \land~\textcolor{pos}{a_{4} = v_{4}} & \land~\textcolor{neg}{a_{5} \neq v_{5}} & & \mathbf{(\times 2)}\\
 (42) & \textcolor{neg}{a_{1} \neq v_{1}} & \land~\textcolor{neg}{a_{2} \neq v_{2}} & \land~\textcolor{pos}{a_{3} = v_{3}} & \land~ & \land~\textcolor{neg}{a_{5} \neq v_{5}} & \\
 (43) & \textcolor{neg}{a_{1} \neq v_{1}} & \land~\textcolor{neg}{a_{2} \neq v_{2}} & \land~\textcolor{pos}{a_{3} = v_{3}} & \land~ & \land~\textcolor{pos}{a_{5} = v_{5}} & & \mathbf{(\times 2)}\\
 (44) & \textcolor{neg}{a_{1} \neq v_{1}} & \land~\textcolor{pos}{a_{2} = v_{2}} & \land~\textcolor{neg}{a_{3} \neq v_{3}} & \land~ & \land~\textcolor{neg}{a_{5} \neq v_{5}} & \\
 (45) & \textcolor{neg}{a_{1} \neq v_{1}} & \land~\textcolor{pos}{a_{2} = v_{2}} & \land~ & \land~\textcolor{neg}{a_{4} \neq v_{4}} & \land~ & \land~\textcolor{pos}{s=1} & \\
 (46) & \textcolor{neg}{a_{1} \neq v_{1}} & \land~\textcolor{pos}{a_{2} = v_{2}} & \land~\textcolor{pos}{a_{3} = v_{3}} & \land~\textcolor{neg}{a_{4} \neq v_{4}} & \land~ & & \mathbf{(\times 3)}\\
 (47) &  & \land~\textcolor{neg}{a_{2} \neq v_{2}} & \land~\textcolor{neg}{a_{3} \neq v_{3}} & \land~ & \land~\textcolor{pos}{a_{5} = v_{5}} & \land~\textcolor{pos}{s=1} & \\
 (48) &  & \land~\textcolor{neg}{a_{2} \neq v_{2}} & \land~\textcolor{pos}{a_{3} = v_{3}} & \land~\textcolor{neg}{a_{4} \neq v_{4}} & \land~ & \land~\textcolor{pos}{s=1} & \mathbf{(\times 2)}\\
 (49) &  & \land~\textcolor{neg}{a_{2} \neq v_{2}} & \land~\textcolor{pos}{a_{3} = v_{3}} & \land~\textcolor{pos}{a_{4} = v_{4}} & \land~ & \land~\textcolor{neg}{s=0} & \\
 (50) &  & \land~\textcolor{neg}{a_{2} \neq v_{2}} & \land~\textcolor{pos}{a_{3} = v_{3}} & \land~\textcolor{pos}{a_{4} = v_{4}} & \land~\textcolor{pos}{a_{5} = v_{5}} & & \mathbf{(\times 3)}\\
 (51) &  & \land~\textcolor{pos}{a_{2} = v_{2}} & \land~\textcolor{pos}{a_{3} = v_{3}} & \land~\textcolor{neg}{a_{4} \neq v_{4}} & \land~\textcolor{neg}{a_{5} \neq v_{5}} & \\
 (52) &  & \land~\textcolor{pos}{a_{2} = v_{2}} & \land~\textcolor{pos}{a_{3} = v_{3}} & \land~\textcolor{pos}{a_{4} = v_{4}} & \land~\textcolor{neg}{a_{5} \neq v_{5}} & & \mathbf{(\times 2)}\\
 (53) &  & \land~\textcolor{pos}{a_{2} = v_{2}} & \land~\textcolor{pos}{a_{3} = v_{3}} & \land~\textcolor{pos}{a_{4} = v_{4}} & \land~ & \land~\textcolor{pos}{s=1} & \\
 (54) & \textcolor{pos}{a_{1} = v_{1}} & \land~\textcolor{neg}{a_{2} \neq v_{2}} & \land~\textcolor{neg}{a_{3} \neq v_{3}} & \land~ & \land~\textcolor{neg}{a_{5} \neq v_{5}} & \\
 (55) & \textcolor{pos}{a_{1} = v_{1}} & \land~\textcolor{neg}{a_{2} \neq v_{2}} & \land~\textcolor{neg}{a_{3} \neq v_{3}} & \land~ & \land~\textcolor{pos}{a_{5} = v_{5}} & \\
 (56) & \textcolor{pos}{a_{1} = v_{1}} & \land~ & \land~ & \land~\textcolor{pos}{a_{4} = v_{4}} & \land~\textcolor{neg}{a_{5} \neq v_{5}} & \land~\textcolor{pos}{s=1} & \\
 (57) & \textcolor{pos}{a_{1} = v_{1}} & \land~\textcolor{pos}{a_{2} = v_{2}} & \land~\textcolor{neg}{a_{3} \neq v_{3}} & \land~ & \land~\textcolor{neg}{a_{5} \neq v_{5}} & & \mathbf{(\times 2)}\\
 (58) & \textcolor{pos}{a_{1} = v_{1}} & \land~\textcolor{pos}{a_{2} = v_{2}} & \land~\textcolor{neg}{a_{3} \neq v_{3}} & \land~\textcolor{pos}{a_{4} = v_{4}} & \land~ & \\
 (59) & \textcolor{pos}{a_{1} = v_{1}} & \land~\textcolor{pos}{a_{2} = v_{2}} & \land~ & \land~\textcolor{neg}{a_{4} \neq v_{4}} & \land~ & \land~\textcolor{neg}{s=0} & \\
 (60) & \textcolor{pos}{a_{1} = v_{1}} & \land~\textcolor{pos}{a_{2} = v_{2}} & \land~ & \land~\textcolor{pos}{a_{4} = v_{4}} & \land~\textcolor{pos}{a_{5} = v_{5}} & \\
 (61) & \textcolor{neg}{a_{1} \neq v_{1}} & \land~\textcolor{neg}{a_{2} \neq v_{2}} & \land~\textcolor{pos}{a_{3} = v_{3}} & \land~\textcolor{pos}{a_{4} = v_{4}} & \land~\textcolor{pos}{a_{5} = v_{5}} & & \mathbf{(\times 2)}\\
 (62) & \textcolor{neg}{a_{1} \neq v_{1}} & \land~\textcolor{pos}{a_{2} = v_{2}} & \land~ & \land~\textcolor{neg}{a_{4} \neq v_{4}} & \land~\textcolor{pos}{a_{5} = v_{5}} & \land~\textcolor{neg}{s=0} & \\
 (63) & \textcolor{neg}{a_{1} \neq v_{1}} & \land~\textcolor{pos}{a_{2} = v_{2}} & \land~ & \land~\textcolor{neg}{a_{4} \neq v_{4}} & \land~\textcolor{pos}{a_{5} = v_{5}} & \land~\textcolor{pos}{s=1} & \\
 (64) &  & \land~\textcolor{neg}{a_{2} \neq v_{2}} & \land~\textcolor{neg}{a_{3} \neq v_{3}} & \land~\textcolor{pos}{a_{4} = v_{4}} & \land~\textcolor{pos}{a_{5} = v_{5}} & \land~\textcolor{pos}{s=1} & \\
 (65) &  & \land~\textcolor{neg}{a_{2} \neq v_{2}} & \land~\textcolor{pos}{a_{3} = v_{3}} & \land~\textcolor{pos}{a_{4} = v_{4}} & \land~\textcolor{neg}{a_{5} \neq v_{5}} & \land~\textcolor{pos}{s=1} & \\
 (66) &  & \land~\textcolor{neg}{a_{2} \neq v_{2}} & \land~\textcolor{pos}{a_{3} = v_{3}} & \land~\textcolor{pos}{a_{4} = v_{4}} & \land~\textcolor{pos}{a_{5} = v_{5}} & \land~\textcolor{neg}{s=0} & \\
 (67) &  & \land~\textcolor{pos}{a_{2} = v_{2}} & \land~\textcolor{pos}{a_{3} = v_{3}} & \land~\textcolor{neg}{a_{4} \neq v_{4}} & \land~\textcolor{pos}{a_{5} = v_{5}} & \land~\textcolor{neg}{s=0} & \\
 (68) & \textcolor{pos}{a_{1} = v_{1}} & \land~\textcolor{neg}{a_{2} \neq v_{2}} & \land~\textcolor{neg}{a_{3} \neq v_{3}} & \land~\textcolor{neg}{a_{4} \neq v_{4}} & \land~\textcolor{neg}{a_{5} \neq v_{5}} & \\
 (69) & \textcolor{pos}{a_{1} = v_{1}} & \land~\textcolor{neg}{a_{2} \neq v_{2}} & \land~\textcolor{neg}{a_{3} \neq v_{3}} & \land~\textcolor{pos}{a_{4} = v_{4}} & \land~ & \land~\textcolor{pos}{s=1} & \\
 (70) & \textcolor{pos}{a_{1} = v_{1}} & \land~\textcolor{neg}{a_{2} \neq v_{2}} & \land~ & \land~\textcolor{pos}{a_{4} = v_{4}} & \land~\textcolor{neg}{a_{5} \neq v_{5}} & \land~\textcolor{neg}{s=0} & \\
 (71) & \textcolor{pos}{a_{1} = v_{1}} & \land~\textcolor{neg}{a_{2} \neq v_{2}} & \land~ & \land~\textcolor{pos}{a_{4} = v_{4}} & \land~\textcolor{neg}{a_{5} \neq v_{5}} & \land~\textcolor{pos}{s=1} & \mathbf{(\times 2)}\\
 (72) & \textcolor{pos}{a_{1} = v_{1}} & \land~\textcolor{neg}{a_{2} \neq v_{2}} & \land~\textcolor{pos}{a_{3} = v_{3}} & \land~\textcolor{pos}{a_{4} = v_{4}} & \land~\textcolor{neg}{a_{5} \neq v_{5}} & \\
 (73) & \textcolor{pos}{a_{1} = v_{1}} & \land~\textcolor{neg}{a_{2} \neq v_{2}} & \land~\textcolor{pos}{a_{3} = v_{3}} & \land~\textcolor{pos}{a_{4} = v_{4}} & \land~ & \land~\textcolor{neg}{s=0} & \\
 (74) & \textcolor{pos}{a_{1} = v_{1}} & \land~ & \land~\textcolor{neg}{a_{3} \neq v_{3}} & \land~\textcolor{pos}{a_{4} = v_{4}} & \land~\textcolor{neg}{a_{5} \neq v_{5}} & \land~\textcolor{neg}{s=0} & \\
 (75) & \textcolor{pos}{a_{1} = v_{1}} & \land~\textcolor{pos}{a_{2} = v_{2}} & \land~\textcolor{neg}{a_{3} \neq v_{3}} & \land~\textcolor{neg}{a_{4} \neq v_{4}} & \land~\textcolor{neg}{a_{5} \neq v_{5}} & \\
 (76) & \textcolor{pos}{a_{1} = v_{1}} & \land~\textcolor{pos}{a_{2} = v_{2}} & \land~\textcolor{neg}{a_{3} \neq v_{3}} & \land~\textcolor{pos}{a_{4} = v_{4}} & \land~\textcolor{pos}{a_{5} = v_{5}} & \\
 (77) & \textcolor{pos}{a_{1} = v_{1}} & \land~\textcolor{pos}{a_{2} = v_{2}} & \land~ & \land~\textcolor{pos}{a_{4} = v_{4}} & \land~\textcolor{neg}{a_{5} \neq v_{5}} & \land~\textcolor{neg}{s=0} & \mathbf{(\times 3)}\\
 (78) & \textcolor{pos}{a_{1} = v_{1}} & \land~\textcolor{pos}{a_{2} = v_{2}} & \land~\textcolor{pos}{a_{3} = v_{3}} & \land~ & \land~\textcolor{pos}{a_{5} = v_{5}} & \land~\textcolor{pos}{s=1} & \\
 (79) & \textcolor{pos}{a_{1} = v_{1}} & \land~\textcolor{pos}{a_{2} = v_{2}} & \land~\textcolor{pos}{a_{3} = v_{3}} & \land~\textcolor{pos}{a_{4} = v_{4}} & \land~ & \land~\textcolor{pos}{s=1} & \\
 (80) & \textcolor{neg}{a_{1} \neq v_{1}} & \land~\textcolor{pos}{a_{2} = v_{2}} & \land~\textcolor{pos}{a_{3} = v_{3}} & \land~\textcolor{pos}{a_{4} = v_{4}} & \land~\textcolor{pos}{a_{5} = v_{5}} & \land~\textcolor{pos}{s=1} & \\
 (81) & \textcolor{pos}{a_{1} = v_{1}} & \land~\textcolor{neg}{a_{2} \neq v_{2}} & \land~\textcolor{neg}{a_{3} \neq v_{3}} & \land~\textcolor{pos}{a_{4} = v_{4}} & \land~\textcolor{pos}{a_{5} = v_{5}} & \land~\textcolor{pos}{s=1} & \\
\end{array}}
\]
\end{framed}

\newpage

\textbf{DPLaplace. }
We next report the attacks found by our procedure for DPLaplace and different values of $\varepsilon$. In the first two cases, we obtain a uniqueness attack of one query using the full budget (see Sec. \ref{subsec:dplaplace} for more details on these attacks). Note that the queries select on strict subsets of the known attributes, $\{ a_1, a_2, a_3\}$ ($\varepsilon=1$) and $\{ a_1, a_2, a_5\}$ ($\varepsilon=5$). Even though this results in a gap in performance with respect to the manual attacks, note how QuerySnout identified the core vulnerability: uniqueness attacks with many repetitions. For $\varepsilon=10$, the attack found consists of a uniqueness attack (with 9 repetitions) and another query. 

\begin{framed}
\noindent\textbf{Example of solution found against DPLaplace($\varepsilon=1$) in the AUXILIARY scenario.}
\[
{\customsize
\begin{array}{llllllll}
 \multicolumn{3}{l}{\textbf{Queries}} \\
 (1) & \textcolor{pos}{a_{1} = v_{1}} & \land~\textcolor{pos}{a_{2} = v_{2}} & \land~\textcolor{pos}{a_{3} = v_{3}} & \land~ & \land~ & \land~\textcolor{neg}{s=0} & \mathbf{(\times 10)}\\
\end{array}}
\]
\end{framed}

\begin{framed}
\noindent\textbf{Example of solution found against DPLaplace($\varepsilon=5$) in the AUXILIARY scenario.}
\[
{\customsize
\begin{array}{llllllll}
 \multicolumn{3}{l}{\textbf{Queries}} \\
 (1) & \textcolor{pos}{a_{1} = v_{1}} & \land~\textcolor{pos}{a_{2} = v_{2}} & \land~ & \land~ & \land~\textcolor{pos}{a_{5} = v_{5}} & \land~\textcolor{pos}{s=1} & \mathbf{(\times 10)}\\
\end{array}}
\]
\end{framed}

\begin{framed}
\noindent\textbf{Example of solution found against DPLaplace($\varepsilon=10$) in the AUXILIARY scenario.}
\[
{\customsize
\begin{array}{llllllll}
 \multicolumn{3}{l}{\textbf{Queries}} \\
 (1) & \textcolor{neg}{a_{1} \neq v_{1}} & \land~\textcolor{neg}{a_{2} \neq v_{2}} & \land~ & \land~ & \land~ & \land~\textcolor{pos}{s=1} & \\
 (2) & \textcolor{pos}{a_{1} = v_{1}} & \land~\textcolor{pos}{a_{2} = v_{2}} & \land~ & \land~ & \land~\textcolor{pos}{a_{5} = v_{5}} & \land~\textcolor{pos}{s=1} & \mathbf{(\times 9)}\\
\end{array}}
\]
\end{framed}

\textbf{SimpleQBS. }
We report the attacks found for four instances of SimpleQBS: $(\tau=0, \sigma=4)$, $(\tau=4, \sigma=0)$, $(\tau=4, \sigma=3)$, and $(\tau=3, \sigma=4)$.
We choose these QBSes to show qualitatively different solutions.
For $(\tau=0,\sigma=4)$, our procedure finds a uniqueness attack over the five attributes, repeated $6$ times (query \#75), as well as other uniqueness attacks for different subsets of attributes (queries \#1, \#3$\times 2$, \#21$\times 2$, \#24, \#26, \#27, \#28, \#42, \#43, \#58, \#60$\times2$, \#62 and \#63).
For $(\tau=4,\sigma=0)$, we obtain a large number of difference attacks~\cite{denning1979tracker} (queries \#1-\#9, grouped in pairs). Since the QBS is deterministic, we observe that queries are rarely repeated, and never more than 3 times. 
Finally, for $(\tau=4,\sigma=3)$, the search finds a large number of difference attacks (queries \#1-\#19) that are often repeated more than once. Similar results are obtained for $(\tau=3,\sigma=4)$, but note that as expected, since the noise is larger some of the queries are repeated a large number of times.
This is consistent with our findings for the other two: since $\tau>0$, the attacks must include difference queries, and since $\sigma>0$, repetitions are needed to average out the noise.

\newpage
\begin{framed}
\noindent\textbf{Example of solution found against SimpleQBS($\tau=0,\sigma=4$) in the AUXILIARY scenario.}
\[
{\customsize
\begin{array}{llllllll}
 \multicolumn{3}{l}{\textbf{Difference queries}} \\
 (1) & \textcolor{pos}{a_{1} = v_{1}} & \land~\textcolor{pos}{a_{2} = v_{2}} & \land~\textcolor{pos}{a_{3} = v_{3}} & \land~ & \land~ & \land~\textcolor{neg}{s=0} & \\
 (2) & \textcolor{pos}{a_{1} = v_{1}} & \land~\textcolor{pos}{a_{2} = v_{2}} & \land~\textcolor{pos}{a_{3} = v_{3}} & \land~ & \land~\textcolor{neg}{a_{5} \neq v_{5}} & \land~\textcolor{neg}{s=0} & \\
 (3) & \textcolor{pos}{a_{1} = v_{1}} & \land~\textcolor{pos}{a_{2} = v_{2}} & \land~\textcolor{pos}{a_{3} = v_{3}} & \land~ & \land~\textcolor{pos}{a_{5} = v_{5}} & \land~\textcolor{pos}{s=1} & \mathbf{(\times 2)}\\
 (4) & \textcolor{pos}{a_{1} = v_{1}} & \land~\textcolor{pos}{a_{2} = v_{2}} & \land~\textcolor{pos}{a_{3} = v_{3}} & \land~\textcolor{neg}{a_{4} \neq v_{4}} & \land~\textcolor{pos}{a_{5} = v_{5}} & \land~\textcolor{pos}{s=1} & \vspace{\customspace}\\
 \multicolumn{3}{l}{\textbf{Other queries}} \\
 (5) &  & \land~ & \land~ & \land~ & \land~ & \land~\textcolor{neg}{s=0} & \\
 (6) &  & \land~ & \land~\textcolor{pos}{a_{3} = v_{3}} & \land~ & \land~ & & \mathbf{(\times 2)}\\
 (7) & \textcolor{pos}{a_{1} = v_{1}} & \land~ & \land~ & \land~ & \land~ & & \mathbf{(\times 2)}\\
 (8) & \textcolor{neg}{a_{1} \neq v_{1}} & \land~ & \land~ & \land~ & \land~ & \land~\textcolor{pos}{s=1} & \\
 (9) &  & \land~\textcolor{neg}{a_{2} \neq v_{2}} & \land~ & \land~ & \land~\textcolor{neg}{a_{5} \neq v_{5}} & \\
 (10) &  & \land~ & \land~ & \land~\textcolor{pos}{a_{4} = v_{4}} & \land~\textcolor{neg}{a_{5} \neq v_{5}} & \\
 (11) &  & \land~ & \land~\textcolor{pos}{a_{3} = v_{3}} & \land~\textcolor{pos}{a_{4} = v_{4}} & \land~ & \\
 (12) &  & \land~\textcolor{pos}{a_{2} = v_{2}} & \land~ & \land~\textcolor{pos}{a_{4} = v_{4}} & \land~ & \\
 (13) & \textcolor{pos}{a_{1} = v_{1}} & \land~ & \land~ & \land~ & \land~\textcolor{pos}{a_{5} = v_{5}} & \\
 (14) & \textcolor{pos}{a_{1} = v_{1}} & \land~\textcolor{pos}{a_{2} = v_{2}} & \land~ & \land~ & \land~ & \\
 (15) & \textcolor{neg}{a_{1} \neq v_{1}} & \land~ & \land~\textcolor{neg}{a_{3} \neq v_{3}} & \land~ & \land~\textcolor{pos}{a_{5} = v_{5}} & \\
 (16) & \textcolor{neg}{a_{1} \neq v_{1}} & \land~ & \land~ & \land~\textcolor{neg}{a_{4} \neq v_{4}} & \land~\textcolor{neg}{a_{5} \neq v_{5}} & \\
 (17) &  & \land~\textcolor{neg}{a_{2} \neq v_{2}} & \land~\textcolor{neg}{a_{3} \neq v_{3}} & \land~ & \land~\textcolor{neg}{a_{5} \neq v_{5}} & & \mathbf{(\times 3)}\\
 (18) &  & \land~\textcolor{neg}{a_{2} \neq v_{2}} & \land~ & \land~\textcolor{neg}{a_{4} \neq v_{4}} & \land~\textcolor{neg}{a_{5} \neq v_{5}} & & \mathbf{(\times 2)}\\
 (19) &  & \land~\textcolor{neg}{a_{2} \neq v_{2}} & \land~ & \land~ & \land~\textcolor{neg}{a_{5} \neq v_{5}} & \land~\textcolor{neg}{s=0} & \mathbf{(\times 2)}\\
 (20) &  & \land~ & \land~ & \land~\textcolor{neg}{a_{4} \neq v_{4}} & \land~\textcolor{pos}{a_{5} = v_{5}} & \land~\textcolor{pos}{s=1} & \\
 (21) &  & \land~ & \land~\textcolor{pos}{a_{3} = v_{3}} & \land~ & \land~\textcolor{pos}{a_{5} = v_{5}} & \land~\textcolor{neg}{s=0} & \mathbf{(\times 2)}\\
 (22) &  & \land~\textcolor{pos}{a_{2} = v_{2}} & \land~\textcolor{neg}{a_{3} \neq v_{3}} & \land~ & \land~\textcolor{pos}{a_{5} = v_{5}} & & \mathbf{(\times 2)}\\
 (23) &  & \land~\textcolor{pos}{a_{2} = v_{2}} & \land~ & \land~\textcolor{neg}{a_{4} \neq v_{4}} & \land~ & \land~\textcolor{pos}{s=1} & \\
 (24) &  & \land~\textcolor{pos}{a_{2} = v_{2}} & \land~ & \land~\textcolor{pos}{a_{4} = v_{4}} & \land~ & \land~\textcolor{neg}{s=0} & \\
 (25) &  & \land~\textcolor{pos}{a_{2} = v_{2}} & \land~ & \land~\textcolor{pos}{a_{4} = v_{4}} & \land~\textcolor{pos}{a_{5} = v_{5}} & \\
 (26) & \textcolor{pos}{a_{1} = v_{1}} & \land~ & \land~ & \land~\textcolor{neg}{a_{4} \neq v_{4}} & \land~ & \land~\textcolor{neg}{s=0} & \\
 (27) & \textcolor{pos}{a_{1} = v_{1}} & \land~ & \land~ & \land~\textcolor{pos}{a_{4} = v_{4}} & \land~ & \land~\textcolor{neg}{s=0} & \\
 (28) & \textcolor{pos}{a_{1} = v_{1}} & \land~ & \land~ & \land~\textcolor{pos}{a_{4} = v_{4}} & \land~ & \land~\textcolor{pos}{s=1} & \\
 (29) & \textcolor{pos}{a_{1} = v_{1}} & \land~ & \land~\textcolor{pos}{a_{3} = v_{3}} & \land~\textcolor{neg}{a_{4} \neq v_{4}} & \land~ & \\
 (30) & \textcolor{neg}{a_{1} \neq v_{1}} & \land~\textcolor{neg}{a_{2} \neq v_{2}} & \land~\textcolor{neg}{a_{3} \neq v_{3}} & \land~ & \land~\textcolor{pos}{a_{5} = v_{5}} & \\
 (31) & \textcolor{neg}{a_{1} \neq v_{1}} & \land~\textcolor{neg}{a_{2} \neq v_{2}} & \land~ & \land~ & \land~\textcolor{neg}{a_{5} \neq v_{5}} & \land~\textcolor{pos}{s=1} & \\
 (32) & \textcolor{neg}{a_{1} \neq v_{1}} & \land~\textcolor{neg}{a_{2} \neq v_{2}} & \land~ & \land~ & \land~\textcolor{pos}{a_{5} = v_{5}} & \land~\textcolor{pos}{s=1} & \\
 (33) & \textcolor{neg}{a_{1} \neq v_{1}} & \land~\textcolor{neg}{a_{2} \neq v_{2}} & \land~ & \land~\textcolor{pos}{a_{4} = v_{4}} & \land~\textcolor{pos}{a_{5} = v_{5}} & \\
 (34) & \textcolor{neg}{a_{1} \neq v_{1}} & \land~ & \land~\textcolor{neg}{a_{3} \neq v_{3}} & \land~\textcolor{pos}{a_{4} = v_{4}} & \land~\textcolor{neg}{a_{5} \neq v_{5}} & \\
 (35) & \textcolor{neg}{a_{1} \neq v_{1}} & \land~ & \land~ & \land~\textcolor{pos}{a_{4} = v_{4}} & \land~\textcolor{pos}{a_{5} = v_{5}} & \land~\textcolor{pos}{s=1} & \mathbf{(\times 2)}\\
 (36) & \textcolor{neg}{a_{1} \neq v_{1}} & \land~ & \land~\textcolor{pos}{a_{3} = v_{3}} & \land~\textcolor{neg}{a_{4} \neq v_{4}} & \land~\textcolor{neg}{a_{5} \neq v_{5}} & & \mathbf{(\times 2)}\\
 (37) & \textcolor{neg}{a_{1} \neq v_{1}} & \land~\textcolor{pos}{a_{2} = v_{2}} & \land~ & \land~\textcolor{neg}{a_{4} \neq v_{4}} & \land~ & \land~\textcolor{neg}{s=0} & \\
 (38) & \textcolor{neg}{a_{1} \neq v_{1}} & \land~\textcolor{pos}{a_{2} = v_{2}} & \land~\textcolor{pos}{a_{3} = v_{3}} & \land~ & \land~\textcolor{pos}{a_{5} = v_{5}} & \\
 (39) &  & \land~\textcolor{neg}{a_{2} \neq v_{2}} & \land~\textcolor{pos}{a_{3} = v_{3}} & \land~ & \land~\textcolor{neg}{a_{5} \neq v_{5}} & \land~\textcolor{pos}{s=1} & \\
 (40) &  & \land~\textcolor{pos}{a_{2} = v_{2}} & \land~\textcolor{pos}{a_{3} = v_{3}} & \land~\textcolor{neg}{a_{4} \neq v_{4}} & \land~\textcolor{pos}{a_{5} = v_{5}} & & \mathbf{(\times 2)}\\
 (41) & \textcolor{pos}{a_{1} = v_{1}} & \land~\textcolor{neg}{a_{2} \neq v_{2}} & \land~\textcolor{pos}{a_{3} = v_{3}} & \land~ & \land~ & \land~\textcolor{pos}{s=1} & \\
 (42) & \textcolor{pos}{a_{1} = v_{1}} & \land~ & \land~ & \land~\textcolor{pos}{a_{4} = v_{4}} & \land~\textcolor{pos}{a_{5} = v_{5}} & \land~\textcolor{neg}{s=0} & \\
 (43) & \textcolor{pos}{a_{1} = v_{1}} & \land~ & \land~\textcolor{pos}{a_{3} = v_{3}} & \land~ & \land~\textcolor{pos}{a_{5} = v_{5}} & \land~\textcolor{pos}{s=1} & \\
 (44) & \textcolor{pos}{a_{1} = v_{1}} & \land~\textcolor{pos}{a_{2} = v_{2}} & \land~\textcolor{neg}{a_{3} \neq v_{3}} & \land~ & \land~ & \land~\textcolor{neg}{s=0} & \mathbf{(\times 2)}\\
 (45) & \textcolor{pos}{a_{1} = v_{1}} & \land~\textcolor{pos}{a_{2} = v_{2}} & \land~ & \land~ & \land~\textcolor{neg}{a_{5} \neq v_{5}} & \land~\textcolor{pos}{s=1} & \\
 (46) & \textcolor{neg}{a_{1} \neq v_{1}} & \land~\textcolor{neg}{a_{2} \neq v_{2}} & \land~\textcolor{neg}{a_{3} \neq v_{3}} & \land~ & \land~\textcolor{neg}{a_{5} \neq v_{5}} & \land~\textcolor{neg}{s=0} & \\
 (47) & \textcolor{neg}{a_{1} \neq v_{1}} & \land~\textcolor{neg}{a_{2} \neq v_{2}} & \land~ & \land~\textcolor{neg}{a_{4} \neq v_{4}} & \land~\textcolor{neg}{a_{5} \neq v_{5}} & \land~\textcolor{pos}{s=1} & \mathbf{(\times 2)}\\
 (48) & \textcolor{neg}{a_{1} \neq v_{1}} & \land~\textcolor{neg}{a_{2} \neq v_{2}} & \land~\textcolor{pos}{a_{3} = v_{3}} & \land~ & \land~\textcolor{pos}{a_{5} = v_{5}} & \land~\textcolor{neg}{s=0} & \\
 (49) & \textcolor{neg}{a_{1} \neq v_{1}} & \land~ & \land~\textcolor{neg}{a_{3} \neq v_{3}} & \land~\textcolor{pos}{a_{4} = v_{4}} & \land~\textcolor{pos}{a_{5} = v_{5}} & \land~\textcolor{pos}{s=1} & \\
 (50) & \textcolor{neg}{a_{1} \neq v_{1}} & \land~ & \land~\textcolor{pos}{a_{3} = v_{3}} & \land~\textcolor{neg}{a_{4} \neq v_{4}} & \land~\textcolor{pos}{a_{5} = v_{5}} & \land~\textcolor{pos}{s=1} & \mathbf{(\times 2)}\\
 (51) & \textcolor{neg}{a_{1} \neq v_{1}} & \land~ & \land~\textcolor{pos}{a_{3} = v_{3}} & \land~\textcolor{pos}{a_{4} = v_{4}} & \land~\textcolor{neg}{a_{5} \neq v_{5}} & \land~\textcolor{neg}{s=0} & \\
 (52) & \textcolor{neg}{a_{1} \neq v_{1}} & \land~ & \land~\textcolor{pos}{a_{3} = v_{3}} & \land~\textcolor{pos}{a_{4} = v_{4}} & \land~\textcolor{neg}{a_{5} \neq v_{5}} & \land~\textcolor{pos}{s=1} & \\
 (53) & \textcolor{neg}{a_{1} \neq v_{1}} & \land~\textcolor{pos}{a_{2} = v_{2}} & \land~\textcolor{neg}{a_{3} \neq v_{3}} & \land~ & \land~\textcolor{pos}{a_{5} = v_{5}} & \land~\textcolor{neg}{s=0} & \\
 (54) &  & \land~\textcolor{neg}{a_{2} \neq v_{2}} & \land~\textcolor{neg}{a_{3} \neq v_{3}} & \land~\textcolor{pos}{a_{4} = v_{4}} & \land~\textcolor{pos}{a_{5} = v_{5}} & \land~\textcolor{neg}{s=0} & \\
 (55) &  & \land~\textcolor{neg}{a_{2} \neq v_{2}} & \land~\textcolor{pos}{a_{3} = v_{3}} & \land~\textcolor{pos}{a_{4} = v_{4}} & \land~\textcolor{pos}{a_{5} = v_{5}} & \land~\textcolor{neg}{s=0} & \\
 (56) &  & \land~\textcolor{neg}{a_{2} \neq v_{2}} & \land~\textcolor{pos}{a_{3} = v_{3}} & \land~\textcolor{pos}{a_{4} = v_{4}} & \land~\textcolor{pos}{a_{5} = v_{5}} & \land~\textcolor{pos}{s=1} & \mathbf{(\times 2)}\\
 (57) &  & \land~\textcolor{pos}{a_{2} = v_{2}} & \land~\textcolor{pos}{a_{3} = v_{3}} & \land~\textcolor{neg}{a_{4} \neq v_{4}} & \land~\textcolor{neg}{a_{5} \neq v_{5}} & \land~\textcolor{pos}{s=1} & \\
 (58) &  & \land~\textcolor{pos}{a_{2} = v_{2}} & \land~\textcolor{pos}{a_{3} = v_{3}} & \land~\textcolor{pos}{a_{4} = v_{4}} & \land~\textcolor{pos}{a_{5} = v_{5}} & \land~\textcolor{neg}{s=0} & \\
 (59) & \textcolor{pos}{a_{1} = v_{1}} & \land~\textcolor{neg}{a_{2} \neq v_{2}} & \land~\textcolor{pos}{a_{3} = v_{3}} & \land~ & \land~\textcolor{pos}{a_{5} = v_{5}} & \land~\textcolor{neg}{s=0} & \\
 (60) & \textcolor{pos}{a_{1} = v_{1}} & \land~ & \land~\textcolor{pos}{a_{3} = v_{3}} & \land~\textcolor{pos}{a_{4} = v_{4}} & \land~\textcolor{pos}{a_{5} = v_{5}} & \land~\textcolor{pos}{s=1} & \mathbf{(\times 2)}\\
 (61) & \textcolor{pos}{a_{1} = v_{1}} & \land~\textcolor{pos}{a_{2} = v_{2}} & \land~\textcolor{neg}{a_{3} \neq v_{3}} & \land~\textcolor{neg}{a_{4} \neq v_{4}} & \land~\textcolor{neg}{a_{5} \neq v_{5}} & & \mathbf{(\times 3)}\\
 (62) & \textcolor{pos}{a_{1} = v_{1}} & \land~\textcolor{pos}{a_{2} = v_{2}} & \land~\textcolor{pos}{a_{3} = v_{3}} & \land~ & \land~\textcolor{pos}{a_{5} = v_{5}} & \land~\textcolor{neg}{s=0} & \\
 (63) & \textcolor{pos}{a_{1} = v_{1}} & \land~\textcolor{pos}{a_{2} = v_{2}} & \land~\textcolor{pos}{a_{3} = v_{3}} & \land~\textcolor{pos}{a_{4} = v_{4}} & \land~ & \land~\textcolor{pos}{s=1} & \\
 (64) & \textcolor{neg}{a_{1} \neq v_{1}} & \land~\textcolor{neg}{a_{2} \neq v_{2}} & \land~\textcolor{neg}{a_{3} \neq v_{3}} & \land~\textcolor{neg}{a_{4} \neq v_{4}} & \land~\textcolor{neg}{a_{5} \neq v_{5}} & \land~\textcolor{neg}{s=0} & \\
 (65) & \textcolor{neg}{a_{1} \neq v_{1}} & \land~\textcolor{neg}{a_{2} \neq v_{2}} & \land~\textcolor{pos}{a_{3} = v_{3}} & \land~\textcolor{neg}{a_{4} \neq v_{4}} & \land~\textcolor{pos}{a_{5} = v_{5}} & \land~\textcolor{neg}{s=0} & \\
 (66) & \textcolor{neg}{a_{1} \neq v_{1}} & \land~\textcolor{pos}{a_{2} = v_{2}} & \land~\textcolor{neg}{a_{3} \neq v_{3}} & \land~\textcolor{neg}{a_{4} \neq v_{4}} & \land~\textcolor{neg}{a_{5} \neq v_{5}} & \land~\textcolor{pos}{s=1} & \mathbf{(\times 2)}\\
 (67) & \textcolor{neg}{a_{1} \neq v_{1}} & \land~\textcolor{pos}{a_{2} = v_{2}} & \land~\textcolor{neg}{a_{3} \neq v_{3}} & \land~\textcolor{neg}{a_{4} \neq v_{4}} & \land~\textcolor{pos}{a_{5} = v_{5}} & \land~\textcolor{pos}{s=1} & \\
 (68) & \textcolor{pos}{a_{1} = v_{1}} & \land~\textcolor{neg}{a_{2} \neq v_{2}} & \land~\textcolor{neg}{a_{3} \neq v_{3}} & \land~\textcolor{pos}{a_{4} = v_{4}} & \land~\textcolor{pos}{a_{5} = v_{5}} & \land~\textcolor{pos}{s=1} & \\
 (69) & \textcolor{pos}{a_{1} = v_{1}} & \land~\textcolor{neg}{a_{2} \neq v_{2}} & \land~\textcolor{pos}{a_{3} = v_{3}} & \land~\textcolor{pos}{a_{4} = v_{4}} & \land~\textcolor{neg}{a_{5} \neq v_{5}} & \land~\textcolor{pos}{s=1} & \\
 (70) & \textcolor{pos}{a_{1} = v_{1}} & \land~\textcolor{pos}{a_{2} = v_{2}} & \land~\textcolor{neg}{a_{3} \neq v_{3}} & \land~\textcolor{neg}{a_{4} \neq v_{4}} & \land~\textcolor{neg}{a_{5} \neq v_{5}} & \land~\textcolor{neg}{s=0} & \\
 (71) & \textcolor{pos}{a_{1} = v_{1}} & \land~\textcolor{pos}{a_{2} = v_{2}} & \land~\textcolor{neg}{a_{3} \neq v_{3}} & \land~\textcolor{neg}{a_{4} \neq v_{4}} & \land~\textcolor{neg}{a_{5} \neq v_{5}} & \land~\textcolor{pos}{s=1} & \\
 (72) & \textcolor{pos}{a_{1} = v_{1}} & \land~\textcolor{pos}{a_{2} = v_{2}} & \land~\textcolor{neg}{a_{3} \neq v_{3}} & \land~\textcolor{neg}{a_{4} \neq v_{4}} & \land~\textcolor{pos}{a_{5} = v_{5}} & \land~\textcolor{pos}{s=1} & \\
 (73) & \textcolor{pos}{a_{1} = v_{1}} & \land~\textcolor{pos}{a_{2} = v_{2}} & \land~\textcolor{neg}{a_{3} \neq v_{3}} & \land~\textcolor{pos}{a_{4} = v_{4}} & \land~\textcolor{neg}{a_{5} \neq v_{5}} & \land~\textcolor{neg}{s=0} & \\
 (74) & \textcolor{pos}{a_{1} = v_{1}} & \land~\textcolor{pos}{a_{2} = v_{2}} & \land~\textcolor{neg}{a_{3} \neq v_{3}} & \land~\textcolor{pos}{a_{4} = v_{4}} & \land~\textcolor{pos}{a_{5} = v_{5}} & \land~\textcolor{pos}{s=1} & \\
 (75) & \textcolor{pos}{a_{1} = v_{1}} & \land~\textcolor{pos}{a_{2} = v_{2}} & \land~\textcolor{pos}{a_{3} = v_{3}} & \land~\textcolor{pos}{a_{4} = v_{4}} & \land~\textcolor{pos}{a_{5} = v_{5}} & \land~\textcolor{pos}{s=1} & \mathbf{(\times 6)}\\
\end{array}}
\]
\end{framed}

\begin{framed}
\noindent\textbf{Example of solution found against SimpleQBS($\tau=4,\sigma=0$) in the AUXILIARY scenario.}
\[
{\customsize
\begin{array}{llllllll}
 \multicolumn{3}{l}{\textbf{Difference queries}} \\
 (1) & \textcolor{pos}{a_{1} = v_{1}} & \land~ & \land~ & \land~ & \land~\textcolor{pos}{a_{5} = v_{5}} & \land~\textcolor{neg}{s=0} & \\
 (2) & \textcolor{pos}{a_{1} = v_{1}} & \land~\textcolor{neg}{a_{2} \neq v_{2}} & \land~ & \land~ & \land~\textcolor{pos}{a_{5} = v_{5}} & \land~\textcolor{neg}{s=0} & \mathbf{(\times 2)}\\
 (3) & \textcolor{pos}{a_{1} = v_{1}} & \land~ & \land~ & \land~ & \land~\textcolor{pos}{a_{5} = v_{5}} & \land~\textcolor{pos}{s=1} & \\
 (4) & \textcolor{pos}{a_{1} = v_{1}} & \land~\textcolor{neg}{a_{2} \neq v_{2}} & \land~ & \land~ & \land~\textcolor{pos}{a_{5} = v_{5}} & \land~\textcolor{pos}{s=1} & \\
 (5) &  & \land~ & \land~\textcolor{pos}{a_{3} = v_{3}} & \land~\textcolor{pos}{a_{4} = v_{4}} & \land~\textcolor{pos}{a_{5} = v_{5}} & \land~\textcolor{neg}{s=0} & \\
 (6) & \textcolor{neg}{a_{1} \neq v_{1}} & \land~ & \land~\textcolor{pos}{a_{3} = v_{3}} & \land~\textcolor{pos}{a_{4} = v_{4}} & \land~\textcolor{pos}{a_{5} = v_{5}} & \land~\textcolor{neg}{s=0} & \\
 (7) &  & \land~\textcolor{neg}{a_{2} \neq v_{2}} & \land~\textcolor{pos}{a_{3} = v_{3}} & \land~\textcolor{pos}{a_{4} = v_{4}} & \land~\textcolor{pos}{a_{5} = v_{5}} & \land~\textcolor{neg}{s=0} & \\
 (8) &  & \land~\textcolor{pos}{a_{2} = v_{2}} & \land~ & \land~\textcolor{pos}{a_{4} = v_{4}} & \land~\textcolor{pos}{a_{5} = v_{5}} & \land~\textcolor{neg}{s=0} & \\
 (9) & \textcolor{neg}{a_{1} \neq v_{1}} & \land~\textcolor{pos}{a_{2} = v_{2}} & \land~ & \land~\textcolor{pos}{a_{4} = v_{4}} & \land~\textcolor{pos}{a_{5} = v_{5}} & \land~\textcolor{neg}{s=0} & \vspace{\customspace}\\
 \multicolumn{3}{l}{\textbf{Other queries}} \\
 (10) & \textcolor{neg}{a_{1} \neq v_{1}} & \land~\textcolor{neg}{a_{2} \neq v_{2}} & \land~ & \land~ & \land~ & \\
 (11) &  & \land~ & \land~\textcolor{neg}{a_{3} \neq v_{3}} & \land~\textcolor{neg}{a_{4} \neq v_{4}} & \land~ & \\
 (12) &  & \land~ & \land~\textcolor{neg}{a_{3} \neq v_{3}} & \land~ & \land~ & \land~\textcolor{pos}{s=1} & \\
 (13) &  & \land~ & \land~ & \land~\textcolor{neg}{a_{4} \neq v_{4}} & \land~\textcolor{neg}{a_{5} \neq v_{5}} & & \mathbf{(\times 2)}\\
 (14) &  & \land~ & \land~ & \land~\textcolor{neg}{a_{4} \neq v_{4}} & \land~\textcolor{pos}{a_{5} = v_{5}} & \\
 (15) & \textcolor{pos}{a_{1} = v_{1}} & \land~ & \land~ & \land~ & \land~\textcolor{neg}{a_{5} \neq v_{5}} & \\
 (16) & \textcolor{neg}{a_{1} \neq v_{1}} & \land~ & \land~ & \land~\textcolor{neg}{a_{4} \neq v_{4}} & \land~\textcolor{neg}{a_{5} \neq v_{5}} & \\
 (17) & \textcolor{neg}{a_{1} \neq v_{1}} & \land~ & \land~\textcolor{pos}{a_{3} = v_{3}} & \land~\textcolor{neg}{a_{4} \neq v_{4}} & \land~ & \\
 (18) & \textcolor{neg}{a_{1} \neq v_{1}} & \land~ & \land~\textcolor{pos}{a_{3} = v_{3}} & \land~\textcolor{pos}{a_{4} = v_{4}} & \land~ & \\
 (19) &  & \land~\textcolor{neg}{a_{2} \neq v_{2}} & \land~ & \land~\textcolor{neg}{a_{4} \neq v_{4}} & \land~\textcolor{neg}{a_{5} \neq v_{5}} & \\
 (20) &  & \land~\textcolor{neg}{a_{2} \neq v_{2}} & \land~\textcolor{pos}{a_{3} = v_{3}} & \land~ & \land~ & \land~\textcolor{neg}{s=0} & \\
 (21) &  & \land~ & \land~\textcolor{neg}{a_{3} \neq v_{3}} & \land~\textcolor{pos}{a_{4} = v_{4}} & \land~ & \land~\textcolor{pos}{s=1} & \\
 (22) &  & \land~\textcolor{pos}{a_{2} = v_{2}} & \land~\textcolor{neg}{a_{3} \neq v_{3}} & \land~\textcolor{pos}{a_{4} = v_{4}} & \land~ & \\
 (23) &  & \land~\textcolor{pos}{a_{2} = v_{2}} & \land~ & \land~\textcolor{pos}{a_{4} = v_{4}} & \land~ & \land~\textcolor{neg}{s=0} & \\
 (24) &  & \land~\textcolor{pos}{a_{2} = v_{2}} & \land~\textcolor{pos}{a_{3} = v_{3}} & \land~ & \land~ & \land~\textcolor{pos}{s=1} & \\
 (25) & \textcolor{pos}{a_{1} = v_{1}} & \land~\textcolor{neg}{a_{2} \neq v_{2}} & \land~\textcolor{neg}{a_{3} \neq v_{3}} & \land~ & \land~ & \\
 (26) & \textcolor{pos}{a_{1} = v_{1}} & \land~ & \land~\textcolor{neg}{a_{3} \neq v_{3}} & \land~\textcolor{pos}{a_{4} = v_{4}} & \land~ & \\
 (27) & \textcolor{pos}{a_{1} = v_{1}} & \land~ & \land~ & \land~\textcolor{neg}{a_{4} \neq v_{4}} & \land~\textcolor{neg}{a_{5} \neq v_{5}} & \\
 (28) & \textcolor{pos}{a_{1} = v_{1}} & \land~ & \land~ & \land~\textcolor{pos}{a_{4} = v_{4}} & \land~ & \land~\textcolor{neg}{s=0} & \\
 (29) & \textcolor{pos}{a_{1} = v_{1}} & \land~ & \land~ & \land~\textcolor{pos}{a_{4} = v_{4}} & \land~\textcolor{pos}{a_{5} = v_{5}} & \\
 (30) & \textcolor{neg}{a_{1} \neq v_{1}} & \land~\textcolor{neg}{a_{2} \neq v_{2}} & \land~\textcolor{neg}{a_{3} \neq v_{3}} & \land~\textcolor{pos}{a_{4} = v_{4}} & \land~ & \\
 (31) & \textcolor{neg}{a_{1} \neq v_{1}} & \land~\textcolor{neg}{a_{2} \neq v_{2}} & \land~\textcolor{pos}{a_{3} = v_{3}} & \land~ & \land~ & \land~\textcolor{pos}{s=1} & \\
 (32) & \textcolor{neg}{a_{1} \neq v_{1}} & \land~ & \land~\textcolor{neg}{a_{3} \neq v_{3}} & \land~ & \land~\textcolor{neg}{a_{5} \neq v_{5}} & \land~\textcolor{pos}{s=1} & \\
 (33) & \textcolor{neg}{a_{1} \neq v_{1}} & \land~ & \land~\textcolor{neg}{a_{3} \neq v_{3}} & \land~ & \land~\textcolor{pos}{a_{5} = v_{5}} & \land~\textcolor{pos}{s=1} & \\
 (34) & \textcolor{neg}{a_{1} \neq v_{1}} & \land~ & \land~\textcolor{neg}{a_{3} \neq v_{3}} & \land~\textcolor{pos}{a_{4} = v_{4}} & \land~ & \land~\textcolor{pos}{s=1} & \\
 (35) & \textcolor{neg}{a_{1} \neq v_{1}} & \land~ & \land~\textcolor{pos}{a_{3} = v_{3}} & \land~ & \land~\textcolor{neg}{a_{5} \neq v_{5}} & \land~\textcolor{neg}{s=0} & \\
 (36) & \textcolor{neg}{a_{1} \neq v_{1}} & \land~ & \land~\textcolor{pos}{a_{3} = v_{3}} & \land~ & \land~\textcolor{pos}{a_{5} = v_{5}} & \land~\textcolor{pos}{s=1} & \\
 (37) & \textcolor{neg}{a_{1} \neq v_{1}} & \land~ & \land~\textcolor{pos}{a_{3} = v_{3}} & \land~\textcolor{pos}{a_{4} = v_{4}} & \land~ & \land~\textcolor{pos}{s=1} & \\
 (38) & \textcolor{neg}{a_{1} \neq v_{1}} & \land~\textcolor{pos}{a_{2} = v_{2}} & \land~\textcolor{neg}{a_{3} \neq v_{3}} & \land~ & \land~\textcolor{pos}{a_{5} = v_{5}} & \\
 (39) & \textcolor{neg}{a_{1} \neq v_{1}} & \land~\textcolor{pos}{a_{2} = v_{2}} & \land~ & \land~\textcolor{neg}{a_{4} \neq v_{4}} & \land~\textcolor{neg}{a_{5} \neq v_{5}} & \\
 (40) & \textcolor{neg}{a_{1} \neq v_{1}} & \land~\textcolor{pos}{a_{2} = v_{2}} & \land~ & \land~\textcolor{pos}{a_{4} = v_{4}} & \land~ & \land~\textcolor{pos}{s=1} & \\
 (41) &  & \land~\textcolor{neg}{a_{2} \neq v_{2}} & \land~\textcolor{neg}{a_{3} \neq v_{3}} & \land~\textcolor{pos}{a_{4} = v_{4}} & \land~ & \land~\textcolor{pos}{s=1} & \\
 (42) &  & \land~\textcolor{neg}{a_{2} \neq v_{2}} & \land~\textcolor{pos}{a_{3} = v_{3}} & \land~\textcolor{pos}{a_{4} = v_{4}} & \land~ & \land~\textcolor{neg}{s=0} & \\
 (43) &  & \land~ & \land~\textcolor{neg}{a_{3} \neq v_{3}} & \land~\textcolor{neg}{a_{4} \neq v_{4}} & \land~\textcolor{neg}{a_{5} \neq v_{5}} & \land~\textcolor{neg}{s=0} & \mathbf{(\times 2)}\\
 (44) &  & \land~\textcolor{pos}{a_{2} = v_{2}} & \land~\textcolor{neg}{a_{3} \neq v_{3}} & \land~\textcolor{neg}{a_{4} \neq v_{4}} & \land~\textcolor{neg}{a_{5} \neq v_{5}} & \\
 (45) &  & \land~\textcolor{pos}{a_{2} = v_{2}} & \land~\textcolor{neg}{a_{3} \neq v_{3}} & \land~ & \land~\textcolor{neg}{a_{5} \neq v_{5}} & \land~\textcolor{pos}{s=1} & \\
 (46) &  & \land~\textcolor{pos}{a_{2} = v_{2}} & \land~\textcolor{neg}{a_{3} \neq v_{3}} & \land~\textcolor{pos}{a_{4} = v_{4}} & \land~ & \land~\textcolor{pos}{s=1} & \\
 (47) &  & \land~\textcolor{pos}{a_{2} = v_{2}} & \land~\textcolor{pos}{a_{3} = v_{3}} & \land~\textcolor{neg}{a_{4} \neq v_{4}} & \land~\textcolor{pos}{a_{5} = v_{5}} & \\
 (48) &  & \land~\textcolor{pos}{a_{2} = v_{2}} & \land~\textcolor{pos}{a_{3} = v_{3}} & \land~ & \land~\textcolor{pos}{a_{5} = v_{5}} & \land~\textcolor{neg}{s=0} & \\
 (49) & \textcolor{pos}{a_{1} = v_{1}} & \land~\textcolor{neg}{a_{2} \neq v_{2}} & \land~\textcolor{neg}{a_{3} \neq v_{3}} & \land~\textcolor{pos}{a_{4} = v_{4}} & \land~ & \\
 (50) & \textcolor{pos}{a_{1} = v_{1}} & \land~\textcolor{neg}{a_{2} \neq v_{2}} & \land~ & \land~\textcolor{pos}{a_{4} = v_{4}} & \land~ & \land~\textcolor{pos}{s=1} & \\
 (51) & \textcolor{pos}{a_{1} = v_{1}} & \land~\textcolor{neg}{a_{2} \neq v_{2}} & \land~ & \land~\textcolor{pos}{a_{4} = v_{4}} & \land~\textcolor{pos}{a_{5} = v_{5}} & \\
 (52) & \textcolor{pos}{a_{1} = v_{1}} & \land~\textcolor{neg}{a_{2} \neq v_{2}} & \land~\textcolor{pos}{a_{3} = v_{3}} & \land~ & \land~ & \land~\textcolor{neg}{s=0} & \\
 (53) & \textcolor{pos}{a_{1} = v_{1}} & \land~\textcolor{neg}{a_{2} \neq v_{2}} & \land~\textcolor{pos}{a_{3} = v_{3}} & \land~ & \land~ & \land~\textcolor{pos}{s=1} & \\
 (54) & \textcolor{pos}{a_{1} = v_{1}} & \land~ & \land~\textcolor{neg}{a_{3} \neq v_{3}} & \land~\textcolor{neg}{a_{4} \neq v_{4}} & \land~ & \land~\textcolor{neg}{s=0} & \\
 (55) & \textcolor{pos}{a_{1} = v_{1}} & \land~ & \land~\textcolor{neg}{a_{3} \neq v_{3}} & \land~\textcolor{neg}{a_{4} \neq v_{4}} & \land~\textcolor{pos}{a_{5} = v_{5}} & \\
 (56) & \textcolor{pos}{a_{1} = v_{1}} & \land~ & \land~\textcolor{pos}{a_{3} = v_{3}} & \land~ & \land~\textcolor{pos}{a_{5} = v_{5}} & \land~\textcolor{pos}{s=1} & \\
 (57) & \textcolor{pos}{a_{1} = v_{1}} & \land~ & \land~\textcolor{pos}{a_{3} = v_{3}} & \land~\textcolor{pos}{a_{4} = v_{4}} & \land~\textcolor{pos}{a_{5} = v_{5}} & \\
 (58) & \textcolor{pos}{a_{1} = v_{1}} & \land~\textcolor{pos}{a_{2} = v_{2}} & \land~\textcolor{neg}{a_{3} \neq v_{3}} & \land~ & \land~\textcolor{pos}{a_{5} = v_{5}} & \\
 (59) & \textcolor{pos}{a_{1} = v_{1}} & \land~\textcolor{pos}{a_{2} = v_{2}} & \land~ & \land~\textcolor{pos}{a_{4} = v_{4}} & \land~\textcolor{neg}{a_{5} \neq v_{5}} & \\
 (60) & \textcolor{neg}{a_{1} \neq v_{1}} & \land~\textcolor{neg}{a_{2} \neq v_{2}} & \land~\textcolor{neg}{a_{3} \neq v_{3}} & \land~\textcolor{neg}{a_{4} \neq v_{4}} & \land~\textcolor{pos}{a_{5} = v_{5}} & \\
 (61) & \textcolor{neg}{a_{1} \neq v_{1}} & \land~\textcolor{neg}{a_{2} \neq v_{2}} & \land~\textcolor{neg}{a_{3} \neq v_{3}} & \land~ & \land~\textcolor{neg}{a_{5} \neq v_{5}} & \land~\textcolor{neg}{s=0} & \\
 (62) & \textcolor{neg}{a_{1} \neq v_{1}} & \land~\textcolor{neg}{a_{2} \neq v_{2}} & \land~\textcolor{pos}{a_{3} = v_{3}} & \land~\textcolor{pos}{a_{4} = v_{4}} & \land~\textcolor{neg}{a_{5} \neq v_{5}} & \\
 (63) & \textcolor{neg}{a_{1} \neq v_{1}} & \land~\textcolor{neg}{a_{2} \neq v_{2}} & \land~\textcolor{pos}{a_{3} = v_{3}} & \land~\textcolor{pos}{a_{4} = v_{4}} & \land~ & \land~\textcolor{neg}{s=0} & \\
 (64) & \textcolor{neg}{a_{1} \neq v_{1}} & \land~ & \land~\textcolor{pos}{a_{3} = v_{3}} & \land~\textcolor{neg}{a_{4} \neq v_{4}} & \land~\textcolor{neg}{a_{5} \neq v_{5}} & \land~\textcolor{neg}{s=0} & \mathbf{(\times 2)}\\
 (65) & \textcolor{neg}{a_{1} \neq v_{1}} & \land~\textcolor{pos}{a_{2} = v_{2}} & \land~ & \land~\textcolor{neg}{a_{4} \neq v_{4}} & \land~\textcolor{pos}{a_{5} = v_{5}} & \land~\textcolor{neg}{s=0} & \\
 (66) & \textcolor{neg}{a_{1} \neq v_{1}} & \land~\textcolor{pos}{a_{2} = v_{2}} & \land~ & \land~\textcolor{neg}{a_{4} \neq v_{4}} & \land~\textcolor{pos}{a_{5} = v_{5}} & \land~\textcolor{pos}{s=1} & \\
 (67) & \textcolor{neg}{a_{1} \neq v_{1}} & \land~\textcolor{pos}{a_{2} = v_{2}} & \land~ & \land~\textcolor{pos}{a_{4} = v_{4}} & \land~\textcolor{neg}{a_{5} \neq v_{5}} & \land~\textcolor{neg}{s=0} & \mathbf{(\times 2)}\\
 (68) & \textcolor{neg}{a_{1} \neq v_{1}} & \land~\textcolor{pos}{a_{2} = v_{2}} & \land~\textcolor{pos}{a_{3} = v_{3}} & \land~ & \land~\textcolor{pos}{a_{5} = v_{5}} & \land~\textcolor{pos}{s=1} & \\
 (69) &  & \land~\textcolor{neg}{a_{2} \neq v_{2}} & \land~\textcolor{pos}{a_{3} = v_{3}} & \land~\textcolor{pos}{a_{4} = v_{4}} & \land~\textcolor{neg}{a_{5} \neq v_{5}} & \land~\textcolor{neg}{s=0} & \\
 (70) &  & \land~\textcolor{pos}{a_{2} = v_{2}} & \land~\textcolor{neg}{a_{3} \neq v_{3}} & \land~\textcolor{neg}{a_{4} \neq v_{4}} & \land~\textcolor{neg}{a_{5} \neq v_{5}} & \land~\textcolor{pos}{s=1} & \\
 (71) &  & \land~\textcolor{pos}{a_{2} = v_{2}} & \land~\textcolor{pos}{a_{3} = v_{3}} & \land~\textcolor{pos}{a_{4} = v_{4}} & \land~\textcolor{neg}{a_{5} \neq v_{5}} & \land~\textcolor{pos}{s=1} & \\
 (72) & \textcolor{pos}{a_{1} = v_{1}} & \land~\textcolor{neg}{a_{2} \neq v_{2}} & \land~ & \land~\textcolor{neg}{a_{4} \neq v_{4}} & \land~\textcolor{neg}{a_{5} \neq v_{5}} & \land~\textcolor{neg}{s=0} & \\
 (73) & \textcolor{pos}{a_{1} = v_{1}} & \land~\textcolor{neg}{a_{2} \neq v_{2}} & \land~ & \land~\textcolor{pos}{a_{4} = v_{4}} & \land~\textcolor{neg}{a_{5} \neq v_{5}} & \land~\textcolor{pos}{s=1} & \\
 (74) & \textcolor{pos}{a_{1} = v_{1}} & \land~ & \land~\textcolor{neg}{a_{3} \neq v_{3}} & \land~\textcolor{pos}{a_{4} = v_{4}} & \land~\textcolor{neg}{a_{5} \neq v_{5}} & \land~\textcolor{neg}{s=0} & \\
 (75) & \textcolor{pos}{a_{1} = v_{1}} & \land~ & \land~\textcolor{neg}{a_{3} \neq v_{3}} & \land~\textcolor{pos}{a_{4} = v_{4}} & \land~\textcolor{pos}{a_{5} = v_{5}} & \land~\textcolor{neg}{s=0} & \\
 (76) & \textcolor{pos}{a_{1} = v_{1}} & \land~ & \land~\textcolor{neg}{a_{3} \neq v_{3}} & \land~\textcolor{pos}{a_{4} = v_{4}} & \land~\textcolor{pos}{a_{5} = v_{5}} & \land~\textcolor{pos}{s=1} & \\
 (77) & \textcolor{pos}{a_{1} = v_{1}} & \land~ & \land~\textcolor{pos}{a_{3} = v_{3}} & \land~\textcolor{pos}{a_{4} = v_{4}} & \land~\textcolor{neg}{a_{5} \neq v_{5}} & \land~\textcolor{neg}{s=0} & \mathbf{(\times 3)}\\
 (78) & \textcolor{pos}{a_{1} = v_{1}} & \land~ & \land~\textcolor{pos}{a_{3} = v_{3}} & \land~\textcolor{pos}{a_{4} = v_{4}} & \land~\textcolor{neg}{a_{5} \neq v_{5}} & \land~\textcolor{pos}{s=1} & \\
 (79) & \textcolor{pos}{a_{1} = v_{1}} & \land~\textcolor{pos}{a_{2} = v_{2}} & \land~ & \land~\textcolor{neg}{a_{4} \neq v_{4}} & \land~\textcolor{neg}{a_{5} \neq v_{5}} & \land~\textcolor{pos}{s=1} & \\
 (80) & \textcolor{pos}{a_{1} = v_{1}} & \land~\textcolor{pos}{a_{2} = v_{2}} & \land~\textcolor{pos}{a_{3} = v_{3}} & \land~\textcolor{neg}{a_{4} \neq v_{4}} & \land~\textcolor{neg}{a_{5} \neq v_{5}} & \\
 (81) & \textcolor{pos}{a_{1} = v_{1}} & \land~\textcolor{pos}{a_{2} = v_{2}} & \land~\textcolor{pos}{a_{3} = v_{3}} & \land~ & \land~\textcolor{pos}{a_{5} = v_{5}} & \land~\textcolor{pos}{s=1} & \\
 (82) & \textcolor{pos}{a_{1} = v_{1}} & \land~\textcolor{pos}{a_{2} = v_{2}} & \land~\textcolor{pos}{a_{3} = v_{3}} & \land~\textcolor{pos}{a_{4} = v_{4}} & \land~\textcolor{pos}{a_{5} = v_{5}} & & \mathbf{(\times 2)}\\
 (83) & \textcolor{neg}{a_{1} \neq v_{1}} & \land~\textcolor{neg}{a_{2} \neq v_{2}} & \land~\textcolor{neg}{a_{3} \neq v_{3}} & \land~\textcolor{neg}{a_{4} \neq v_{4}} & \land~\textcolor{neg}{a_{5} \neq v_{5}} & \land~\textcolor{neg}{s=0} & \\
 (84) & \textcolor{neg}{a_{1} \neq v_{1}} & \land~\textcolor{neg}{a_{2} \neq v_{2}} & \land~\textcolor{neg}{a_{3} \neq v_{3}} & \land~\textcolor{pos}{a_{4} = v_{4}} & \land~\textcolor{neg}{a_{5} \neq v_{5}} & \land~\textcolor{pos}{s=1} & \\
 (85) & \textcolor{neg}{a_{1} \neq v_{1}} & \land~\textcolor{neg}{a_{2} \neq v_{2}} & \land~\textcolor{pos}{a_{3} = v_{3}} & \land~\textcolor{neg}{a_{4} \neq v_{4}} & \land~\textcolor{neg}{a_{5} \neq v_{5}} & \land~\textcolor{neg}{s=0} & \\
 (86) & \textcolor{neg}{a_{1} \neq v_{1}} & \land~\textcolor{pos}{a_{2} = v_{2}} & \land~\textcolor{neg}{a_{3} \neq v_{3}} & \land~\textcolor{pos}{a_{4} = v_{4}} & \land~\textcolor{pos}{a_{5} = v_{5}} & \land~\textcolor{pos}{s=1} & \\
 (87) & \textcolor{neg}{a_{1} \neq v_{1}} & \land~\textcolor{pos}{a_{2} = v_{2}} & \land~\textcolor{pos}{a_{3} = v_{3}} & \land~\textcolor{neg}{a_{4} \neq v_{4}} & \land~\textcolor{neg}{a_{5} \neq v_{5}} & \land~\textcolor{pos}{s=1} & \\
 (88) & \textcolor{pos}{a_{1} = v_{1}} & \land~\textcolor{neg}{a_{2} \neq v_{2}} & \land~\textcolor{pos}{a_{3} = v_{3}} & \land~\textcolor{neg}{a_{4} \neq v_{4}} & \land~\textcolor{neg}{a_{5} \neq v_{5}} & \land~\textcolor{neg}{s=0} & \\
 (89) & \textcolor{pos}{a_{1} = v_{1}} & \land~\textcolor{neg}{a_{2} \neq v_{2}} & \land~\textcolor{pos}{a_{3} = v_{3}} & \land~\textcolor{pos}{a_{4} = v_{4}} & \land~\textcolor{pos}{a_{5} = v_{5}} & \land~\textcolor{neg}{s=0} & \\
 (90) & \textcolor{pos}{a_{1} = v_{1}} & \land~\textcolor{pos}{a_{2} = v_{2}} & \land~\textcolor{neg}{a_{3} \neq v_{3}} & \land~\textcolor{neg}{a_{4} \neq v_{4}} & \land~\textcolor{neg}{a_{5} \neq v_{5}} & \land~\textcolor{pos}{s=1} & \\
 (91) & \textcolor{pos}{a_{1} = v_{1}} & \land~\textcolor{pos}{a_{2} = v_{2}} & \land~\textcolor{neg}{a_{3} \neq v_{3}} & \land~\textcolor{pos}{a_{4} = v_{4}} & \land~\textcolor{pos}{a_{5} = v_{5}} & \land~\textcolor{pos}{s=1} & \mathbf{(\times 2)}\\
\end{array}}
\]
\end{framed}

\begin{framed}
\noindent\textbf{Example of solution found against SimpleQBS($\tau=4, \sigma=3$) in the AUXILIARY scenario.}
\[
{\customsize
\begin{array}{llllllll}
 \multicolumn{3}{l}{\textbf{Difference queries}} \\
 (1) &  & \land~ & \land~ & \land~\textcolor{pos}{a_{4} = v_{4}} & \land~ & \land~\textcolor{pos}{s=1} & \\
 (2) &  & \land~ & \land~\textcolor{neg}{a_{3} \neq v_{3}} & \land~\textcolor{pos}{a_{4} = v_{4}} & \land~ & \land~\textcolor{pos}{s=1} & \\
 (3) &  & \land~\textcolor{pos}{a_{2} = v_{2}} & \land~ & \land~ & \land~ & \land~\textcolor{neg}{s=0} & \mathbf{(\times 2)}\\
 (4) &  & \land~\textcolor{pos}{a_{2} = v_{2}} & \land~\textcolor{neg}{a_{3} \neq v_{3}} & \land~ & \land~ & \land~\textcolor{neg}{s=0} & \\
 (5) &  & \land~\textcolor{pos}{a_{2} = v_{2}} & \land~ & \land~ & \land~\textcolor{neg}{a_{5} \neq v_{5}} & \land~\textcolor{neg}{s=0} & \\
 (6) &  & \land~\textcolor{pos}{a_{2} = v_{2}} & \land~ & \land~ & \land~ & \land~\textcolor{pos}{s=1} & \\
 (7) & \textcolor{neg}{a_{1} \neq v_{1}} & \land~\textcolor{pos}{a_{2} = v_{2}} & \land~ & \land~ & \land~ & \land~\textcolor{pos}{s=1} & \mathbf{(\times 2)}\\
 (8) & \textcolor{pos}{a_{1} = v_{1}} & \land~ & \land~ & \land~ & \land~\textcolor{pos}{a_{5} = v_{5}} & \land~\textcolor{pos}{s=1} & \mathbf{(\times 2)}\\
 (9) & \textcolor{pos}{a_{1} = v_{1}} & \land~ & \land~\textcolor{neg}{a_{3} \neq v_{3}} & \land~ & \land~\textcolor{pos}{a_{5} = v_{5}} & \land~\textcolor{pos}{s=1} & \mathbf{(\times 2)}\\
 (10) & \textcolor{pos}{a_{1} = v_{1}} & \land~\textcolor{pos}{a_{2} = v_{2}} & \land~ & \land~ & \land~ & \land~\textcolor{pos}{s=1} & \\
 (11) & \textcolor{pos}{a_{1} = v_{1}} & \land~\textcolor{pos}{a_{2} = v_{2}} & \land~\textcolor{neg}{a_{3} \neq v_{3}} & \land~ & \land~ & \land~\textcolor{pos}{s=1} & \mathbf{(\times 2)}\\
 (12) & \textcolor{pos}{a_{1} = v_{1}} & \land~\textcolor{pos}{a_{2} = v_{2}} & \land~ & \land~ & \land~\textcolor{pos}{a_{5} = v_{5}} & \land~\textcolor{neg}{s=0} & \mathbf{(\times 3)}\\
 (13) & \textcolor{pos}{a_{1} = v_{1}} & \land~\textcolor{pos}{a_{2} = v_{2}} & \land~ & \land~\textcolor{neg}{a_{4} \neq v_{4}} & \land~\textcolor{pos}{a_{5} = v_{5}} & \land~\textcolor{neg}{s=0} & \\
 (14) & \textcolor{pos}{a_{1} = v_{1}} & \land~\textcolor{pos}{a_{2} = v_{2}} & \land~ & \land~ & \land~\textcolor{pos}{a_{5} = v_{5}} & \land~\textcolor{pos}{s=1} & \mathbf{(\times 2)}\\
 (15) & \textcolor{pos}{a_{1} = v_{1}} & \land~\textcolor{pos}{a_{2} = v_{2}} & \land~\textcolor{neg}{a_{3} \neq v_{3}} & \land~ & \land~\textcolor{pos}{a_{5} = v_{5}} & \land~\textcolor{pos}{s=1} & \mathbf{(\times 2)}\\
 (16) & \textcolor{pos}{a_{1} = v_{1}} & \land~\textcolor{pos}{a_{2} = v_{2}} & \land~ & \land~\textcolor{pos}{a_{4} = v_{4}} & \land~ & \land~\textcolor{neg}{s=0} & \mathbf{(\times 4)}\\
 (17) & \textcolor{pos}{a_{1} = v_{1}} & \land~\textcolor{pos}{a_{2} = v_{2}} & \land~\textcolor{neg}{a_{3} \neq v_{3}} & \land~\textcolor{pos}{a_{4} = v_{4}} & \land~ & \land~\textcolor{neg}{s=0} & \mathbf{(\times 4)}\\
 (18) & \textcolor{pos}{a_{1} = v_{1}} & \land~\textcolor{pos}{a_{2} = v_{2}} & \land~ & \land~\textcolor{pos}{a_{4} = v_{4}} & \land~ & \land~\textcolor{pos}{s=1} & \mathbf{(\times 2)}\\
 (19) & \textcolor{pos}{a_{1} = v_{1}} & \land~\textcolor{pos}{a_{2} = v_{2}} & \land~\textcolor{neg}{a_{3} \neq v_{3}} & \land~\textcolor{pos}{a_{4} = v_{4}} & \land~ & \land~\textcolor{pos}{s=1} & \vspace{\customspace}\\
 \multicolumn{3}{l}{\textbf{Other queries}} \\
 (20) &  & \land~ & \land~ & \land~ & \land~ & \land~\textcolor{neg}{s=0} & \\
 (21) & \textcolor{neg}{a_{1} \neq v_{1}} & \land~ & \land~ & \land~ & \land~\textcolor{pos}{a_{5} = v_{5}} & \\
 (22) &  & \land~\textcolor{neg}{a_{2} \neq v_{2}} & \land~ & \land~ & \land~\textcolor{neg}{a_{5} \neq v_{5}} & \mathbf{(\times 2)}\\
 (23) &  & \land~ & \land~\textcolor{neg}{a_{3} \neq v_{3}} & \land~\textcolor{neg}{a_{4} \neq v_{4}} & \land~ & \\
 (24) &  & \land~ & \land~ & \land~ & \land~\textcolor{pos}{a_{5} = v_{5}} & \land~\textcolor{neg}{s=0} & \\
 (25) &  & \land~ & \land~ & \land~\textcolor{pos}{a_{4} = v_{4}} & \land~\textcolor{pos}{a_{5} = v_{5}} & \\
 (26) & \textcolor{neg}{a_{1} \neq v_{1}} & \land~\textcolor{neg}{a_{2} \neq v_{2}} & \land~ & \land~ & \land~ & \land~\textcolor{neg}{s=0} & \\
 (27) & \textcolor{neg}{a_{1} \neq v_{1}} & \land~ & \land~ & \land~ & \land~\textcolor{neg}{a_{5} \neq v_{5}} & \land~\textcolor{pos}{s=1} & \\
 (28) & \textcolor{neg}{a_{1} \neq v_{1}} & \land~ & \land~ & \land~\textcolor{pos}{a_{4} = v_{4}} & \land~\textcolor{pos}{a_{5} = v_{5}} & \\
 (29) &  & \land~ & \land~\textcolor{neg}{a_{3} \neq v_{3}} & \land~ & \land~\textcolor{neg}{a_{5} \neq v_{5}} & \land~\textcolor{pos}{s=1} & \\
 (30) &  & \land~ & \land~\textcolor{neg}{a_{3} \neq v_{3}} & \land~ & \land~\textcolor{pos}{a_{5} = v_{5}} & \land~\textcolor{pos}{s=1} & \mathbf{(\times 4)}\\
 (31) &  & \land~ & \land~ & \land~\textcolor{neg}{a_{4} \neq v_{4}} & \land~\textcolor{neg}{a_{5} \neq v_{5}} & \land~\textcolor{pos}{s=1} & \mathbf{(\times 2)}\\
 (32) &  & \land~ & \land~ & \land~\textcolor{pos}{a_{4} = v_{4}} & \land~\textcolor{pos}{a_{5} = v_{5}} & \land~\textcolor{neg}{s=0} & \mathbf{(\times 2)}\\
 (33) &  & \land~ & \land~\textcolor{pos}{a_{3} = v_{3}} & \land~\textcolor{neg}{a_{4} \neq v_{4}} & \land~ & \land~\textcolor{neg}{s=0} & \\
 (34) &  & \land~\textcolor{pos}{a_{2} = v_{2}} & \land~\textcolor{neg}{a_{3} \neq v_{3}} & \land~\textcolor{pos}{a_{4} = v_{4}} & \land~ & \\
 (35) &  & \land~\textcolor{pos}{a_{2} = v_{2}} & \land~ & \land~\textcolor{pos}{a_{4} = v_{4}} & \land~ & \land~\textcolor{pos}{s=1} & \\
 (36) &  & \land~\textcolor{pos}{a_{2} = v_{2}} & \land~\textcolor{pos}{a_{3} = v_{3}} & \land~\textcolor{neg}{a_{4} \neq v_{4}} & \land~ & \\
 (37) &  & \land~\textcolor{pos}{a_{2} = v_{2}} & \land~\textcolor{pos}{a_{3} = v_{3}} & \land~ & \land~ & \land~\textcolor{neg}{s=0} & \\
 (38) & \textcolor{pos}{a_{1} = v_{1}} & \land~\textcolor{neg}{a_{2} \neq v_{2}} & \land~ & \land~\textcolor{neg}{a_{4} \neq v_{4}} & \land~ & \\
 (39) & \textcolor{pos}{a_{1} = v_{1}} & \land~\textcolor{neg}{a_{2} \neq v_{2}} & \land~ & \land~ & \land~ & \land~\textcolor{neg}{s=0} & \mathbf{(\times 2)}\\
 (40) & \textcolor{pos}{a_{1} = v_{1}} & \land~ & \land~\textcolor{neg}{a_{3} \neq v_{3}} & \land~ & \land~ & \land~\textcolor{neg}{s=0} & \\
 (41) & \textcolor{pos}{a_{1} = v_{1}} & \land~ & \land~ & \land~\textcolor{pos}{a_{4} = v_{4}} & \land~ & \land~\textcolor{pos}{s=1} & \\
 (42) & \textcolor{pos}{a_{1} = v_{1}} & \land~ & \land~\textcolor{pos}{a_{3} = v_{3}} & \land~ & \land~\textcolor{pos}{a_{5} = v_{5}} & \\
 (43) & \textcolor{neg}{a_{1} \neq v_{1}} & \land~\textcolor{neg}{a_{2} \neq v_{2}} & \land~ & \land~ & \land~\textcolor{pos}{a_{5} = v_{5}} & \land~\textcolor{neg}{s=0} & \\
 (44) & \textcolor{neg}{a_{1} \neq v_{1}} & \land~ & \land~ & \land~\textcolor{neg}{a_{4} \neq v_{4}} & \land~\textcolor{pos}{a_{5} = v_{5}} & \land~\textcolor{pos}{s=1} & \\
 (45) & \textcolor{neg}{a_{1} \neq v_{1}} & \land~\textcolor{pos}{a_{2} = v_{2}} & \land~ & \land~\textcolor{neg}{a_{4} \neq v_{4}} & \land~ & \land~\textcolor{neg}{s=0} & \\
 (46) & \textcolor{neg}{a_{1} \neq v_{1}} & \land~\textcolor{pos}{a_{2} = v_{2}} & \land~ & \land~\textcolor{pos}{a_{4} = v_{4}} & \land~\textcolor{neg}{a_{5} \neq v_{5}} & \\
 (47) & \textcolor{neg}{a_{1} \neq v_{1}} & \land~\textcolor{pos}{a_{2} = v_{2}} & \land~ & \land~\textcolor{pos}{a_{4} = v_{4}} & \land~ & \land~\textcolor{neg}{s=0} & \\
 (48) &  & \land~\textcolor{neg}{a_{2} \neq v_{2}} & \land~\textcolor{neg}{a_{3} \neq v_{3}} & \land~\textcolor{neg}{a_{4} \neq v_{4}} & \land~ & \land~\textcolor{neg}{s=0} & \\
 (49) &  & \land~\textcolor{pos}{a_{2} = v_{2}} & \land~\textcolor{neg}{a_{3} \neq v_{3}} & \land~ & \land~\textcolor{pos}{a_{5} = v_{5}} & \land~\textcolor{neg}{s=0} & \mathbf{(\times 3)}\\
 (50) &  & \land~\textcolor{pos}{a_{2} = v_{2}} & \land~\textcolor{neg}{a_{3} \neq v_{3}} & \land~\textcolor{pos}{a_{4} = v_{4}} & \land~ & \land~\textcolor{neg}{s=0} & \mathbf{(\times 2)}\\
 (51) & \textcolor{pos}{a_{1} = v_{1}} & \land~\textcolor{neg}{a_{2} \neq v_{2}} & \land~\textcolor{neg}{a_{3} \neq v_{3}} & \land~ & \land~ & \land~\textcolor{pos}{s=1} & \mathbf{(\times 2)}\\
 (52) & \textcolor{pos}{a_{1} = v_{1}} & \land~\textcolor{neg}{a_{2} \neq v_{2}} & \land~ & \land~ & \land~\textcolor{pos}{a_{5} = v_{5}} & \land~\textcolor{neg}{s=0} & \mathbf{(\times 2)}\\
 (53) & \textcolor{pos}{a_{1} = v_{1}} & \land~ & \land~\textcolor{pos}{a_{3} = v_{3}} & \land~ & \land~\textcolor{pos}{a_{5} = v_{5}} & \land~\textcolor{neg}{s=0} & \\
 (54) & \textcolor{pos}{a_{1} = v_{1}} & \land~\textcolor{pos}{a_{2} = v_{2}} & \land~ & \land~ & \land~\textcolor{neg}{a_{5} \neq v_{5}} & \land~\textcolor{neg}{s=0} & \\
 (55) & \textcolor{neg}{a_{1} \neq v_{1}} & \land~\textcolor{neg}{a_{2} \neq v_{2}} & \land~ & \land~\textcolor{neg}{a_{4} \neq v_{4}} & \land~\textcolor{neg}{a_{5} \neq v_{5}} & \land~\textcolor{pos}{s=1} & \mathbf{(\times 2)}\\
 (56) & \textcolor{neg}{a_{1} \neq v_{1}} & \land~ & \land~\textcolor{neg}{a_{3} \neq v_{3}} & \land~\textcolor{neg}{a_{4} \neq v_{4}} & \land~\textcolor{pos}{a_{5} = v_{5}} & \land~\textcolor{neg}{s=0} & \mathbf{(\times 3)}\\
 (57) & \textcolor{neg}{a_{1} \neq v_{1}} & \land~\textcolor{pos}{a_{2} = v_{2}} & \land~\textcolor{neg}{a_{3} \neq v_{3}} & \land~\textcolor{pos}{a_{4} = v_{4}} & \land~\textcolor{neg}{a_{5} \neq v_{5}} & \\
 (58) & \textcolor{pos}{a_{1} = v_{1}} & \land~\textcolor{neg}{a_{2} \neq v_{2}} & \land~\textcolor{neg}{a_{3} \neq v_{3}} & \land~ & \land~\textcolor{neg}{a_{5} \neq v_{5}} & \land~\textcolor{pos}{s=1} & \\
 (59) & \textcolor{pos}{a_{1} = v_{1}} & \land~\textcolor{neg}{a_{2} \neq v_{2}} & \land~\textcolor{neg}{a_{3} \neq v_{3}} & \land~\textcolor{pos}{a_{4} = v_{4}} & \land~ & \land~\textcolor{pos}{s=1} & \\
 (60) & \textcolor{pos}{a_{1} = v_{1}} & \land~ & \land~\textcolor{neg}{a_{3} \neq v_{3}} & \land~\textcolor{pos}{a_{4} = v_{4}} & \land~\textcolor{pos}{a_{5} = v_{5}} & \land~\textcolor{neg}{s=0} & \\
 (61) & \textcolor{pos}{a_{1} = v_{1}} & \land~\textcolor{pos}{a_{2} = v_{2}} & \land~ & \land~\textcolor{pos}{a_{4} = v_{4}} & \land~\textcolor{pos}{a_{5} = v_{5}} & \land~\textcolor{pos}{s=1} & \\
 (62) & \textcolor{pos}{a_{1} = v_{1}} & \land~\textcolor{pos}{a_{2} = v_{2}} & \land~\textcolor{pos}{a_{3} = v_{3}} & \land~ & \land~\textcolor{pos}{a_{5} = v_{5}} & \land~\textcolor{neg}{s=0} & \mathbf{(\times 2)}\\
 (63) & \textcolor{neg}{a_{1} \neq v_{1}} & \land~\textcolor{neg}{a_{2} \neq v_{2}} & \land~\textcolor{neg}{a_{3} \neq v_{3}} & \land~\textcolor{pos}{a_{4} = v_{4}} & \land~\textcolor{pos}{a_{5} = v_{5}} & \land~\textcolor{neg}{s=0} & \\
 (64) & \textcolor{neg}{a_{1} \neq v_{1}} & \land~\textcolor{neg}{a_{2} \neq v_{2}} & \land~\textcolor{pos}{a_{3} = v_{3}} & \land~\textcolor{pos}{a_{4} = v_{4}} & \land~\textcolor{neg}{a_{5} \neq v_{5}} & \land~\textcolor{neg}{s=0} & \\
 (65) & \textcolor{neg}{a_{1} \neq v_{1}} & \land~\textcolor{pos}{a_{2} = v_{2}} & \land~\textcolor{neg}{a_{3} \neq v_{3}} & \land~\textcolor{pos}{a_{4} = v_{4}} & \land~\textcolor{pos}{a_{5} = v_{5}} & \land~\textcolor{neg}{s=0} & \\
 (66) & \textcolor{pos}{a_{1} = v_{1}} & \land~\textcolor{neg}{a_{2} \neq v_{2}} & \land~\textcolor{neg}{a_{3} \neq v_{3}} & \land~\textcolor{neg}{a_{4} \neq v_{4}} & \land~\textcolor{pos}{a_{5} = v_{5}} & \land~\textcolor{neg}{s=0} & \\
 (67) & \textcolor{pos}{a_{1} = v_{1}} & \land~\textcolor{pos}{a_{2} = v_{2}} & \land~\textcolor{neg}{a_{3} \neq v_{3}} & \land~\textcolor{neg}{a_{4} \neq v_{4}} & \land~\textcolor{neg}{a_{5} \neq v_{5}} & \land~\textcolor{neg}{s=0} & \\
 (68) & \textcolor{pos}{a_{1} = v_{1}} & \land~\textcolor{pos}{a_{2} = v_{2}} & \land~\textcolor{neg}{a_{3} \neq v_{3}} & \land~\textcolor{pos}{a_{4} = v_{4}} & \land~\textcolor{neg}{a_{5} \neq v_{5}} & \land~\textcolor{pos}{s=1} & \\
\end{array}}
\]
\end{framed}

\begin{framed}
\noindent\textbf{Example of solution found against SimpleQBS($\tau=3, \sigma=4$) in the AUXILIARY scenario.}
\[
{\customsize
\begin{array}{llllllll}
 \multicolumn{3}{l}{\textbf{Difference queries}} \\
 (1) &  & \land~ & \land~ & \land~ & \land~\textcolor{pos}{a_{5} = v_{5}} & \land~\textcolor{pos}{s=1} & \mathbf{(\times 4)}\\
 (2) &  & \land~ & \land~\textcolor{neg}{a_{3} \neq v_{3}} & \land~ & \land~\textcolor{pos}{a_{5} = v_{5}} & \land~\textcolor{pos}{s=1} & \\
 (3) &  & \land~ & \land~\textcolor{pos}{a_{3} = v_{3}} & \land~ & \land~\textcolor{pos}{a_{5} = v_{5}} & \land~\textcolor{neg}{s=0} & \mathbf{(\times 2)}\\
 (4) & \textcolor{neg}{a_{1} \neq v_{1}} & \land~ & \land~\textcolor{pos}{a_{3} = v_{3}} & \land~ & \land~\textcolor{pos}{a_{5} = v_{5}} & \land~\textcolor{neg}{s=0} & \mathbf{(\times 2)}\\
 (5) &  & \land~ & \land~\textcolor{pos}{a_{3} = v_{3}} & \land~\textcolor{pos}{a_{4} = v_{4}} & \land~\textcolor{pos}{a_{5} = v_{5}} & \land~\textcolor{neg}{s=0} & \mathbf{(\times 2)}\\
 (6) & \textcolor{neg}{a_{1} \neq v_{1}} & \land~ & \land~\textcolor{pos}{a_{3} = v_{3}} & \land~\textcolor{pos}{a_{4} = v_{4}} & \land~\textcolor{pos}{a_{5} = v_{5}} & \land~\textcolor{neg}{s=0} & \mathbf{(\times 5)}\\
 (7) & \textcolor{pos}{a_{1} = v_{1}} & \land~\textcolor{pos}{a_{2} = v_{2}} & \land~ & \land~ & \land~\textcolor{pos}{a_{5} = v_{5}} & \land~\textcolor{pos}{s=1} & \\
 (8) & \textcolor{pos}{a_{1} = v_{1}} & \land~\textcolor{pos}{a_{2} = v_{2}} & \land~\textcolor{neg}{a_{3} \neq v_{3}} & \land~ & \land~\textcolor{pos}{a_{5} = v_{5}} & \land~\textcolor{pos}{s=1} & \\
 (9) & \textcolor{pos}{a_{1} = v_{1}} & \land~\textcolor{pos}{a_{2} = v_{2}} & \land~ & \land~\textcolor{neg}{a_{4} \neq v_{4}} & \land~\textcolor{pos}{a_{5} = v_{5}} & \land~\textcolor{pos}{s=1} & \\
 (10) & \textcolor{pos}{a_{1} = v_{1}} & \land~\textcolor{pos}{a_{2} = v_{2}} & \land~ & \land~\textcolor{pos}{a_{4} = v_{4}} & \land~ & \land~\textcolor{pos}{s=1} & \\
 (11) & \textcolor{pos}{a_{1} = v_{1}} & \land~\textcolor{pos}{a_{2} = v_{2}} & \land~\textcolor{neg}{a_{3} \neq v_{3}} & \land~\textcolor{pos}{a_{4} = v_{4}} & \land~ & \land~\textcolor{pos}{s=1} & \\
 (12) &  & \land~\textcolor{pos}{a_{2} = v_{2}} & \land~\textcolor{pos}{a_{3} = v_{3}} & \land~\textcolor{pos}{a_{4} = v_{4}} & \land~\textcolor{pos}{a_{5} = v_{5}} & \land~\textcolor{pos}{s=1} & \mathbf{(\times 7)}\\
 (13) & \textcolor{neg}{a_{1} \neq v_{1}} & \land~\textcolor{pos}{a_{2} = v_{2}} & \land~\textcolor{pos}{a_{3} = v_{3}} & \land~\textcolor{pos}{a_{4} = v_{4}} & \land~\textcolor{pos}{a_{5} = v_{5}} & \land~\textcolor{pos}{s=1} & \mathbf{(\times 5)}\vspace{\customspace}\\
 \multicolumn{3}{l}{\textbf{Other queries}} \\
 (14) &  & \land~\textcolor{neg}{a_{2} \neq v_{2}} & \land~ & \land~\textcolor{neg}{a_{4} \neq v_{4}} & \land~ & \\
 (15) &  & \land~\textcolor{neg}{a_{2} \neq v_{2}} & \land~ & \land~ & \land~\textcolor{neg}{a_{5} \neq v_{5}} & \\
 (16) &  & \land~\textcolor{neg}{a_{2} \neq v_{2}} & \land~ & \land~ & \land~ & \land~\textcolor{neg}{s=0} & \\
 (17) &  & \land~\textcolor{neg}{a_{2} \neq v_{2}} & \land~ & \land~ & \land~\textcolor{pos}{a_{5} = v_{5}} & \\
 (18) &  & \land~ & \land~\textcolor{neg}{a_{3} \neq v_{3}} & \land~ & \land~\textcolor{neg}{a_{5} \neq v_{5}} & \\
 (19) &  & \land~ & \land~ & \land~\textcolor{pos}{a_{4} = v_{4}} & \land~\textcolor{neg}{a_{5} \neq v_{5}} & \\
 (20) &  & \land~\textcolor{pos}{a_{2} = v_{2}} & \land~ & \land~ & \land~ & \land~\textcolor{neg}{s=0} & \\
 (21) & \textcolor{pos}{a_{1} = v_{1}} & \land~ & \land~\textcolor{pos}{a_{3} = v_{3}} & \land~ & \land~ & \\
 (22) & \textcolor{neg}{a_{1} \neq v_{1}} & \land~\textcolor{pos}{a_{2} = v_{2}} & \land~\textcolor{pos}{a_{3} = v_{3}} & \land~ & \land~ & \\
 (23) &  & \land~\textcolor{neg}{a_{2} \neq v_{2}} & \land~\textcolor{neg}{a_{3} \neq v_{3}} & \land~ & \land~ & \land~\textcolor{neg}{s=0} & \\
 (24) &  & \land~\textcolor{neg}{a_{2} \neq v_{2}} & \land~\textcolor{neg}{a_{3} \neq v_{3}} & \land~\textcolor{pos}{a_{4} = v_{4}} & \land~ & \\
 (25) &  & \land~\textcolor{neg}{a_{2} \neq v_{2}} & \land~ & \land~\textcolor{neg}{a_{4} \neq v_{4}} & \land~\textcolor{pos}{a_{5} = v_{5}} & \\
 (26) &  & \land~\textcolor{neg}{a_{2} \neq v_{2}} & \land~ & \land~\textcolor{pos}{a_{4} = v_{4}} & \land~ & \land~\textcolor{pos}{s=1} & \\
 (27) &  & \land~ & \land~\textcolor{neg}{a_{3} \neq v_{3}} & \land~\textcolor{neg}{a_{4} \neq v_{4}} & \land~\textcolor{neg}{a_{5} \neq v_{5}} & \\
 (28) &  & \land~ & \land~\textcolor{neg}{a_{3} \neq v_{3}} & \land~\textcolor{pos}{a_{4} = v_{4}} & \land~ & \land~\textcolor{pos}{s=1} & \\
 (29) &  & \land~ & \land~ & \land~\textcolor{neg}{a_{4} \neq v_{4}} & \land~\textcolor{neg}{a_{5} \neq v_{5}} & \land~\textcolor{pos}{s=1} & \\
 (30) &  & \land~ & \land~ & \land~\textcolor{pos}{a_{4} = v_{4}} & \land~\textcolor{pos}{a_{5} = v_{5}} & \land~\textcolor{pos}{s=1} & \\
 (31) & \textcolor{pos}{a_{1} = v_{1}} & \land~\textcolor{neg}{a_{2} \neq v_{2}} & \land~ & \land~ & \land~ & \land~\textcolor{pos}{s=1} & \mathbf{(\times 4)}\\
 (32) & \textcolor{pos}{a_{1} = v_{1}} & \land~\textcolor{pos}{a_{2} = v_{2}} & \land~ & \land~ & \land~ & \land~\textcolor{pos}{s=1} & \\
 (33) & \textcolor{neg}{a_{1} \neq v_{1}} & \land~\textcolor{neg}{a_{2} \neq v_{2}} & \land~\textcolor{pos}{a_{3} = v_{3}} & \land~ & \land~\textcolor{neg}{a_{5} \neq v_{5}} & \\
 (34) & \textcolor{neg}{a_{1} \neq v_{1}} & \land~ & \land~ & \land~\textcolor{neg}{a_{4} \neq v_{4}} & \land~\textcolor{pos}{a_{5} = v_{5}} & \land~\textcolor{pos}{s=1} & \\
 (35) & \textcolor{neg}{a_{1} \neq v_{1}} & \land~ & \land~ & \land~\textcolor{pos}{a_{4} = v_{4}} & \land~\textcolor{neg}{a_{5} \neq v_{5}} & \land~\textcolor{pos}{s=1} & \\
 (36) & \textcolor{neg}{a_{1} \neq v_{1}} & \land~ & \land~\textcolor{pos}{a_{3} = v_{3}} & \land~ & \land~\textcolor{pos}{a_{5} = v_{5}} & \land~\textcolor{pos}{s=1} & \\
 (37) & \textcolor{neg}{a_{1} \neq v_{1}} & \land~ & \land~\textcolor{pos}{a_{3} = v_{3}} & \land~\textcolor{pos}{a_{4} = v_{4}} & \land~ & \land~\textcolor{neg}{s=0} & \\
 (38) & \textcolor{neg}{a_{1} \neq v_{1}} & \land~\textcolor{pos}{a_{2} = v_{2}} & \land~ & \land~ & \land~\textcolor{pos}{a_{5} = v_{5}} & \land~\textcolor{neg}{s=0} & \\
 (39) &  & \land~\textcolor{neg}{a_{2} \neq v_{2}} & \land~\textcolor{neg}{a_{3} \neq v_{3}} & \land~\textcolor{pos}{a_{4} = v_{4}} & \land~\textcolor{neg}{a_{5} \neq v_{5}} & \\
 (40) &  & \land~\textcolor{neg}{a_{2} \neq v_{2}} & \land~ & \land~\textcolor{neg}{a_{4} \neq v_{4}} & \land~\textcolor{neg}{a_{5} \neq v_{5}} & \land~\textcolor{pos}{s=1} & \\
 (41) &  & \land~\textcolor{neg}{a_{2} \neq v_{2}} & \land~ & \land~\textcolor{neg}{a_{4} \neq v_{4}} & \land~\textcolor{pos}{a_{5} = v_{5}} & \land~\textcolor{pos}{s=1} & \\
 (42) &  & \land~\textcolor{neg}{a_{2} \neq v_{2}} & \land~\textcolor{pos}{a_{3} = v_{3}} & \land~ & \land~\textcolor{neg}{a_{5} \neq v_{5}} & \land~\textcolor{neg}{s=0} & \\
 (43) &  & \land~ & \land~\textcolor{neg}{a_{3} \neq v_{3}} & \land~\textcolor{neg}{a_{4} \neq v_{4}} & \land~\textcolor{neg}{a_{5} \neq v_{5}} & \land~\textcolor{pos}{s=1} & \\
 (44) &  & \land~ & \land~\textcolor{neg}{a_{3} \neq v_{3}} & \land~\textcolor{pos}{a_{4} = v_{4}} & \land~\textcolor{neg}{a_{5} \neq v_{5}} & \land~\textcolor{pos}{s=1} & \mathbf{(\times 2)}\\
 (45) &  & \land~\textcolor{pos}{a_{2} = v_{2}} & \land~\textcolor{neg}{a_{3} \neq v_{3}} & \land~ & \land~\textcolor{neg}{a_{5} \neq v_{5}} & \land~\textcolor{pos}{s=1} & \\
 (46) &  & \land~\textcolor{pos}{a_{2} = v_{2}} & \land~\textcolor{neg}{a_{3} \neq v_{3}} & \land~\textcolor{pos}{a_{4} = v_{4}} & \land~\textcolor{pos}{a_{5} = v_{5}} & \\
 (47) &  & \land~\textcolor{pos}{a_{2} = v_{2}} & \land~\textcolor{pos}{a_{3} = v_{3}} & \land~ & \land~\textcolor{pos}{a_{5} = v_{5}} & \land~\textcolor{pos}{s=1} & \\
 (48) &  & \land~\textcolor{pos}{a_{2} = v_{2}} & \land~\textcolor{pos}{a_{3} = v_{3}} & \land~\textcolor{pos}{a_{4} = v_{4}} & \land~ & \land~\textcolor{pos}{s=1} & \mathbf{(\times 2)}\\
 (49) &  & \land~\textcolor{pos}{a_{2} = v_{2}} & \land~\textcolor{pos}{a_{3} = v_{3}} & \land~\textcolor{pos}{a_{4} = v_{4}} & \land~\textcolor{pos}{a_{5} = v_{5}} & \\
 (50) & \textcolor{pos}{a_{1} = v_{1}} & \land~\textcolor{neg}{a_{2} \neq v_{2}} & \land~ & \land~ & \land~\textcolor{pos}{a_{5} = v_{5}} & \land~\textcolor{pos}{s=1} & \\
 (51) & \textcolor{pos}{a_{1} = v_{1}} & \land~\textcolor{neg}{a_{2} \neq v_{2}} & \land~\textcolor{pos}{a_{3} = v_{3}} & \land~ & \land~\textcolor{pos}{a_{5} = v_{5}} & & \mathbf{(\times 2)}\\
 (52) & \textcolor{pos}{a_{1} = v_{1}} & \land~ & \land~ & \land~\textcolor{pos}{a_{4} = v_{4}} & \land~\textcolor{pos}{a_{5} = v_{5}} & \land~\textcolor{neg}{s=0} & \mathbf{(\times 2)}\\
 (53) & \textcolor{neg}{a_{1} \neq v_{1}} & \land~ & \land~\textcolor{pos}{a_{3} = v_{3}} & \land~\textcolor{neg}{a_{4} \neq v_{4}} & \land~\textcolor{pos}{a_{5} = v_{5}} & \land~\textcolor{neg}{s=0} & \\
 (54) & \textcolor{neg}{a_{1} \neq v_{1}} & \land~\textcolor{pos}{a_{2} = v_{2}} & \land~\textcolor{pos}{a_{3} = v_{3}} & \land~ & \land~\textcolor{pos}{a_{5} = v_{5}} & \land~\textcolor{neg}{s=0} & \\
 (55) &  & \land~\textcolor{neg}{a_{2} \neq v_{2}} & \land~\textcolor{neg}{a_{3} \neq v_{3}} & \land~\textcolor{neg}{a_{4} \neq v_{4}} & \land~\textcolor{pos}{a_{5} = v_{5}} & \land~\textcolor{pos}{s=1} & \mathbf{(\times 3)}\\
 (56) &  & \land~\textcolor{neg}{a_{2} \neq v_{2}} & \land~\textcolor{neg}{a_{3} \neq v_{3}} & \land~\textcolor{pos}{a_{4} = v_{4}} & \land~\textcolor{neg}{a_{5} \neq v_{5}} & \land~\textcolor{neg}{s=0} & \\
 (57) &  & \land~\textcolor{neg}{a_{2} \neq v_{2}} & \land~\textcolor{neg}{a_{3} \neq v_{3}} & \land~\textcolor{pos}{a_{4} = v_{4}} & \land~\textcolor{neg}{a_{5} \neq v_{5}} & \land~\textcolor{pos}{s=1} & \\
 (58) &  & \land~\textcolor{pos}{a_{2} = v_{2}} & \land~\textcolor{neg}{a_{3} \neq v_{3}} & \land~\textcolor{neg}{a_{4} \neq v_{4}} & \land~\textcolor{pos}{a_{5} = v_{5}} & \land~\textcolor{neg}{s=0} & \\
 (59) &  & \land~\textcolor{pos}{a_{2} = v_{2}} & \land~\textcolor{pos}{a_{3} = v_{3}} & \land~\textcolor{pos}{a_{4} = v_{4}} & \land~\textcolor{neg}{a_{5} \neq v_{5}} & \land~\textcolor{neg}{s=0} & \\
 (60) &  & \land~\textcolor{pos}{a_{2} = v_{2}} & \land~\textcolor{pos}{a_{3} = v_{3}} & \land~\textcolor{pos}{a_{4} = v_{4}} & \land~\textcolor{pos}{a_{5} = v_{5}} & \land~\textcolor{neg}{s=0} & \mathbf{(\times 3)}\\
 (61) & \textcolor{pos}{a_{1} = v_{1}} & \land~\textcolor{neg}{a_{2} \neq v_{2}} & \land~\textcolor{pos}{a_{3} = v_{3}} & \land~\textcolor{pos}{a_{4} = v_{4}} & \land~\textcolor{pos}{a_{5} = v_{5}} & \\
 (62) & \textcolor{pos}{a_{1} = v_{1}} & \land~\textcolor{pos}{a_{2} = v_{2}} & \land~\textcolor{pos}{a_{3} = v_{3}} & \land~\textcolor{neg}{a_{4} \neq v_{4}} & \land~ & \land~\textcolor{pos}{s=1} & \\
 (63) & \textcolor{pos}{a_{1} = v_{1}} & \land~\textcolor{pos}{a_{2} = v_{2}} & \land~\textcolor{pos}{a_{3} = v_{3}} & \land~\textcolor{pos}{a_{4} = v_{4}} & \land~\textcolor{pos}{a_{5} = v_{5}} & \\
 (64) & \textcolor{pos}{a_{1} = v_{1}} & \land~\textcolor{neg}{a_{2} \neq v_{2}} & \land~\textcolor{neg}{a_{3} \neq v_{3}} & \land~\textcolor{pos}{a_{4} = v_{4}} & \land~\textcolor{neg}{a_{5} \neq v_{5}} & \land~\textcolor{neg}{s=0} & \mathbf{(\times 2)}\\
 (65) & \textcolor{pos}{a_{1} = v_{1}} & \land~\textcolor{pos}{a_{2} = v_{2}} & \land~\textcolor{neg}{a_{3} \neq v_{3}} & \land~\textcolor{pos}{a_{4} = v_{4}} & \land~\textcolor{neg}{a_{5} \neq v_{5}} & \land~\textcolor{neg}{s=0} & \\
 (66) & \textcolor{pos}{a_{1} = v_{1}} & \land~\textcolor{pos}{a_{2} = v_{2}} & \land~\textcolor{neg}{a_{3} \neq v_{3}} & \land~\textcolor{pos}{a_{4} = v_{4}} & \land~\textcolor{pos}{a_{5} = v_{5}} & \land~\textcolor{neg}{s=0} & \\
 (67) & \textcolor{pos}{a_{1} = v_{1}} & \land~\textcolor{pos}{a_{2} = v_{2}} & \land~\textcolor{pos}{a_{3} = v_{3}} & \land~\textcolor{pos}{a_{4} = v_{4}} & \land~\textcolor{neg}{a_{5} \neq v_{5}} & \land~\textcolor{neg}{s=0} & \\
 (68) & \textcolor{pos}{a_{1} = v_{1}} & \land~\textcolor{pos}{a_{2} = v_{2}} & \land~\textcolor{pos}{a_{3} = v_{3}} & \land~\textcolor{pos}{a_{4} = v_{4}} & \land~\textcolor{pos}{a_{5} = v_{5}} & \land~\textcolor{neg}{s=0} & \\
\end{array}}
\]
\end{framed}

\end{document}